\documentclass[a4paper]{article}
\usepackage[leqno]{amsmath}
\usepackage{amsfonts}
\usepackage{graphicx, yhmath, mathabx}
\usepackage{indentfirst}
\usepackage{extarrows}
\usepackage{float}
\usepackage{amsthm, mathrsfs}
\usepackage[titletoc,title]{appendix}
\numberwithin{equation}{section}

\usepackage[top=2.8cm,bottom=2.3cm,inner=2.3cm,outer=2.3cm,
bindingoffset=0.5cm,footskip=0.55cm]{geometry}
\usepackage{tikz}
\usetikzlibrary{decorations.pathreplacing}
\usetikzlibrary{arrows}
\usetikzlibrary{decorations.markings}
\newtheorem{theorem}{Theorem}[section]
\newtheorem{lemma}[theorem]{Lemma}
\newtheorem{proposition}[theorem]{Proposition}

\newtheorem{remark}[theorem]{Remark}
\definecolor{light-gray}{gray}{0.75}

\title{On the Open Question of The Tracy-Widom Distribution of $\beta$-Ensemble With $\beta=6$}
\author{Li YuQi}
\begin{document}
\maketitle

{\large \bf Abstract:}
We determine completely the Tracy-Widom distribution for Dyson’s $\beta$-ensemble with $\beta=6$.
The problem of the Tracy-Widom distribution of $\beta$-ensemble for general $\beta>0$ 
has been reduced to find out a bounded solution of  the Bloemendal-Vir\'{a}g equation  with a specified boundary.
Rumanov proposed a  Lax pair approach to solve the Bloemendal-Vir\'{a}g equation for even integer $\beta$.
He  also specially studied the $\beta=6$ case with his approach and found a second order nonlinear ordinary
differential equation (ODE) for the logarithmic derivative of the Tracy–Widom distribution  for $\beta=6$.
Grava et al. continued to study $\beta=6$ 
and found  Rumanov's Lax pair is gauge equivalent to that of Painlev\'{e} II in this case.
They started with Rumanov's basic idea and came down to  two auxiliary functions $\alpha(t)$ and $q_2(t)$, 
which satisfy a coupled first-order ODE.
The open question by  Grava et al.  asks whether a global smooth solution of the  ODE with boundary condition 
$\alpha(\infty)=0$ and $q_2(\infty)=-1$ exists.
By studying the linear equation that is associated with $q_2$ and $\alpha$,
we  give a positive answer to the open question.
Moreover, we find that the solutions of the ODE with  $\alpha(\infty)=0$ and $q_2(\infty)=-1$
are parameterized by  $c_1$ and $c_2$. 
Not all $c_1$ and $c_2$ give global smooth solutions. 
But if $(c_1, c_2) \in R_{smooth}$,  where $R_{smooth}$ is  a large region containing $(0,0)$, they do give.
We prove the constructed solution is a bounded solution of 
the Bloemendal-Vir\'{a}g equation with the required boundary condition  if and only if  $(c_1, c_2)=(0,0)$.

\section{Introduction}

In the one dimensional case, the interaction energy of two point charges is
$$k_e \ln (\frac{r_0}{|x_A-x_B|}), $$
where $k_e$ is the electric force constant, $r_0$ is the  distance that the interaction energy is $0$,
$x_A$ and $x_B$ are the positions of the two point charges.
Dyson's Coulomb gas model is  $N$  particles with like charges, i.e., $k_e>0$,  in an external field $V=V(x)$. 
By the canonical ensemble, the probability 
that the first particle is in $[x_1,x_1+dx_1]$, $\cdots$,
and that $N$-th particle is in $[x_N, x_N+d x_N]$, is
\begin{eqnarray*}
p(x_1,x_2, \cdots, x_N) dx_1 dx_2\cdots dx_N
&=&\frac{1}{Z_N} e^{-\frac{1}{k_B T} \left( \sum \limits_{1 \le i<j\le N}  k_e \ln \left(\frac{r_0}{|x_i-x_j|} \right)+ 
\sum \limits_{j=1}^N  V(x_j) \right)} dx_1  dx_2 \cdots dx_N \\
&=&\frac{1}{Z_N} \left( \prod_{1 \le i <j \le N} \left| \frac{x_i-x_j}{r_0} \right|^{ \frac{k_e}{k_B T}} \right) 
e^{-\frac{1}{k_B T} \sum\limits_{j=1}^N V(x_j)}  
dx_1 dx_2 \cdots dx_N,
\end{eqnarray*}
where $k_B$ is the Boltzmann constant, $T$ is the temperature, and $Z_N$ is the normalization constant.
Here we assume $V$ is Gaussian, i.e., 
\begin{eqnarray}
V(x)=\frac{1}{2} \nu x^2. \label{HarmonicPotential}
\end{eqnarray}
Let $$\beta=\frac{k_e}{k_B T}, \quad \lambda=\sqrt{\frac{\nu}{k_e}} x.$$
Then the particle distribution becomes
\begin{eqnarray}
	\tilde p(\lambda_1,\lambda_2, \cdots, \lambda_N) d\lambda_1 d\lambda_2\cdots d\lambda_N
	&=&\frac{1}{\tilde Z_N} \left( \prod_{1 \le i <j \le N}|\lambda_i-\lambda_j|^\beta \right) 
	e^{-\frac{\beta}{2} \sum\limits_{j=1}^N \lambda_j^2}  
	d\lambda_1 d\lambda_2\cdots d\lambda_N. \label{beta-ensemble}
\end{eqnarray}
A system  of random variables $\lambda_1, \lambda_2, \cdots, \lambda_N$ with distribution (\ref{beta-ensemble})
is called the $\beta$-ensemble.
The $\beta$-ensemble, $\beta=1, 2, 4$, describes  the joint density of eigenvalues of 
the three classical matrix models, i.e., the Gaussian orthogonal ensemble(GOE), Gaussian unitary ensemble(GUE) 
and Gaussian symplectic ensemble(GSE), respectively.
For general $\beta>0$, the $\beta$-ensemble can be realised as the joint density of 
eigenvalues of the {\sl spiked $\beta$-Hermite} matrix ensemble \cite{BV}.
$\beta$-ensemble for general $\beta$ has also other physical applications,
for example, it can be mapped to a chiral Liouville theory with central charge \cite{DV}.
Also, in some sense, the harmonic potential (\ref{HarmonicPotential}) is not so serious a limitation
since the universality of the $\beta$-ensemble had been proved by Bourgade et al. \cite{BEY}.

The interesting case is the thermodynamic limit $N \rightarrow \infty$.
Almost all particles distribute in $[-\sqrt{2 N}, \sqrt{2 N}]$ obeying  the Wigner semicircle law
with an  approximate density $\sigma(\lambda)=\pi^{-1} \sqrt{2 N-\lambda^2} $ \cite{Mehta}, 
i.e.,  the particle  number in $[\lambda,\lambda+d \lambda]$ is about $\sigma(\lambda) d\lambda$.
But few particles may lie outside $[-\sqrt{2 N}, \sqrt{2 N}]$.
It is proved that near the edge a proper scaling limit is the soft edge probability distribution \cite{RRV}
$$F_\beta(t)=\lim_{N \rightarrow \infty} E_{\beta N}^\mathrm{Soft}
\left( 0; \left( \sqrt{2 N}+\frac{t}{\sqrt{2}N^{1/6}} ,\infty \right) \right), $$
where
$$ E_{\beta N}^\mathrm{Soft}
\left( 0; \left( t ,\infty \right) \right) =\int_{\lambda_N=-\infty}^t \cdots 
\int_{\lambda_1=-\infty}^t \tilde{p}(\lambda_1, \cdots, \lambda_N)
 d\lambda_1 \cdots d \lambda_N .$$
$F_\beta(t)$ is called the Tracy-Widom distribution.

The  explicit expressions for $F_\beta(t)$ for $\beta=1,2,4$ are classical \cite{TW1, TW2, RRV}
\begin{eqnarray}
F_\beta(t)=\left\{ \begin{array}{ll}
\exp\left(\frac{1}{2}\int_{\infty}^t (s-t) u(s)^2 ds \right) \exp\left(\frac{1}{2}\int_{\infty}^t u(s) ds \right),
& \beta=1, \\
\exp\left( \int_{\infty}^t (s-t) u(s)^2 ds \right),  & \beta=2, \\
\exp\left(\frac{1}{2}\int_{\infty}^{2^{2/3} t} (s-2^{2/3} t) u(s)^2 ds \right) 
\cosh \left(\frac{1}{2}\int_{\infty}^{2^{2/3} t}u(s)ds \right), &\beta=4.
\end{array}\right.\nonumber
\end{eqnarray}
Also $F_2(t)$  has a Fredholm determinant representation
$F_2(t)=\det (I-A_t)$, where $A_t$ is the Fredholm integral operator on $(t, \infty)$ with the Airy kernel
$\frac{\mathrm{Ai}(x) \mathrm{Ai}'(y)-\mathrm{Ai}'(x) \mathrm{Ai}(y)}{x-y}$.

The expansions of $F_\beta(t)$ at $t=-\infty$ are of special interests.
In \cite{TW1} and \cite{TW2}, Tracy and Widom obtained and proved $F_\beta(t)$  for $\beta=1, 2, 4$ without the constant term.
They also conjectured the values of the constant term $c_0$ for  $\beta=1, 2, 4$.
By the Deift-Zhou nonlinear steepest descent method \cite{DZ},
Deift et al. \cite{DIK} proved the constant term for $\beta=2$ 
and Baik et al. \cite{BBD} proved the constant terms for $\beta=1,2 ,4$. 
Finally Borot et al. \cite{BEMN} derived an amazing asymptotic expression of $F_\beta(t)$  for general $\beta>0$ at $t=-\infty$ 
by the loop-equation technique.
Their asymptotic expression is
\begin{eqnarray}
F_\beta(t)=\exp \left(-\frac{\beta}{24}|t|^3+\frac{\sqrt{2}}{3}(\frac{\beta}{2}-1)|t|^\frac{3}{2}+
\frac{1}{8}(\frac{\beta}{2}+\frac{2}{\beta}-3) \ln|t|+c_\beta+O(|t|^{-\frac{3}{2}}) \right) ,
\label{BEMN-asymp}
\end{eqnarray}
where the constant term $c_\beta$ is
\begin{eqnarray}
c_\beta&=&\frac{\gamma_{E}}{6 \beta}+\left(\frac{17}{8}-\frac{25}{24}(\frac{\beta}{2}+\frac{2}{\beta}) \right) \ln 2
-\frac{1}{2} \ln(\frac{\beta}{2})-\frac{\ln (2 \pi)}{4}+\frac{\beta}{2}\left(\frac{1}{12}-\zeta'(-1)\right)\nonumber\\
&&+\int_0^\infty \frac{1}{t^2}\frac{1}{e^{\beta t/2}-1} \left(\frac{t}{e^t-1}-1+\frac{t}{2}-\frac{t^2}{12}  \right) dt.
\label{BEMN-c0}
\end{eqnarray}
Here $\gamma_E$ denotes the  Euler's constant  and $\zeta$ refers to the Riemann zeta function. 
Note  the prime ' will always be used to denote derivative.

But the asymptotics (\ref{BEMN-asymp}) is only valid  at $t=-\infty$ and can not be continued to finite $t$.
In fact,  even with infinite terms  (\ref{BEMN-asymp})  can not determine $F_\beta(t)$, see Theorem \ref{theorem-MultiSolution} below.
So we still need the explicit expression  for $F_\beta(t)$ beyond $\beta=1, 2, 4$.
Following the pioneering work of Dumitriu and Edelman \cite{DE},
Bloemendal and Vir\'{a}g \cite{BV} finally found out a representation of $F_\beta(t)$ in terms of the solution of 
a linear  partial differential equation(PDE).
They represent $F_\beta(t)$  by the limit of $F(\beta; x,t)$
\begin{eqnarray}
F_\beta(t)=\lim_{x \rightarrow \infty} F(\beta; x,t), \label{BV-Represent}
\end{eqnarray}
where $F(\beta; x,t)$ is a special solution of the linear PDE
\begin{eqnarray}
\frac{\partial F}{\partial t}+\frac{2}{\beta} \frac{\partial^2 F}{\partial x^2}+(t-x^2) \frac{\partial F}{\partial x}=0.
\label{BVeq}
\end{eqnarray}
More precisely, they proved the following theorem.
\begin{theorem} \cite{BV} \label{thm-BV} \quad
PDE (\ref{BVeq}) with boundary 
\begin{eqnarray} \left\{ \begin{array}{c} 
F(\beta; x,t)  \xlongrightarrow{x \rightarrow \infty, t \rightarrow \infty}1\\ 
F(\beta; x,t)  \xlongrightarrow{x \rightarrow -\infty, t\, \mathrm{fixed}}0
\end{array}\right. \label{BV-boundary}
\end{eqnarray}
has a unique bounded smooth solution. $F_\beta(t)$ is represented by  the solution through (\ref{BV-Represent}).
\end{theorem}

So the remaining problem is to find a bounded solution of the Bloemendal-Vir\'{a}g equation (\ref{BVeq}) 
with the boundary condition (\ref{BV-boundary}).
In \cite{R1},  Rumanov proposed a Lax representation of (\ref{BVeq}) for even integer $\beta$.
Let  
$$\Psi_x=\hat L \Psi, \quad \Psi_t=\hat B \Psi $$ 
be Rumanov's Lax pair,
where $\hat L$ and $\hat B$ are $2 \times 2$ matrices.
$\Psi$  can denote  both a $2 \times 2$ non-singular matrix or a $2 \times 1$ column  vector.
Here we assume it is a column vector 
\begin{eqnarray}
\Psi(x,t)=\left( \begin{array}{c} \mathcal{F}(x,t)\\ \mathcal{G}(x,t) \end{array}\right). \label{RumanovPsi}
\end{eqnarray}
The key of  Rumanov's scheme  is to let $\mathcal{F}$ satisfy the rescaled Bloemendal-Vir\'{a}g equation
\begin{eqnarray}
\frac{\beta}{2} \frac{\partial \mathcal{F}}{\partial t}+\frac{\partial^2 \mathcal{F} }{\partial x^2}
+(t-x^2) \frac{\partial \mathcal{F}}{\partial x}=0. \label{RumanovScaling}
\end{eqnarray}
Combining some other considerations, Rumanov concluded
\begin{eqnarray}
F(\beta; x,t)=\mathcal{F}\left( \big( \frac{\beta}{2} \big)^\frac{1}{3} x,\big(\frac{\beta}{2} \big)^\frac{2}{3} t \right).
\label{Rumanov-Conj}
\end{eqnarray}
In \cite{R2}, Rumanov studied the $\beta=6$ case and expressed $F_\beta(t)$ by an auxiliary function $\eta(t)$ 
and the solution of Painlev\'{e} II 
\begin{eqnarray}
u''(t)=t u(t)+2 u(t)^3,\label{PII}
\end{eqnarray}
which he deduced as the Hastings-McLeod solution \cite{HM}.
The auxiliary function $\eta(t)$ satisfies a second order ODE that can be linearized.

Grava, Its, Kapaev and  Mezzadri \cite{GIKM} found  Rumanov's Lax pair for $\beta=6$ is gauge equivalent to the Lax pair of Painlev\'{e} II.
Their gauge transformation is of form
\begin{eqnarray}
\Psi(x,t)=e^{\frac{x^3}{6}-\frac{x t}{2} } \kappa(t) 
\left(  \begin{array}{cc}
\frac{1+q_2(t)}{2}x-\alpha(t)&-1\\
\frac{1-q_2^2(t)}{4}&0
\end{array}\right) \psi(t)^{\sigma_3} \psi_0(x,t),\label{GIKM-Gauge}
\end{eqnarray}
where $\kappa(t)$ and $\psi(t)$ are scalar functions,
and $\psi_0(x,t)$ is the $2 \times 1$ column vector of the wave function of Painlev\'{e} II,
and $\sigma_3$ is the Pauli matrix
$\sigma_3=\left( \begin{array}{cc}1&0\\0&-1 \end{array}\right) $.
They also suggested $\psi=- \frac{\mathrm{i}}{\sqrt{u}}$,  
where $u$ is the Hastings-McLeod solution of Painlev\'{e} II.
Then they showed  $q_2(t)$ and $\alpha(t)$ satisfy the ODE
\begin{eqnarray}
&&q_2'(t)=\frac{2}{3} \alpha q_2 +\frac{u'}{u} \frac{(1+q_2)(2-q_2)}{3}, \label{q2alpha-deqsA}\\
&&\alpha'(t)=\alpha \left( \frac{2}{3} \alpha+\frac{u'}{u} \frac{2-q_2}{3} \right)-\frac{t}{6}(1+q_2)
-\frac{u^2}{3} (3+q_2).
\label{q2alpha-deqs}
\end{eqnarray}
Moreover, they proved that
\begin{eqnarray}
q_2(t) \xlongrightarrow{t \rightarrow \infty}-1+o(1), \quad \alpha(t) \xlongrightarrow{t \rightarrow \infty} o(1).
\label{GIKM-pinf}
\end{eqnarray}
It is straightforward to verify that
\begin{eqnarray}
&&q_2(t) = \frac{1}{\sqrt{2}} (-t)^{-3/2} +\frac{21}{8} (-t)^{-3}
+\frac{1707}{64 \sqrt{2}} (-t)^{-9/2}
+\frac{49123}{256} (-t)^{-6} 
+\cdots, \label{NinfSolution-A}\\
&&\alpha(t)= \frac{1}{\sqrt{2}} (-t)^{1/2}-\frac{1}{8} (-t)^{-1}
-\frac{37}{64 \sqrt{2}} (-t)^{-5/2}
-\frac{373}{256} (-t)^{-4}
+\cdots\label{NinfSolution}
\end{eqnarray}
is an  asymptotic solution of Equation (\ref{q2alpha-deqsA})-(\ref{q2alpha-deqs}) at $t=-\infty$.

The open question in \cite{GIKM} contains two parts:
\begin{itemize}
\item[(1).] Prove the system (\ref{q2alpha-deqsA})-(\ref{q2alpha-deqs}) with (\ref{GIKM-pinf}) has a smooth solution 
on $(-\infty, \infty)$.
\item[(2).] Assume (1) succeeds.
Prove the solution in (1) has expansions (\ref{NinfSolution-A})-(\ref{NinfSolution}) at $t=-\infty$.
\end{itemize}

In this paper we will show that 
there are a $2$-parameter family of solutions of (\ref{q2alpha-deqsA})-(\ref{q2alpha-deqs}) that satisfy 
$q_2 \xlongrightarrow{t \rightarrow \infty} -1$ and
$\alpha \xlongrightarrow{t \rightarrow \infty}  0$. 
More precisely, at $t=\infty$ these solutions have   asymptotics  
\begin{eqnarray}
&&q_2(t)=-1+ \mathrm{Ai}(t) \left[  c_1 \, \mathrm{Bi}(\frac{t}{3^{2/3}})
+\ c_2 \,  \mathrm{Ai}(\frac{t}{3^{2/3}}) +c_1^2 \tilde M_2(t) \right]+o\left(e^{-\frac{4}{3} t^{3/2}} \right),  \label{PinfSolution-A}\\
&&\alpha(t)=-\frac{3^{1/3}}{2} \mathrm{Ai}(t)
\left[ c_1   \mathrm{Bi}'(\frac{t}{3^{2/3}})+ c_2 \mathrm{Ai}'(\frac{t}{3^{2/3}}) 
+c_1^2  \tilde N_2(t)  \right]+o\left(e^{-\frac{4}{3} t^{3/2}} \right), \label{PinfSolution}
\end{eqnarray}
where 
$\tilde M_2(t)=o\left( \mathrm{Ai}(\frac{t}{3^{2/3}})\right)$ and $\tilde N_2(t)=o\left( \mathrm{Ai}'(\frac{t}{3^{2/3}})\right)$.
The detailed expressions of $\tilde M_2(t)$ and $\tilde N_2(t)$ will be given in Section \ref{secPInf}.
If  $c_1 \neq 0$ or $c_2 \neq 0$,
the leading terms  in the asymptotics  (\ref{PinfSolution-A}) and (\ref{PinfSolution})  are obvious.
If $c_1=0$ and $c_2=0$, the leading terms  for the asymptotics  of $q_2+1$ and $\alpha$ are

\begin{eqnarray}
q_2+1&=&\frac{4 \pi }{3^{4/3}} \mathrm{Ai}(t) \left[ 
\mathrm{Bi}\left(\frac{t}{3^{2/3}}\right)  \int_{\infty}^t \mathrm{Ai}(s) \mathrm{Ai}\left(\frac{s}{3^{2/3}} \right) ds
-\mathrm{Ai}\left(\frac{t}{3^{2/3}}\right)  \int_{\infty}^t \mathrm{Ai}(s) \mathrm{Bi}\left(\frac{s}{3^{2/3}} \right) ds 
\right] \nonumber\\
&&+o\left(e^{-\frac{8}{3} t^{3/2}} \right), \label{q2Pinf00} \\
\alpha&=&\frac{2 \pi }{3} \mathrm{Ai}(t) \left[ 
\mathrm{Ai}'\left(\frac{t}{3^{2/3}}\right)  \int_{\infty}^t \mathrm{Ai}(s) \mathrm{Bi}\left(\frac{s}{3^{2/3}} \right) ds
-\mathrm{Bi}'\left(\frac{t}{3^{2/3}}\right)  \int_{\infty}^t \mathrm{Ai}(s) \mathrm{Ai}\left(\frac{s}{3^{2/3}} \right) ds 
\right] \nonumber\\
&&+o\left(\sqrt{t} e^{-\frac{8}{3} t^{3/2}} \right). \label{alphaPinf00}
\end{eqnarray}

Not all solutions with asymptotics of (\ref{PinfSolution-A}) and (\ref{PinfSolution}) at $t=\infty$ can be smoothly evolved to $t=-\infty$.
It may develop to singularity at $t=t_0$, which depends on $c_1$ and $c_2$.
Our first main result of this paper is the following.
\begin{theorem} \label{theorem-MultiSolution}
There is a region $R_{smooth}$ that is the neighbourhood of the positive $c_2$-axis including the origin $(0,0)$
in the $(c_1,c_2)$-plane, such that if $(c_1, c_2)$ is in the region  $R_{smooth}$  
then the solution defined by the asymptotics of (\ref{PinfSolution-A})-({\ref{PinfSolution}}) at $t=\infty$
is smooth for $t \in (-\infty, \infty)$ and has asymptotics of (\ref{NinfSolution-A})-(\ref{NinfSolution}) at $t=-\infty$.
\end{theorem}
In fact, the region $R_{smooth}$ is very large.
The numerical results for $R_{smooth}$ are shown in Figure 1.

\begin{center}
\begin{tikzpicture}[scale=1.0]
\def\xsl{0.06}
\def\ysl{0.15}
\def\kA{-0.15541289336570088408}
\def\kB{0.26918302746067802855}
\def\FigureAdata{(-82*\xsl,11.871736836678694551*\ysl)(-80*\xsl,10.804369079113448152*\ysl)(-78*\xsl,9.757263044204257262*\ysl)
	(-76*\xsl,8.7308340368661254861*\ysl)(-74.*\xsl,7.7255187944265246885*\ysl ) (-72.*\xsl,6.7417772305324282795*\ysl )
	(-70.*\xsl,5.7800943743569925183*\ysl )(-68.*\xsl,4.8409825332135549801*\ysl )(-66.*\xsl,3.924983711671050424*\ysl )
	(-64.*\xsl,3.0326723263116310396*\ysl )(-62.*\xsl,2.1646582626430298043*\ysl )(-60.*\xsl,1.321590329716352225*\ysl )
	(-58.*\xsl,0.50416017914765579964*\ysl )(-56.*\xsl,-0.28689323092144166585*\ysl )(-54.*\xsl,-1.0507785290983803741*\ysl )
	(-52.*\xsl,-1.7866460616630091204*\ysl )(-50.*\xsl,-2.4935811139578790093*\ysl )(-48.*\xsl,-3.1705960397853563446*\ysl )
	(-46.*\xsl,-3.8166210708598681614*\ysl )(-44.*\xsl,-4.4304935191577298655*\ysl )(-42.*\xsl,-5.0109450071859176933*\ysl )
	(-40.*\xsl,-5.556586257809025008*\ysl )(-38.*\xsl,-6.0658888363556192731*\ysl )(-36.*\xsl,-6.5371630487425954791*\ysl )
	(-34.*\xsl,-6.9685309388252799784*\ysl )(-32.*\xsl,-7.3578929638174185241*\ysl )(-30.*\xsl,-7.7028864090724406226*\ysl )
	(-28.*\xsl,-8.0008328559353228076*\ysl )(-26.*\xsl,-8.2486709168361969543*\ysl )(-24.*\xsl,-8.4428688030279054487*\ysl )
	(-22.*\xsl,-8.579308767041221961*\ysl )(-20.*\xsl,-8.6531315184834107728*\ysl )(-18.*\xsl,-8.6585224240855683604*\ysl )
	(-16.*\xsl,-8.5884111169331749531*\ysl )(-14.*\xsl,-8.4340395883317048771*\ysl )(-12.*\xsl,-8.184327885655102345*\ysl )
	(-11.*\xsl,-8.0193987097547176432*\ysl )(-10.*\xsl,-7.8249327357754552876*\ysl )(-9.*\xsl,-7.5984042020401527556*\ysl )
	(-8.*\xsl,-7.3368937564847637271*\ysl )(-7.*\xsl,-7.0370477404388446585*\ysl )(-6.*\xsl,-6.6950992287131021809*\ysl )
	(-5.*\xsl,-6.3070803742360884765*\ysl )(-4.*\xsl,-5.8695641693142019998*\ysl )(-3.5*\xsl,-5.6317769346673879945*\ysl )
	(-3.*\xsl,-5.3816949892451604968*\ysl )(-2.5*\xsl,-5.1203029248073909793*\ysl )(-2.*\xsl,-4.8492881225781356491*\ysl )
	(-1.5*\xsl,-4.5709321959563785152*\ysl )(-1.*\xsl,-4.2876849405313592513*\ysl )(-0.5*\xsl,-4.0016144778234764478*\ysl )
	(0*\xsl,-3.7141079331281263929*\ysl )(0.5*\xsl,-3.4259345704941775874*\ysl )(1.*\xsl,-3.1374679841879923513*\ysl )
	(2.*\xsl,-2.560231977605626288*\ysl )(2.5*\xsl,-2.2715666689318457117*\ysl )(2.75*\xsl,-2.1272306337917092481*\ysl )
	(3.*\xsl,-1.9828935921925428472*\ysl )(3.1*\xsl,-1.9251586234874645675*\ysl )(3.2*\xsl,-1.8674236111241669537*\ysl )}

\fill[fill=green!16] (-82*\xsl, 18*\ysl)-- plot[smooth] coordinates {\FigureAdata} --(30*\xsl,13.605563603024801057*\ysl)--(30*\xsl,18*\ysl) --cycle;
\fill[fill=yellow!16] (-82*\xsl, -14*\ysl)-- plot[smooth] coordinates {\FigureAdata} --(30*\xsl,13.605563603024801057*\ysl)--(30*\xsl,-14*\ysl) --cycle;
\draw[color=red] plot[smooth] coordinates {\FigureAdata};
\draw[scale=1,domain=(-0.5/\kA)*\xsl:30*\xsl,smooth,variable=\x,red] plot ({\x},{\ysl*(-1-\x/\xsl*\kA)/\kB});
\draw[->] (-87*\xsl,0)--(35*\xsl,0) node[right] {$c_1$};
\draw[->] (0,-14*\ysl)--(0,18*\ysl) node[right] {$c_2$};
\foreach \x in {-80*\xsl, -60*\xsl, -40*\xsl, -20*\xsl, 20*\xsl} \draw (\x,2 pt)--(\x,-2pt);
\foreach \y in {-10*\ysl, -5*\ysl, 5*\ysl, 10*\ysl, 15*\ysl} \draw (2 pt, \y)--(-2pt, \y);
\fill (3.21723628697558*\xsl,-1.8574722363319859394*\ysl) circle (1 pt) node[right] {$P_c$};
\node at(-0.4,-10*\ysl){$-10$}; \node at(-0.4,-5*\ysl){$-5$};  \node at(-0.3,5*\ysl){$5$};
\node at(-0.3,10*\ysl){$10$}; \node at(-0.3,15*\ysl){$15$};
\node at(20*\xsl,-0.3){$20$}; \node at(-20*\xsl-0.1,-0.3){$-20$}; \node at(-40*\xsl-0.1,-0.3){$-40$};
\node at(-60*\xsl-0.1,-0.3){$-60$}; \node at(-80*\xsl-0.1,-0.3){$-80$};
\node at(-0.2,0.2){O};
\node at(-1.8,1.5){$R_{smooth}$};
\node at(-1.7,-1.8){$R_{singular}$};
\end{tikzpicture}
\end{center}
\begin{center}
\begin{minipage}{12.5cm}
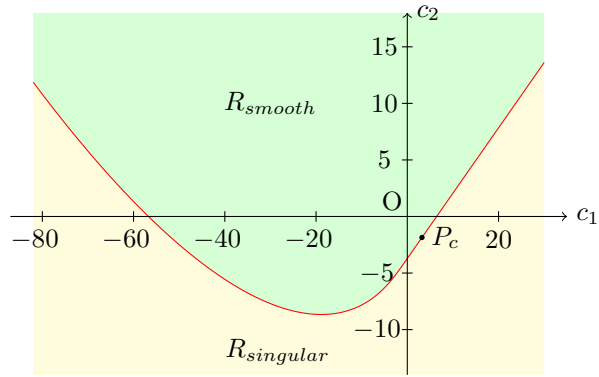
	Figure 1. $R_{smooth}$ and $R_{singular}$. $R_{smooth}$ is the light green region. If $(c_1, c_2)$ belongs to $R_{smooth}$, the solution defined at $t=\infty$ 
	by this $(c_1, c_2)$ is smooth on $(-\infty, \infty)$. 
	Else if $(c_1, c_2)$ belongs to $R_{singular}$(the light yellow region), the corresponding solution must have singularity 
	at some finite $t=t_0$. The red curve is the boundary between $R_{smooth}$ and $R_{singular}$.
	$P_c$ is a special point on the boundary curve: the boundary curve becomes straight on the right of $P_c$ .
\end{minipage}
\end{center}

Theorem \ref{theorem-MultiSolution} gives positive answers to the open questions (1) and (2) of \cite{GIKM}.
But the non-uniqueness of $q_2=q_2(c_1,c_2;t)$ causes the non-uniqueness of $F_6(t)$. In fact,
by formula \cite{GIKM}
\begin{eqnarray}
F_6 \left(\frac{t}{3^{2/3}} \right)=\frac{q_2-1}{2 q_2} \exp \left( -\frac{1}{3} \int_{\infty}^t \omega(s) ds 
	+\frac{2}{3}\int_{\infty}^t \frac{u'(s)}{u(s)} \frac{1+q_2(s)}{q_2(s)} ds \right),
\label{GIKM-F6}
\end{eqnarray}
where $\omega(s)=u(s)^4+s u(s)^2-u'(s)^2$,
we can verify that  $F_6(t)$ is indeed dependent on $c_1$ and $c_2$.
So we have to determine the  values of $c_1$ and $c_2$ to guarantee there is only a unique $F_6$(t).
To determine $c_1$ and $c_2$, we rely on Theorem \ref{thm-BV}.
Grava et al. \cite{GIKM} have formulated   $F(\beta=6; x, t)$ from $q_2$ and $\alpha$ as
\begin{eqnarray}
&&F\left(\beta=6;\frac{x}{3^{1/3}},\frac{t}{3^{2/3}}\right)=\kappa u^\frac{1}{2} 
\left[u^{-1} \left(\frac{1+q_2}{2}x-\alpha \right)Y_{12}^{(6)}(x,t)+Y_{22}^{(6)}(x,t) \right], \quad x \ge 0,
\label{GIKM-F-Y6} \\
&&F\left(\beta=6;\frac{x}{3^{1/3}},\frac{t}{3^{2/3}}\right)=-\kappa u^\frac{1}{2} e^{\frac{x^3}{3}-x t}
\left[u^{-1} \left(\frac{1+q_2}{2}x-\alpha \right)Y_{11}^{(3)}(x,t)+Y_{21}^{(3)}(x,t) \right],\quad x \le 0,
\label{GIKM-F-Y3}
\end{eqnarray}
where $Y^{(3)}$ and $Y^{(6)}$ are $2 \times 2$  matrices of the wave function of the Painlev\'{e} II.
By (\ref{GIKM-F-Y6})-(\ref{GIKM-F-Y3}),  $F(\beta=6;x,t)$ contains parameters $c_1$ and $c_2$.
By  (\ref{PinfSolution-A})-(\ref{PinfSolution}) and (\ref{GIKM-F-Y6})-(\ref{GIKM-F-Y3}), 
we can prove $c_1=0$  is enough to guarantee the boundary condition (\ref{BV-boundary}).
But if $c_2 \neq 0$, $F\left(\beta=6;\frac{x}{3^{1/3}},\frac{t}{3^{2/3}}\right)$ will 
grow exponentially near the line $x=-\sqrt{t}$ for $t \rightarrow \infty$.

The second main result of this paper is the following:
\begin{theorem} \label{theorem-UniqueSolution}
If and only if  $c_1=c_2=0$, the resulting $F(\beta=6; x,t)$ is bounded at the boundary $x^2+t^2=\infty$.
\end{theorem}

Now all requirements of Theorem \ref{thm-BV} are satisfied: 
$F(\beta=6;x,t)$ given by (\ref{GIKM-F-Y6})-(\ref{GIKM-F-Y3}) satisfies the Bloemendal-Vir\'{a}g equation (\ref{BVeq}); 
$c_1=0$ guarantees it satisfies the boundary condition (\ref{BV-boundary});
Theorem \ref{theorem-UniqueSolution} guarantees it is a bounded solution.
So $F(\beta=6;x,t)$ is indeed given by (\ref{GIKM-F-Y6}) and (\ref{GIKM-F-Y3}) with $c_1=c_2=0$.
We also note that $F_6(t)$ is  given by (\ref{GIKM-F6}) with the hidden parameters $c_1=c_2=0$.

\section{Derivation of the ODEs of $q_2$, $\alpha$ and $\kappa$ \label{sec-q2alpha-ODE}}
The Flaschka-Newell Lax pair of Painlev\'{e} II is \cite{FN}
\begin{eqnarray}
\frac{d \psi_0}{dx}&=& \hat L_0 \psi_0, \label{psix} \\
\frac{d \psi_0}{dt}&=& \hat B_0 \psi_0 \label{psit} 
\end{eqnarray}
where 
\begin{eqnarray}
&&\psi_0= \left(\begin{array}{c} \mathcal{F}_0(x,t)\\ \mathcal{G}_0(x,t) \end{array} \right), \label{psi0-DEF}\\
&&\hat L_0=\frac{x^2}{2}\sigma_3+x 
\left( \begin{array}{cc}
0 &u(t)\\
u(t)&0
\end{array} \right)
+\left( \begin{array}{cc}
-\frac{t}{2}-u(t)^2 &-u'(t)\\
u'(t)&\frac{t}{2}+u(t)^2
\end{array} \right), \label{L0}\\
&&\hat B_0=-\frac{x}{2} \sigma_3-
\left( \begin{array}{cc}
0 &u(t)\\
u(t)&0
\end{array} \right) .\label{B0}
\end{eqnarray}
By (\ref{RumanovPsi}) and (\ref{GIKM-Gauge}), we get
\begin{eqnarray}
\mathcal{F}(x,t)=\frac{1}{2} e^{\frac{x^3}{6}-\frac{t x}{2} }  (q_2(t) x+x-2 \alpha(t)) \kappa(t) \psi(t) \mathcal{F}_0(x,t)
- e^{\frac{x^3}{6}-\frac{t x}{2} } \frac{\kappa(t)}{\psi(t)} \mathcal{G}_0(x,t) . \label{GIKM-F}
\end{eqnarray}
Grava et. al suggested 
\begin{eqnarray}
\psi=- \frac{\mathrm{i}}{\sqrt{u}} . \label{psi-u} 
\end{eqnarray}
Substituting (\ref{GIKM-F}) and (\ref{psi-u}) into  (\ref{RumanovScaling}) with $\beta=6$,
we immediately obtain (\ref{q2alpha-deqsA})-(\ref{q2alpha-deqs}) and
\begin{eqnarray}
\frac{\kappa'}{\kappa}=-\frac{2}{3} \alpha-\frac{1}{3} (t+u^2) u^2+\frac{u'}{6 u}(2 q_2-1)+\frac{1}{3}(u')^2 .\label{kappa-deq}
\end{eqnarray}

By requiring $\mathcal{F}(x,t)\xlongrightarrow{x \rightarrow \infty, t \rightarrow \infty}1$, 
Grava et. al \cite{GIKM} proved (\ref{GIKM-pinf}) and
\begin{eqnarray}
\kappa u^\frac{1}{2} \xlongrightarrow{t \rightarrow \infty}1 .\label{kappa-pinf}
\end{eqnarray}
Equation (\ref{kappa-deq}) with boundary condition (\ref{kappa-pinf})  determine $\kappa$ completely
if a smooth solution of $(\alpha, q_2)$ has been obtained 
under the boundary condition (\ref{GIKM-pinf}), more precisely, the asymptotics of (\ref {PinfSolution-A})-(\ref {PinfSolution}).

The expansion of $\kappa(t)$ at $t=-\infty$ can also be obtained.
By (\ref{kappa-deq}), (\ref{NinfSolution-A})-(\ref{NinfSolution}) and the  asymptotics of Hastings-McLeod solution $u$, 
the asymptotics of $\frac{\kappa'}{\kappa}$ at $t=-\infty$ is obtained as
$$\frac{\kappa'(t)}{\kappa(t)}=\frac{t^2}{12}-\frac{\sqrt{2}}{3}(-t)^{\frac{1}{2}}+\frac{5}{24} (-t)^{-1}+
\frac{7}{32 \sqrt{2}} (-t)^{-\frac{5}{2}}+\cdots .$$
Therefore
\begin{eqnarray}
\kappa(t)=C_\kappa \times 
 e^{-\frac{1}{36} (-t)^{3}+\frac{2 \sqrt{2}}{9} (-t)^{3/2}-\frac{5}{24} \ln(-t)
 	+\frac{7}{48 \sqrt{2}} (-t)^{-3/2} +\cdots }.
\label{kappa-Ninf}
\end{eqnarray}

Assuming  (\ref{BEMN-asymp}), we get
\begin{eqnarray}
\ln C_\kappa =c_{\beta=6} -\frac{\ln 3}{36}+\frac{5}{4} \ln 2 .\label{C-kappa-1}
\end{eqnarray}
For $\beta=6$, Borot et al. \cite{BEMN}  was able to simplify (\ref{BEMN-c0}) to
\begin{eqnarray}
c_{\beta=6}=-\frac{97}{72} \ln 2-\frac{7}{36} \ln 3-\frac{\ln (2 \pi)}{6}+\frac{\ln \Gamma(\frac{1}{3})}{3}+\frac{\zeta'(-1)}{3}.
\label{BEMN-c0-beta6}
\end{eqnarray}
So we have
\begin{eqnarray}
\ln C_\kappa= -\frac{7}{72} \ln 2-\frac{2}{9} \ln 3-\frac{\ln (2 \pi)}{6}+\frac{\ln \Gamma(\frac{1}{3})}{3}+\frac{\zeta'(-1)}{3}.
\label{C-kappa-2}
\end{eqnarray}

The  value of $\ln C_\kappa$ can be obtained by  numerical experiments similar to the ones 
in Section \ref{secFig1} with $c_1=c_2=0$.
Our numerical experiments give
\begin{eqnarray}
\ln C_\kappa=-0.3445050500286934815501994065702518\cdots. \label{C-kappa-V}
\end{eqnarray}
In fact,  $\ln C_\kappa$ from our numerical experiments coincides with (\ref{C-kappa-2}) for more than $100$ digits, 
which gives a numerical verification of  (\ref{BEMN-asymp}) and (\ref{BEMN-c0}) for $\beta=6$.

Altogether, the algorithm  is as following.
First give the ansatz for $\mathcal{F}(x,t)$ as (\ref{GIKM-F}); 
then by (\ref{RumanovScaling}) obtain the ODEs for the unkowns;
next by the boundary condition for $\mathcal{F}(x,t)$ get all boundary conditions for the unknowns,
which should determine all unknowns uniquely;
at last prove the obtained $\mathcal{F}(x,t)$ satisfies all the requirements for it.

\section{Asymptotics of $q_2$ and $\alpha$ at $t=-\infty$ \label{secNInf}}
In this section we will show by  linearization analysis 
that the asymptotics  (\ref{NinfSolution-A})-(\ref{NinfSolution}) are not the asymptotics of a  specific solution of (\ref{q2alpha-deqsA})-(\ref{q2alpha-deqs}),
but of a general solution of (\ref{q2alpha-deqsA})-(\ref{q2alpha-deqs}).
More detailed analysis of these asymptotics will be given in Section \ref{sec-phi012}.

Suppose $(q_{20}(t), \alpha_0(t)) $ is a smooth solution of  (\ref{q2alpha-deqsA})-(\ref{q2alpha-deqs})
with asymptotics of (\ref{NinfSolution-A})-(\ref{NinfSolution}).
Let $(q_{2}(t), \alpha(t)) $ be a solution of (\ref{q2alpha-deqsA})-(\ref{q2alpha-deqs}) near $(q_{20}(t), \alpha_0(t))$.
Then $(q_{2}(t), \alpha(t)) $ can be expressed as
\begin{eqnarray}
&&q_2(t)= q_{20}(t)+\epsilon \, \mathscr{Q}(t), \nonumber\\
&&\alpha(t)=\alpha_{0}(t)+\epsilon \, \mathscr{A}(t), \nonumber
\end{eqnarray} 
where $\epsilon \rightarrow 0$ is  infinitesimal.

So $(\mathscr{Q}, \mathscr{A})$ satisfies the ODE
\begin{eqnarray}
&&\mathscr{Q}'(t)= 
\left(\frac{2}{3}\alpha_0+ \frac{u'}{u} \frac{1-2 q_{20}}{3}  \right)\mathscr{Q}+\frac{2}{3} q_{20}  \mathscr{A}, \nonumber\\
&&\mathscr{A}'(t)= -\frac{1}{3}\left( \frac{u'}{u}\alpha_0+\frac{t}{2}+u^2\right) \mathscr{Q}+
\left( \frac{4}{3} \alpha_0+\frac{u'}{u}  \frac{2-q_{20}}{3}\right)  \mathscr{A}. \nonumber 
\end{eqnarray}
At $t=-\infty$, the expansions of $u$, $\frac{u'}{u}$, $\alpha_0$ and $q_{20}$ are known
$$u=\sqrt{\frac{-t}{2}}+\cdots,\quad \frac{u'}{u}=\frac{1}{2 t}+\cdots,
\quad  q_{20}=\frac{1}{\sqrt{2}} (-t)^{-\frac{3}{2}}+\cdots, \quad \alpha_0=\sqrt{\frac{-t}{2}}+\cdots.$$
Therefore, $\mathscr{Q}$ and $\mathscr{A}$ satisfy
\begin{eqnarray}
&&\mathscr{Q}'(t)= \left(\frac{\sqrt{-2t}}{3}+\cdots \right) \mathscr{Q}+ \left(\frac{\sqrt{2}}{3}(-t)^{-\frac{3}{2}}+\cdots \right)
\mathscr{A}, \nonumber\\
&& \mathscr{A}'(t)= \left(\frac{1}{6}(-2 t)^{-\frac{1}{2}} +\cdots \right) \mathscr{Q} +
\left(\frac{2 \sqrt{-2t}}{3}+\cdots \right) \mathscr{A}.  \nonumber
\end{eqnarray}
Now it is clear that the  solution $(q_2,\alpha)$ 
is exponentially close to the solution $(q_{20}, \alpha_0)$  in an order of $e^{-\frac{2 \sqrt{2}}{9}(-t)^{3/2}}$. 

So we reach the following result.
\begin{theorem}
	If a solution $ ( q_{20}(t), \alpha_0(t) )$ 
	is nonsingular on $(-\infty, t_0]$ and has asymptotics (\ref{NinfSolution-A})-(\ref{NinfSolution}), 
	then the general solutions that are close to $ ( q_{20}(t), \alpha_0(t) )$  
	are all  non-singular on $(-\infty, t_0 ]$ and also have asymptotics (\ref{NinfSolution-A})-(\ref{NinfSolution}) at $t=-\infty$.
\end{theorem}

\section{The linear variables and the integral equations}
At $t=\infty$, it is convenient to work with $\alpha$ and
$$\tilde q_2= 1+q_2 .$$

Also it is helpful to remember
\begin{eqnarray}
u(t) =\mathrm{Ai}(t)+ \left( \frac{1}{32 \pi^{3/2}} t^{-\frac{7}{4}} +\cdots \right) e^{-\frac{6}{3}t^{3/2}}+\cdots.
\label{u-Pexpansion}
\end{eqnarray}

The ODEs for $\tilde q_2$ and $\alpha$ are
\begin{eqnarray}
&&\tilde q_2'(t)=\frac{2}{3}  (\tilde q_2-1) \alpha+\frac{u'}{u} \frac{3-\tilde q_2}{3} \tilde q_2, \label{wq2alpha-deqsA}\\
&&\alpha'(t)=\left(\frac{2}{3} \alpha +\frac{u'}{u} \frac{3-\tilde q_2}{3}\right) \alpha-\frac{t}{6} \tilde q_2
-\frac{2+\tilde q_2}{3} u^2. \label{wq2alpha-deqs}
\end{eqnarray}
Equations (\ref{wq2alpha-deqsA})-(\ref{wq2alpha-deqs}) are linearized by
\begin{eqnarray}
\tilde q_2(t)=\frac{\phi_1(t)}{\phi_0(t)} u(t), \quad \alpha(t)=\frac{\phi_2(t)}{\phi_0(t)} u(t), \label{Cole-Holpf}
\end{eqnarray}
where $\phi_1$, $\phi_2$ and $\phi_0$ satisfy
\begin{eqnarray}
&&\phi_1'(t)=-\frac{2}{3} \phi_2(t), \label{phis-1} \\
&&\phi_2'(t)=-\frac{2}{3}u(t) \phi_0(t)-\frac{1}{6} (t+2 u(t)^2) \phi_1(t), \label{phis-2} \\
&&\phi_0'(t)=\frac{1}{3} u'(t) \phi_1(t)-\frac{2}{3} u(t) \phi_2(t). \label{phis-0}
\end{eqnarray}

To analyze (\ref{phis-1})-(\ref{phis-0}), we need the following estimations. 
\begin{proposition} \label{prop-ineq}
	There exists $t_0^P$, such that for $t \ge t_0^P$:
	\begin{itemize}
		\item{$|\mathrm{Bi}(t)|<e^{\frac{2}{3} t^{3/2}} $ and $|\mathrm{Bi}'(t)|<e^{\frac{2}{3} t^{3/2}} \sqrt{t} $;}
		\item{$|\mathrm{Ai}(t)|<e^{-\frac{2}{3} t^{3/2}} $ and $|\mathrm{Ai}'(t)|<e^{-\frac{2}{3} t^{3/2}} \sqrt{t} $;}
		\item{$|u(t)|<e^{-\frac{2}{3} t^{3/2}} $ and $|u'(t)|<e^{-\frac{2}{3} t^{3/2}} \sqrt{t} $.}
	\end{itemize}
\end{proposition}
The proof is routine. Thus we omit it.
\begin{remark}
	For example, we can take $t_0^P=1$.
	By numerical results, $t_0^P=1$ satisfies all requirements for $t_0^P$  in Proposition \ref{prop-ineq}. 
	We assume  $t_0^P$ is sufficiently large. Also,  the actual value of $t_0^P$ is never needed in the following proofs. 
\end{remark}

\subsection{The three independent solutions defined near $t=\infty$}

For convenience, denote
$$g_1(t)=\mathrm{Bi}(\frac{t}{3^{2/3}}), \quad  g_2(t)=\mathrm{Ai}(\frac{t}{3^{2/3}}), \quad  g_0(t)=0. $$

We will show, for $k=1, 2, 0$, the  following integral equations
\begin{eqnarray}
\Phi_{1k}(t)&=&g_k(t)+\frac{2 \pi}{3^{1/3}} \left[ 
\mathrm{Ai}(\frac{t}{3^{2/3}}) \int_\infty^t  \mathrm{Bi}(\frac{s}{3^{2/3}})
\left(-\frac{2}{3} u(s) \Phi_{0k}(s)-\frac{1}{3} u(s)^2  \Phi_{1k}(s) \right) ds  \right. \nonumber \\
&&\left.-\mathrm{Bi}(\frac{t}{3^{2/3}}) \int_\infty^t  \mathrm{Ai}(\frac{s}{3^{2/3}})\left(-\frac{2}{3} u(s) \Phi_{0k}(s)-\frac{1}{3} u(s)^2  \Phi_{1k}(s) \right) ds
\right] , \label{Phi1k-InteglDEF}  \\
\Phi_{2k}(t)&=&-\frac{3}{2}g_k'(t)-\pi\left[
\mathrm{Ai}'(\frac{t}{3^{2/3}}) \int_\infty^t  \mathrm{Bi}(\frac{s}{3^{2/3}}) \left(-\frac{2}{3} u(s) \Phi_{0k}(s)-\frac{1}{3} u(s)^2  \Phi_{1k}(s) \right) ds \right. \nonumber \\
&& \left. -\mathrm{Bi}'(\frac{t}{3^{2/3}}) \int_\infty^t  \mathrm{Ai}(\frac{s}{3^{2/3}})\left(-\frac{2}{3} u(s) \Phi_{0k}(s)-\frac{1}{3} u(s)^2  \Phi_{1k}(s) \right) ds
\right], \label{Phi2k-InteglDEF} \\
\Phi_{0k}(t)&=&\delta_k^0 +\int_\infty^t \left( \frac{1}{3} u'(s) \Phi_{1k}(s)-\frac{2}{3} u(s) \Phi_{2k}(s)  \right) ds.
\label{Phi0k-InteglDEF}
\end{eqnarray}
define three independent solutions for (\ref{phis-1})-(\ref{phis-0}),
i.e., $(\phi_1(t),\phi_2(t), \phi_0(t))=(\Phi_{1k}(t), \Phi_{2k}(t), \Phi_{0k}(t))$, $k=1, 2, 0$,  
solves (\ref{phis-1})-(\ref{phis-0}).

For $k=1, 2, 0$, define 
\begin{eqnarray}
&&\Phi_{1k}^{(0)}(t)=(1-\delta_k^0)  g_k(t),  \label{Phi1k-0}\\
&&\Phi_{2k}^{(0)}(t)=-\frac{3}{2}(1-\delta_k^0) g_k'(t) , \label{Phi2k-0}\\
&&\Phi_{0k}^{(0)}(t)=\delta_k^0 +\int_\infty^t \left( \frac{1}{3} u'(s) \Phi_{1k}^{(0)}(s)-\frac{2}{3} u(s) \Phi_{2k}^{(0)}(s) \right) ds, \label{Phi0k-0}
\end{eqnarray}
where $\delta_k^0=\left\{ \begin{array}{ll}1, & k=0\\0,& k \neq 0 \end{array} \right.$.

For $j>0$, define
\begin{eqnarray}
\Phi_{1k}^{(j+1)}(t)&=&g_k(t)+\frac{2 \pi}{3^{1/3}} 
\left[
\mathrm{Ai}(\frac{t}{3^{2/3}}) \int_\infty^t  \mathrm{Bi}(\frac{s}{3^{2/3}})
\left(-\frac{2}{3} u(s) \Phi_{0k}^{(j)}(s)-\frac{1}{3} u(s)^2  \Phi_{1k}^{(j)}(s) \right) ds \right. \nonumber\\
&&\left. -\mathrm{Bi}(\frac{t}{3^{2/3}}) \int_\infty^t  \mathrm{Ai}(\frac{s}{3^{2/3}}) 
\left(-\frac{2}{3} u(s) \Phi_{0k}^{(j)}(s)-\frac{1}{3} u(s)^2  \Phi_{1k}^{(j)}(s) \right) ds
\right] , \label{Phi1kj1}\\
\Phi_{2k}^{(j+1)}(t)&=&-\frac{3}{2}g_k'(t)
-\pi \left[
\mathrm{Ai}'(\frac{t}{3^{2/3}}) \int_\infty^t  \mathrm{Bi}(\frac{s}{3^{2/3}}) 
\left(-\frac{2}{3} u(s) \Phi_{0k}^{(j)}(s)-\frac{1}{3} u(s)^2  \Phi_{1k}^{(j)}(s) \right) ds \right. \nonumber\\
&&\left. -\mathrm{Bi}'(\frac{t}{3^{2/3}}) \int_\infty^t  \mathrm{Ai}(\frac{s}{3^{2/3}}) 
\left(-\frac{2}{3} u(s) \Phi_{0k}^{(j)}(s)-\frac{1}{3} u(s)^2  \Phi_{1k}^{(j)}(s) \right) ds
\right],\label{Phi2kj1}\\
\Phi_{0k}^{(j+1)}(t)&=&\delta_k^0+\int_\infty^t \left( \frac{1}{3} u'(s) \Phi_{1k}^{(j+1)}(s)-\frac{2}{3} u(s) \Phi_{2k}^{(j+1)}(s)  \right) ds.
\label{Phi0kj1}
\end{eqnarray}

\begin{proposition} 
	For $t \ge t_0^P$,
	\begin{eqnarray*}
		&&|\Phi_{1k}^{(j+1)}(t)-\Phi_{1k}^{(j)}(t)|< C_{1k}^{(j+1)} e^{-(\frac{4}{3}j+\frac{4}{9} k+\frac{2}{3})t^{3/2}}, \quad
		|\Phi_{2k}^{(j+1)}(t)-\Phi_{2k}^{(j)}(t)|<  C_{2k}^{(j+1)} t^{\frac{1}{2}}  e^{-(\frac{4}{3}j+\frac{4}{9} k+\frac{2}{3})t^{3/2}},\\
		&&|\Phi_{0k}^{(j+1)}(t)-\Phi_{0k}^{(j)}(t)|< C_{0k}^{(j+1)} e^{-(\frac{4}{3}j+\frac{4}{9} k+\frac{4}{3} )t^{3/2}},
	\end{eqnarray*}	
	where
	\begin{eqnarray*}
		&&C_{11}^{(j)}=2 C_{21}^{(j)}=\pi^{j} 3^{\frac{-j}{3}} \prod_{i=1}^j \frac{6 i-1}{(3 i-1)(3 i-2)},
		\quad C_{01}^{(j)}=\pi^{j} 3^{\frac{-j}{3}} \prod_{i=1}^j \frac{6 i-1}{(3 i-1) (3 i+1)},  \\
		&&C_{12}^{(j)}=2 C_{22}^{(j)}=\pi^{j} 3^{\frac{-j}{3}} \prod_{i=1}^j \frac{6 i+1}{3 i (3 i-1)},
		\quad C_{02}^{(j)}=\frac{1}{2}\pi^{j} 3^{\frac{-j}{3}} \prod_{i=1}^j \frac{6 i+1}{3 i (3 i+2) }, \\
		&&C_{10}^{(j)}=2 C_{20}^{(j)}=3 \pi^{j} 3^{\frac{-j}{3}}  \frac{1}{(j-1)!} \prod_{i=1}^{j} \frac{2 i-1}{3 i-2},
		\quad C_{00}^{(j)}=\pi^{j} 3^{\frac{-j}{3}} \frac{1}{j!}\prod_{i=1}^{j} \frac{2 i-1}{3 i-2} .
	\end{eqnarray*}	
\end{proposition}

\begin{theorem}
$( \Phi_{1k}^{(j)}, \Phi_{2k}^{(j)}, \Phi_{0k}^{(j)})$, $k=1, 2, 0$,  defined by (\ref{Phi1k-0})-(\ref{Phi0kj1})
converge to the solutions of (\ref{phis-1})-(\ref{phis-0})  for $t \ge t_0^P$.
\end{theorem}

\section{Asymptotics of $q_2$ and $\alpha$ at $t=\infty$ \label{secPInf}}
At $t=\infty$, it is straightforward to verify
\begin{eqnarray}
&&\Phi_{11}(t)=\mathrm{Bi}(\frac{t}{3^{2/3}})+O(e^{-\frac{10}{9} t^{3/2}}),  \label{Phi11P} \\
&&\Phi_{21}(t)= -\frac{3^{1/3}}{2}\mathrm{Bi}'(\frac{t}{3^{2/3}})+O(\sqrt{t}e^{-\frac{10}{9} t^{3/2}}),  \label{Phi21P} \\
&&\Phi_{01}(t)=\int_\infty^t \left( \frac{1}{3} \mathrm{Ai}'(s) \mathrm{Bi}(\frac{s}{3^{2/3}})
+\frac{1}{3^{2/3}} \mathrm{Ai}(s) \mathrm{Bi}'(\frac{s}{3^{2/3}}) \right) ds +O(e^{-\frac{16}{9} t^{3/2}}), \label{Phi01P} \\
&&\Phi_{12}(t)=\mathrm{Ai}(\frac{t}{3^{2/3}})+O(e^{-\frac{14}{9} t^{3/2}}), \label{Phi12P} \\
&&\Phi_{22}(t)= -\frac{3^{1/3}}{2}\mathrm{Ai}'(\frac{t}{3^{2/3}}) +O(\sqrt{t}e^{-\frac{14}{9} t^{3/2}}), \label{Phi22P} \\
&&\Phi_{02}(t)=\int_\infty^t \left( \frac{1}{3} \mathrm{Ai}'(s) \mathrm{Ai}(\frac{s}{3^{2/3}})
+\frac{1}{3^{2/3}} \mathrm{Ai}(s) \mathrm{Ai}'(\frac{s}{3^{2/3}}) \right) ds +O(e^{-\frac{20}{9} t^{3/2}}), \label{Phi02P}\\
&&\Phi_{10}(t)=-\frac{4 \pi}{3^{4/3}}\left(\mathrm{Ai}(\frac{t}{3^{2/3}}) \int_\infty^t  \mathrm{Bi}(\frac{s}{3^{2/3}}) \mathrm{Ai}(s) ds -\mathrm{Bi}(\frac{t}{3^{2/3}}) \int_\infty^t  \mathrm{Ai}(\frac{s}{3^{2/3}}) \mathrm{Ai}(s) ds  \right)+O(e^{-2 t^{3/2}}) , \label{Phi10P} \\
&&\Phi_{20}(t)=\frac{2 \pi}{3}\left(\mathrm{Ai}'(\frac{t}{3^{2/3}}) \int_\infty^t  \mathrm{Bi}(\frac{s}{3^{2/3}}) \mathrm{Ai}(s) ds -\mathrm{Bi}'(\frac{t}{3^{2/3}}) \int_\infty^t  \mathrm{Ai}(\frac{s}{3^{2/3}}) \mathrm{Ai}(s) ds  \right) +O(\sqrt{t}e^{-2 t^{3/2}}), \label{Phi20P} \\
&&\Phi_{00}(t)=1+\int_\infty^t \left( \frac{1}{3} \mathrm{Ai}'(s) \tilde \Phi_{10}^{(1)}(s)
-\frac{2}{3} \mathrm{Ai}(s) \tilde \Phi_{20}^{(1)}(s)  \right) ds+ O(e^{-\frac{8}{3} t^{3/2}}), \label{Phi00P}
\end{eqnarray}
where $\tilde \Phi_{10}^{(1)}$ and $\tilde \Phi_{20}^{(1)}$  
are defined by dropping the error terms of (\ref{Phi10P}) and (\ref{Phi20P}) respectively.
\begin{lemma} \label{LEM-Classification}
The asymptotics of a solution of (\ref{wq2alpha-deqsA})-(\ref{wq2alpha-deqs}) 	at $t=\infty$ must belong to one of 
the following  three classes.
	
	\begin{itemize}
		\item[] Class A: \quad  $\tilde q_2(t) \xlongrightarrow {t \rightarrow \infty} 0$ 
		and $\alpha(t) \xlongrightarrow {t \rightarrow \infty} 0$. 
		\begin{eqnarray}
		&&\tilde{q_2}(t) = \left( c_1 \mathrm{Bi}(\frac{t}{3^{2/3}})
		+c_2 \mathrm{Ai}(\frac{t}{3^{2/3}}) +c_1^2 \tilde M_2(t)  e^{-\frac{8}{9} t^{3/2}}\right) \mathrm{Ai}(t)+o\left(e^{-\frac{4}{3} t^{3/2}} \right), \label{wq2-Pinf}\\
		&&\alpha(t) =
		-\frac{3^{1/3}}{2} \left( c_1 \mathrm{Bi}'(\frac{t}{3^{2/3}})+c_2 \mathrm{Ai}'(\frac{t}{3^{2/3}})
		+c_1^2 \tilde N_2(t)  e^{-\frac{8}{9} t^{3/2}} \right)
		 \mathrm{Ai}(t)+o\left(e^{-\frac{4}{3} t^{3/2}} \right), \label{alpha-Pinf}
		\end{eqnarray}
		where
		\begin{eqnarray}
		&&\tilde M_2(t)= -\mathrm{Bi}(\frac{t}{3^{2/3}}) \int_\infty^t \left( \frac{1}{3} \mathrm{Ai}'(s) \mathrm{Bi}(\frac{s}{3^{2/3}})
		+\frac{1}{3^{2/3}} \mathrm{Ai}(s) \mathrm{Bi}'(\frac{s}{3^{2/3}}) \right) ds   , \nonumber\\
		&&\tilde N_2(t)=-\mathrm{Bi}'(\frac{t}{3^{2/3}}) \int_\infty^t \left( \frac{1}{3} \mathrm{Ai}'(s) \mathrm{Bi}(\frac{s}{3^{2/3}})
		+\frac{1}{3^{2/3}} \mathrm{Ai}(s) \mathrm{Bi}'(\frac{s}{3^{2/3}}) \right) ds . \nonumber
		\end{eqnarray}
		\item[] Class B: \quad $\tilde q_2(t) \xlongrightarrow {t \rightarrow \infty} \infty$ 
		and $\alpha(t) \xlongrightarrow {t \rightarrow \infty} -\infty$.
		\begin{eqnarray}
		&&\tilde{q}(t)=\left[ \mathrm{Bi}(\frac{t}{3^{2/3}}) 
		+c_2 \left( \mathrm{Ai}(\frac{t}{3^{2/3}}) 
		-\mathrm{Bi}(\frac{t}{3^{2/3}})  \frac{ \Phi_{02}^{(0)}(t)}{\Phi_{01}^{(0)}(t)} \right) \right] \frac{\mathrm{Ai}(t)}{\Phi_{01}^{(0)}(t)}+o\left(e^{-\frac{8}{9} t^{3/2}} \right), \nonumber\\
		&&\alpha(t) = 
		-\frac{3^{1/3}}{2}\ \left[\mathrm{Bi}'(\frac{t}{3^{2/3}})
		+c_2 \left(\mathrm{Ai}'(\frac{t}{3^{2/3}})
	-\mathrm{Bi}'(\frac{t}{3^{2/3}})\frac{ \Phi_{02}^{(0)}(t)}{\Phi_{01}^{(0)}(t)} \right) \right] \frac{\mathrm{Ai}(t)}{\Phi_{01}^{(0)}(t)}+o\left(e^{-\frac{8}{9} t^{3/2}} \right).\nonumber 
		\end{eqnarray}
		\item[] Class C: \quad $\tilde{q}(t) \xlongrightarrow {t \rightarrow \infty} 2$ and 
		$\alpha(t) \xlongrightarrow {t \rightarrow \infty} \sqrt{t}$.
	\end{itemize}
	 
\end{lemma}

\begin{proof}
The general solution of (\ref{phis-1})-(\ref{phis-0}) is
$(\phi_1, \phi_2,\phi_0)\allowbreak  =(c_1 \Phi_{11}+c_2 \Phi_{12}+c_0 \Phi_{10}, 
\allowbreak c_1 \Phi_{21}+c_2 \Phi_{22}+c_0 \Phi_{20},\allowbreak 
c_1 \Phi_{01}+c_2 \Phi_{02}+c_0 \Phi_{00} )$.
So $\tilde q_2=\frac{c_1 \Phi_{11}+c_2 \Phi_{12}+c_0 \Phi_{10}}{c_1 \Phi_{01}+c_2 \Phi_{02}+c_0 \Phi_{00} }u $ and 
$\alpha= \frac{c_1 \Phi_{21}+c_2 \Phi_{22}+c_0 \Phi_{20}}{c_1 \Phi_{01}+c_2 \Phi_{02}+c_0 \Phi_{00} } u$.
If $c_0 \neq 0$, $c_0$ can be taken as $1$. 
By (\ref{Phi11P})-(\ref{Phi00P}), we have $\tilde q_2=(c_1 \Phi_{11}+c_2 \Phi_{12}+o(e^{-\frac{2}{3}t^{3/2}}) )(1-c_1 \Phi_{01}+o(e^{-\frac{8}{9}t^{3/2}})) u =(c_1 \Phi_{11}+ c_2 \Phi_{12}-c_1^2  \Phi_{11} \Phi_{01} +o(e^{-\frac{2}{3}t^{3/2}})) u $.
Also considering $u(t)=\mathrm{Ai}(t)+o(e^{-2 t^{3/2}})$ and (\ref{Phi11P})-(\ref{Phi00P}), 
we finally get (\ref{wq2-Pinf}).
We can verify directly $\tilde M_2(t)=o(\mathrm{Ai}(\frac{t}{3^{2/3}}))$.
Similarly, (\ref{alpha-Pinf}) is obtained.
Therefore, the Class A describes the asymptotics of $\tilde q_2(t)$ and $\alpha(t)$ at $t=\infty$  for $c_0 \neq 0$.
If $c_0=0$ and $c_1 \neq 0$, $c_1$ can be taken as $1$. Then we can prove the asymptotics belong to Class B in this case.
At last, if $c_0=c_1=0$, $c_2$ can be taken as $1$. Class C describes the  asymptotics of this case.
\end{proof}

\begin{proposition} \label{Prop-Asymp-P}
	A solution of (\ref{q2alpha-deqsA})-(\ref{q2alpha-deqs}),  which has property (\ref{GIKM-pinf}), 
	must have asymptotics  (\ref{PinfSolution-A})-(\ref{alphaPinf00}).
\end{proposition}
\begin{proof}
	Let $(\phi_1, \phi_2,\phi_0)$ and $(c_1, c_2, c_0)$ be the ones defined in the proof of Lemma \ref{LEM-Classification}.
	By (\ref{GIKM-pinf}) and $\tilde q_2=q_2+1$, we know $\tilde q_2(t) \xlongrightarrow {t \rightarrow \infty} 0$ and $\alpha(t) \xlongrightarrow {t \rightarrow \infty} 0$. This is the Class A case.
	The only remaining problem is to verify  (\ref{q2Pinf00})-(\ref{alphaPinf00}) for $c_1=c_2=0$.
	It is straightforward to verify this by (\ref{Phi10P})-(\ref{Phi00P}).
\end{proof}

\begin{remark}
	The error terms of (\ref{q2Pinf00})-(\ref{alphaPinf00}) are not optimal. 
	In fact, the error terms can be shown to be 
	$O\left(t^{-\frac{5}{2}} e^{-\frac{8}{3} t^{3/2}} \right)$ and $O\left(t^{-2} e^{-\frac{8}{3} t^{3/2}} \right)$ respectively 
	by a tedious calculation from (\ref{Phi1k-0})-(\ref{Phi0kj1}).
\end{remark}

\section{Proof of Theorem \ref{theorem-MultiSolution} \label{sec-phi012}}

\begin{proposition}
	There exists a minimal $k_0>0$ so as to $t+k_0 u(t)^2 \ge 0$ for all $t$.
\end{proposition}
\begin{proof}
	It is obvious $k_0>0$.
	Let $f(t)=t+2 u(t)^2$. 
	$f(t)$ has  minimum since $f(0)>0$, $f(-\infty)=0$ and $f(t)<0$ for  large negative enough $t$.
	Since $u(t) \neq 0$ for $t<\infty$, there exists  $k$ such that $t+2 u(t)^2+ k u(t)^2 \ge 0$ for all $t$.
	Obviously, such $k$ has minimum.
\end{proof}
\begin{remark} \label{rem-k0}
	By the preceding proof, we know $k_0>2$.
	To calculate $k_0$ numerically, we use $k_0 =-\min \left( \frac{t}{u^2} \right)$, 
	by which $k_0$  is calculated up to more than $100$ digits. 
	Though $k_0$ is so accurately known, its ``closed form" is still unknown.
	$k_0 \approx 2.1228589561253469$ is achieved at $t \approx -1.188111911480737877$.
	In the following, $k_0<\frac{10}{3}$ is needed.
	The proof of $k_0<\frac{10}{3}$ is somewhat technical and digressed, so we put it in Appendix \ref{Append-k0}.
\end{remark}

A visualized estimation of $k_0$ is given by Figure 2. 
\begin{center}
\begin{tikzpicture}[scale=1.0]
\def\xsl{1.2}
\def\ysl{1.2}
\def\FigureAdata{(-5.0*\xsl,-2.0041851709995756754*\ysl)(-4.9*\xsl,-2.0044617332891570902*\ysl)
	(-4.8*\xsl,-2.0047641171446057621*\ysl)(-4.7*\xsl,-2.005095498922797911*\ysl)
	(-4.6*\xsl,-2.0054595310743104192*\ysl)(-4.5*\xsl,-2.0058604209500617718*\ysl)
	(-4.4*\xsl,-2.0063030223589067119*\ysl)(-4.3*\xsl,-2.0067929414979289221*\ysl)
	(-4.2*\xsl,-2.0073366588888518036*\ysl)(-4.1*\xsl,-2.0079416688736610262*\ysl)
	(-4.0*\xsl,-2.0086166380080195191*\ysl)(-3.9*\xsl,-2.0093715832871079161*\ysl)
	(-3.8*\xsl,-2.0102180704736404564*\ysl)(-3.7*\xsl,-2.0111694317803674353*\ysl)
	(-3.6*\xsl,-2.0122410006729168043*\ysl)(-3.5*\xsl,-2.0134503594562040656*\ysl)
	(-3.4*\xsl,-2.0148175924036862999*\ysl)(-3.3*\xsl,-2.016365533250859526*\ysl)
	(-3.2*\xsl,-2.0181199906102483134*\ysl)(-3.1*\xsl,-2.0201099279061117735*\ysl)
	(-3.0*\xsl,-2.0223675653058709032*\ysl)(-2.9*\xsl,-2.0249283592447756705*\ysl)
	(-2.8*\xsl,-2.0278307997304881243*\ysl)(-2.7*\xsl,-2.0311159456713589819*\ysl)
	(-2.6*\xsl,-2.0348265926739214595*\ysl)(-2.5*\xsl,-2.0390059343416114757*\ysl)
	(-2.4*\xsl,-2.0436955347142895445*\ysl)(-2.3*\xsl,-2.0489323729183351357*\ysl)
	(-2.2*\xsl,-2.0547446469966389729*\ysl)(-2.1*\xsl,-2.0611459263068555917*\ysl)
	(-2.0*\xsl,-2.0681271126598727607*\ysl)(-1.9*\xsl,-2.0756454983215746335*\ysl)
	(-1.8*\xsl,-2.0836099787285206392*\ysl)(-1.7*\xsl,-2.0918611680989467408*\ysl)
	(-1.6*\xsl,-2.1001447479038755953*\ysl)(-1.5*\xsl,-2.1080758112246781531*\ysl)
	(-1.4*\xsl,-2.1150911949199831747*\ysl)(-1.3*\xsl,-2.1203857396030499731*\ysl)
	(-1.2*\xsl,-2.1228269784754752649*\ysl)(-1.1*\xsl,-2.1208407825102792655*\ysl)
	(-1.0*\xsl,-2.1122577757161327855*\ysl)(-0.9*\xsl,-2.0941065928528510436*\ysl)
	(-0.8*\xsl,-2.0623348799916673567*\ysl)(-0.7*\xsl,-2.0114317686878775037*\ysl)
	(-0.6*\xsl,-1.9339155867909321688*\ysl)(-0.5*\xsl,-1.8196366693134160426*\ysl)
	(-0.4*\xsl,-1.6548256902551500032*\ysl)(-0.3*\xsl,-1.4207906540373294795*\ysl)
	(-0.2*\xsl,-1.0921272782193970499*\ysl)(-0.1*\xsl,-0.63425325378673627679*\ysl)
	(0,0)}

\draw[color=green, thick] plot[smooth] coordinates {\FigureAdata};
\draw[->] (-5.1*\xsl,0)--(1*\xsl,0) node[right] {$t$};
\draw[->] (0,-2.4*\ysl)--(0,0.5*\ysl) node[right] {$y$};
\foreach \x in {-5*\xsl, -4*\xsl, -3*\xsl, -2*\xsl, -1*\xsl} \draw (\x,2 pt)--(\x,-2pt);
\foreach \y in {-2*\ysl, -1*\ysl} \draw (2 pt, \y)--(-2pt, \y);
\node at(0.4,-2*\ysl){$-2$}; \node at(0.4,-1*\ysl){$-1$};
\node at(-5*\xsl-0.1, -0.3){$-5$}; \node at(-4*\xsl-0.1, -0.3){$-4$};  \node at(-3*\xsl-0.1, -0.3){$-3$}; 
\node at(-2*\xsl-0.1, -0.3){$-2$}; \node at(-1*\xsl-0.1, -0.3){$-1$}; 
\node at(0.2,-0.2){O};
\node at(-0.5*\xsl, -1.2*\ysl){$\frac{t}{u^2}$};
\draw[color=red] (-5*\xsl,-2.2*\ysl)--(-0*\xsl,-2.2*\ysl); 
\end{tikzpicture}
\end{center}
\begin{center}
\begin{minipage}{12.5cm}
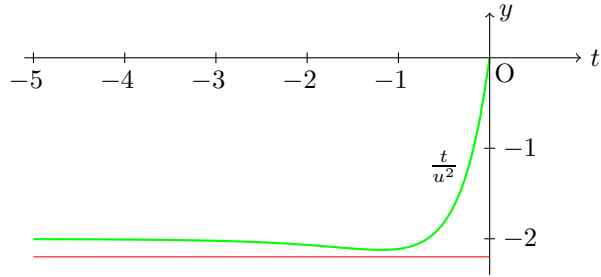
	Figure 2. Estimate $k_0$ by the graph of $\frac{t}{u^2}$. The green curve is the plot of $\frac{t}{u^2}$ and
	the horizontal red line is $y=-2.2$. 
	By the graph, it is obvious $\min\left( \frac{t}{u^2} \right)>-2.2$, i.e., $k_0<2.2$.
\end{minipage}
\end{center}

\begin{lemma} \label{LEM-PPPP}
	If  a solution of (\ref{phis-1})-(\ref{phis-0}) has properties $\phi_1(t_0)>0$, $\phi_2(t_0)>0$ and $\phi_0(t_0)>0$
	and $\frac{2}{3} \phi_0(t_0) -\frac{k_0-2}{6} u(t_0) \phi_1(t_0)>0$,
	then
	\begin{itemize}
	\item[(1)] $\phi_1(t)$, $\phi_2(t)$,  $\phi_0(t)$ and $\frac{2}{3} \phi_0(t) -\frac{k_0-2}{6} u(t) \phi_1(t)$
	are all monotonic decreasing on $(-\infty, t_0]$;
	\item[(2)]	$\phi_1(t)  \xlongrightarrow {t \rightarrow -\infty} \infty$,
	$\phi_2(t)  \xlongrightarrow {t \rightarrow -\infty} \infty$
	and $\phi_0(t)  \xlongrightarrow {t \rightarrow -\infty} \infty  $.
	\item[(3)] 	\begin{eqnarray}
	\lim_{t \rightarrow -\infty} \frac{\phi_1(t)}{\phi_0(t)} u(t)=1, \quad
	\lim_{t \rightarrow -\infty} \frac{\phi_2(t)}{\phi_0(t)}=1. \label{NLIM-phi-phi}
	\end{eqnarray}
	Correspondingly, 	near $t=-\infty$, 
	\begin{eqnarray}
	q(t)=o(1), \quad
	\alpha(t)=\sqrt{\frac{-t}{2}}+o((-t)^{\frac{1}{2}}) . \label{NLIM-q-alpha}
	\end{eqnarray}
	\end{itemize}
\end{lemma}
\begin{proof}
	(1)
	$\phi_1'(t_0)=-\frac{2}{3} \phi_2(t_0)<0$.
	
	$\phi_2'(t_0)=-\frac{2}{3} u(t_0) \phi_0(t_0)-\frac{1}{6} (t_0+k_0 u(t_0)^2 ) \phi_1(t_0)
	+\frac{1}{6}(k_0-2) u(t_0)^2 \phi_1(t_0) <0$.
	
	$\phi_0'(t_0)=\frac{1}{3} u'(t_0) \phi_1(t_0)-\frac{2}{3} u(t_0) \phi_2(t_0)<0$.
	\begin{eqnarray*}
		&&\frac{d}{dt}\left( \frac{2}{3} \phi_0(t_0) -\frac{k_0-2}{6} u(t)\phi_1(t_0) \right)\\
		&=&\frac{2}{3}\left( \frac{1}{3} u'(t) \phi_1(t_0) -\frac{2}{3} u(t_0) \phi_2(t_0) \right)
		-\frac{k_0-2}{6} u'(t_0)\phi_1(t_0)-\frac{k_0-2}{6} u(t_0) \left(-\frac{2}{3} \phi_2(t_0) \right)\\
		&=&\frac{10-3 k_0}{18} u'(t_0) \phi_1(t_0)-\frac{6-k_0}{9} u(t_0) \phi_2(t_0).
	\end{eqnarray*}
	By Remark \ref{rem-k0}, $k_0<\frac{10}{3}$.
	Therefore, $\frac{d}{dt}\left( \frac{2}{3} \varphi_0(t_0) -\frac{k_0-2}{6} u(t_0)\varphi_1(t_0) \right)<0$.
	So we have  $\phi_1(t_0-\epsilon)>\phi_1(t_0)>0$, $\phi_2(t_0-\epsilon)>\phi_2(t_0)>0$, 
	$\phi_0(t_0-\epsilon)>\phi_3(t_0)>0$
	and $\frac{2}{3} \phi_0(t) -\frac{k_0-2}{6} u(t) \phi_1(t)|_{t=t_0-\epsilon}>\frac{2}{3} \phi_0(t) -\frac{k_0-2}{6} u(t) \phi_1(t)|_{t=t_0}>0$.
	This process can be repeated endlessly. So the first statement of the lemma is proved.
	
	(2) By the preceding proof, $\phi_1(t)>0$,  $\phi_2(t)>0$,	$\phi_1'(t)=-\frac{2}{3} \phi_2(t)<0$ and
	$\phi_2'(t)<-\frac{1}{6} (t+k_0 u(t)^2 ) \phi_1(t) \le 0$. 
	So both $\phi_1(t)$ and $\phi_2(t)$ grow exponentially to infinity as $t \rightarrow -\infty$.
	By 	$\phi_0'(t)=\frac{1}{3} u'(t) \phi_1(t)-\frac{2}{3} u(t) \phi_2(t)$, 
	$\phi_0(t)  \xlongrightarrow {t \rightarrow -\infty} \infty  $ is got.
	
	(3) Let $x_1=\lim \limits_{t \rightarrow -\infty} \frac{\phi_1(t)}{\phi_0(t)} u(t)$ and 
	$x_2=\lim \limits_{t \rightarrow -\infty} \frac{\phi_2(t)}{\phi_0(t)}$.
	We apply  L'Hospital's rule to obtain the values of $x_1$ and $x_2$.
	It is legal since 	$\phi_1(t) u(t) \xlongrightarrow {t \rightarrow -\infty} \infty$,
	$\phi_2(t)  \xlongrightarrow {t \rightarrow -\infty} \infty$,
	$\phi_0(t)  \xlongrightarrow {t \rightarrow -\infty} \infty$ 
	and $\phi_0'(t)<0$ for $t<t_0$.
	Then by L'Hospital's rule and (\ref{phis-1})-(\ref{phis-0}), we obtain the algebraic equations for $x_1$ and $x_2$
	\begin{eqnarray}
	&&x_1= \frac{-\frac{2}{3}x_2 u(t)+\frac{u'(t)}{u(t)} x_1}{\frac{1}{3}\frac{u'(t)}{u(t)}x_1-\frac{2}{3}u(t) x_2 },\label{x1x2-1}\\
	&&x_2= \frac{-\frac{1}{6}\frac{t+2 u(t)^2}{u(t)}x_1-\frac{2}{3} u(t)}{\frac{1}{3}\frac{u'(t)}{u(t)}x_1-\frac{2}{3}u(t) x_2 }.
	\label{x1x2-2}
	\end{eqnarray}
	Note $t$ in (\ref{x1x2-1})-(\ref{x1x2-2}) should be understood as $t \rightarrow -\infty$.
	The algebraic equations  (\ref{x1x2-1})-(\ref{x1x2-2}) for $x_1$ and $x_2$ have $3$ set of solutions.
	Considering 
	\begin{eqnarray}
	u(t)=\sqrt{\frac{-t}{2}} \left(1-\frac{1}{8}(-t)^{-3}-\frac{73}{128}(-t)^{-6}-\frac{10657}{1024} (-t)^{-9}+\cdots \right),
	\label{u-expansion}
	\end{eqnarray}
	we can write  out  explicitly the $3$ set of solutions  as following.
	\begin{eqnarray*}
	&&\text{Set A:} \quad  x_1=1+\frac{1}{\sqrt{2}} (-t)^{-\frac{3}{2}}+\cdots , 
	\quad x_2=1-\frac{1}{4 \sqrt{2}} (-t)^{-\frac{3}{2}}+\cdots  .\\
	&&\text{Set B:} \quad  x_1=-8 (-t)^3+\frac{57}{2}+\cdots , 
	\quad x_2=2 \sqrt{2}  (-t)^{\frac{3}{2}}-\frac{39}{4 \sqrt{2}} (-t)^{-\frac{3}{2}}+\cdots. \\
	&&\text{Set C:} \quad  x_1=1-\frac{1}{\sqrt{2}} (-t)^{-\frac{3}{2}}+\cdots , 
	\quad x_2=-1 -\frac{1}{\sqrt{2}} (-t)^{-\frac{3}{2}}+\cdots.
	\end{eqnarray*}
The solutions of Set 2 and Set 3 are contradictory with the fact that $x_1>0$ and $x_2>0$.
So we get (\ref{NLIM-phi-phi}).
Considering $q_2(t)=\tilde q_2(t)-1$ and $u(t) \sim \sqrt{\frac{-t}{2}}$,  (\ref{NLIM-q-alpha}) is immediately obtained.
\end{proof}
By Proposition \ref{Prop-Negative}, we will see (\ref{NLIM-phi-phi}) is the general case.

\begin{proposition} \label{Prop-Positive}
	For $i=1,2,0$ and $j=2,0$, 
	$\Phi_{ij}$ are all positive and monotonic decreasing.
	Furthermore, all of them approach to positive infinity as $t \rightarrow -\infty$.
\end{proposition}

\begin{proof}
	(1) $j=2$ case.
	
	Both $\Phi_{12}(t)\xlongrightarrow {t \rightarrow \infty} 0_+ $ and 
	$\Phi_{22}(t)\xlongrightarrow {t \rightarrow \infty} 0_+ $ are obvious.
	By $\Phi_{02}(t)\xlongrightarrow {t \rightarrow \infty}  \frac{3^{1/6}}{8 \pi}  e^{-\frac{8}{9} t^{3/2}}t^{-\frac{1}{2}} $,
	we get $\Phi_{02}(t)\xlongrightarrow {t \rightarrow \infty} 0_+$.
	Further, $\frac{2}{3} \Phi_{02}(t)-\frac{k_0-2}{6} u(t) \Phi_{12}(t)=
	\frac{4-k_0}{8 \times 3^{5/6} \pi} e^{-\frac{8}{9} t^{3/2}} (t^{-\frac{1}{2}}+O(t^{-2})) $.
	Thus $\frac{2}{3} \Phi_{02}(t)-\frac{k_0-2}{6} u(t) \Phi_{12}(t) \xlongrightarrow {t \rightarrow \infty} 0_+$.
	By Lemma \ref{LEM-PPPP}, we proved the proposition for $j=2$.
	
	(2) $j=0$ case.
	
From
$$-\frac{4 \pi}{3^{4/3}}\left(\mathrm{Ai}(\frac{t}{3^{2/3}}) \int_\infty^t  \mathrm{Bi}(\frac{s}{3^{2/3}}) u(s) ds
-\mathrm{Bi}(\frac{t}{3^{2/3}}) \int_\infty^t  \mathrm{Ai}(\frac{s}{3^{2/3}}) u(s) ds  \right) 
=e^{-\frac{2}{3} t^{3/2}} \left( \frac{3}{4 \sqrt{\pi}} t^{-\frac{5}{4}}+O(t^{-\frac{11}{4}}) \right)$$
and
$$\frac{2 \pi}{3}\left(\mathrm{Ai}'(\frac{t}{3^{2/3}}) \int_\infty^t  \mathrm{Bi}(\frac{s}{3^{2/3}}) u(s) ds
-\mathrm{Bi}'(\frac{t}{3^{2/3}}) \int_\infty^t  \mathrm{Ai}(\frac{s}{3^{2/3}}) u(s) ds  \right)= 
e^{-\frac{2}{3} t^{3/2}} \left( \frac{3}{8 \sqrt{\pi}} t^{-\frac{3}{4}}+O(t^{-\frac{9}{4}}) \right),$$
$\Phi_{10} \xlongrightarrow {t \rightarrow \infty} 0_+ $ and $\Phi_{20} \xlongrightarrow {t \rightarrow \infty} 0_+ $ are obtained.
Obviously,  $\Phi_{00}(t) \xlongrightarrow {t \rightarrow \infty} 1 >0$.
Also, $\frac{2}{3} \Phi_{00}(t)-\frac{k_0-2}{6} u(t) \Phi_{10}(t) \xlongrightarrow {t \rightarrow \infty} \frac{2}{3} > 0$.
By Lemma \ref{LEM-PPPP}, the proposition is also true for $j=0$.
\end{proof}

\begin{proposition} \label{Prop-Positive-2} 
	For any fixed finite real $t_0$,
	if $c_2 \ge 0$, $c_0 >0$ and $c_1$ is sufficiently small, then
	$c_1 \Phi_{11}+c_2 \Phi_{12}+c_0 \Phi_{10}$, $c_1 \Phi_{21}+c_2 \Phi_{22}+c_0 \Phi_{20}$ and
	$c_1 \Phi_{01}+c_2 \Phi_{02}+c_0 \Phi_{00}$
	are all monotonic decreasing and positive on $(-\infty, t_0]$.
	Furthermore, $c_1 \Phi_{01}+c_2 \Phi_{02}+c_0 \Phi_{00} >0$ for $t \ge t_0$.
\end{proposition}
\begin{proof}
	Let $\phi_1(t)=c_1 \Phi_{11}+c_2 \Phi_{12}+c_0 \Phi_{10}$, $\phi_2(t)=c_1 \Phi_{21}+c_2 \Phi_{22}+c_0 \Phi_{20}$
	and $\phi_0(t)= c_1 \Phi_{01}+c_2 \Phi_{02}+c_0 \Phi_{00}$.
	It is obvious there exists $\delta_1>0$ such that $\phi_1(t_0)>0$,
	$\phi_2(t_0)>0$, $\phi_0(t_0)>0$ and  $\frac{2}{3} \phi_0(t_0) -\frac{k_0-2}{6} u(t_0) \phi_1(t_0)>0$
	for any $|c_1|<\delta_1$, since they are all greater than $0$ for $c_1=0$.
	By Lemma \ref{LEM-PPPP}, 	$\phi_1(t)$, $\phi_2(t)$ and $\phi_0(t)$ are positive and monotonic decreasing for $t \le t_0$.
	By (\ref{Phi02P}), there exists $\delta_2>0$ such that $\phi_0(t)>0$ for all $t \in [t_0,  \infty)$ for $|c_1|<\delta_2$,
	since $c_2 \Phi_{02}(t)+c_0 \Phi_{00}(t)>0$ for all $t \ge t_0$.
	Let $\delta=\min(\delta_1, \delta_2)$. Then, if $|c_1|<\delta$,
	$\phi_1(t)$, $\phi_2(t)$ and $\phi_0(t)$ have all the desired properties.
\end{proof}

After changing of variables $s=\sqrt{-t}$, $\tilde \phi_1(s)=\phi_1(t)$, $\tilde \phi_2(s)=\phi_2(t)$,
$\tilde \phi_0(s)=\phi_0(t)$,
we can see the ODE system for $\tilde \phi_1(s)$,$\tilde \phi_2(s)$ and $\tilde \phi_0(s)$ satisfies  
all the requirements of Theorem 12.3 of \cite{Wasow}.
After changing the variables back, we get the following result.
\begin{proposition} \label{Prop-Negative}
	At $t=-\infty$,  $\phi_1(t)$ and $\phi_2(t)$ and $\phi_0(t)$ have asymptotics
	\begin{eqnarray}
	&&\phi_1(t) \sim k_P \times \varphi_{1P}(t)	+k_O \times\varphi_{1O}(t) 	+k_N \times\varphi_{1N}(t), \label{PropN-1} \\
	&&\phi_2(t) \sim k_P\times \varphi_{2P}(t)	+k_O \times\varphi_{2O}(t)	+k_N \times\varphi_{2N}(t) , \label{PropN-2}\\
	&&\phi_0(t) \sim k_P \times\varphi_{0P}(t) 	+k_O\times \varphi_{0O}(t) 	+k_N \times \varphi_{0N}(t),  \label{PropN-0}
	\end{eqnarray}
	where
	\begin{eqnarray}
	&&\varphi_{1P}(t)=\left( \sqrt{2} (-t)^{-\frac{1}{2}} +\frac{55}{48} (-t)^{-2}+\frac{9107}{1536 \sqrt{2}} (-t)^{-\frac{7}{2}}
	+\cdots \right)  (-t)^\frac{1}{12} e^{\frac{2 \sqrt{2} }{9}  (-t)^{3/2}}, \label{varphi1P} \\
	&&\varphi_{2P}(t)=\left( 1-\frac{5}{48 \sqrt{2}} (-t)^{-\frac{3}{2}}-\frac{1013}{3072} (-t)^{-3}
	-\frac{2547101}{1327104 \sqrt{2}} (-t)^{-\frac{9}{2}} +\cdots\right) 
	 (-t)^\frac{1}{12} e^{\frac{2 \sqrt{2} }{9}  (-t)^{3/2}} , \label{varphi2P}\\
	&&\varphi_{0P}(t)=\left( 1+\frac{7}{48 \sqrt{2}} t^{-\frac{3}{2}}+\frac{145}{1024} (-t)^{-3} 
	+\frac{1496311}{1327104 \sqrt{2}} t^{-\frac{9}{2}}+\cdots\right)
	(-t)^\frac{1}{12} e^{\frac{2 \sqrt{2} }{9}  (-t)^{3/2}}, \label{varphi0P}\\
	&&\varphi_{1O}(t)=\left( 1+\frac{67}{72} (-t)^{-3}+\frac{551671}{10368} (-t)^{-6}+\frac{22894539769}{2239488} (-t)^{-9}
	+\cdots \right) (-t)^{-\frac{1}{6} }, \label{varphi1O} \\
	&&\varphi_{2O}(t)=\left( -\frac{1}{4} (-t)^{-1}-\frac{1273}{288} (-t)^{-4}
	- \frac{20411827}{41472} (-t)^{-7} +\cdots \right) (-t)^{-\frac{1}{6} },\label{varphi2O}\\
	&&\varphi_{0O}(t)=\left( \frac{1}{\sqrt{2}} (-t)^{-\frac{5}{2}}+\frac{1009}{18 \sqrt{2}} (-t)^{-\frac{11}{2}}
	+ \frac{6873355}{648 \sqrt{2}} (-t)^{-\frac{17}{2}}  +\cdots
	\right) (-t)^{-\frac{1}{6} } ,\label{varphi0O}\\
	&&\varphi_{1N}(t)=\left(-\sqrt{2} (-t)^{-\frac{1}{2}} +\frac{55}{48} (-t)^{-2}-\frac{9107}{1536 \sqrt{2}} (-t)^{-\frac{7}{2}}
	+\cdots \right)  (-t)^\frac{1}{12} e^{-\frac{2 \sqrt{2} }{9}  (-t)^{3/2}}, \label{varphi1N} \\
	&&\varphi_{2N}(t)=\left( 1+\frac{5}{48 \sqrt{2}} (-t)^{-\frac{3}{2}}-\frac{1013}{3072} (-t)^{-3}
	+\frac{2547101}{1327104 \sqrt{2}} (-t)^{-\frac{9}{2}} +\cdots
	\right)  (-t)^\frac{1}{12} e^{-\frac{2 \sqrt{2} }{9}  (-t)^{3/2}}, \label{varphi2N} \\
	&&\varphi_{0N}(t)=\left( -1+\frac{7}{48 \sqrt{2}} t^{-\frac{3}{2}} -\frac{145}{1024} (-t)^{-3}
	+\frac{1496311}{1327104 \sqrt{2}} t^{-\frac{9}{2}}+\cdots \right)
	 (-t)^\frac{1}{12} e^{-\frac{2 \sqrt{2} }{9}  (-t)^{3/2}}. \label{varphi0N} 
	\end{eqnarray}
\end{proposition}
\begin{remark}
	Proposition \ref{Prop-Negative} gives a straightforward explanation for the $3$ sets of solutions appearing in the proof of 
	Lemma \ref{LEM-PPPP}. If $k_P \neq 0$, the limits  are given by the Set A.
	If $k_P=0$ and $k_O \neq 0$, the limits are given by the Set B. Else if $k_P=k_O=0$, the limits are given by the Set C.
	There is no other possibility for the limits.
	However, the understanding of  Proposition \ref{Prop-Negative} is subtle:
	for example, 
	if in case the best approximation (obtained by optimal truncation) of $u(t)$ by its asymptotic series has an error more than 
	the order of $e^{-\frac{\sqrt{2}}{9} (-t)^{3/2}}$,
	the lower order terms in (\ref{PropN-1})-(\ref{PropN-0}) lost their meaning for REAL $t$.
	Fortunately, the error order of the best approximation of $u(t)$ by its asymptotic series
	is $e^{-(-t)^{3/2}}$. 
	So all terms in (\ref{PropN-1})-(\ref{PropN-0})  are  contributing.
\end{remark}

So we have constructed two sets of solutions for (\ref{phis-1})-(\ref{phis-0}):
 at $t=\infty$, we have $(\Phi_{1,i},\Phi_{2,i},\Phi_{0,i})$, $i=1, 2, 0$;
 and at $t=-\infty$, we have $(\varphi_{1,i},\varphi_{2,i},\varphi_{0,i})$, $i=P, O, N$.
Therefore, they only differ by a constant matrix
\begin{eqnarray}
&&\left(\begin{array}{ccc}
\Phi_{11}(t)&\Phi_{12}(t)&\Phi_{10}(t)\\
\Phi_{21}(t)&\Phi_{22}(t)&\Phi_{20}(t)\\
\Phi_{01}(t)&\Phi_{02}(t)&\Phi_{00}(t)
\end{array}
\right)=
\left(\begin{array}{ccc}
\varphi_{1P}(t)&\varphi_{1O}(t)&\varphi_{1N}(t)\\
\varphi_{2P}(t)&\varphi_{2O}(t)&\varphi_{2N}(t)\\
\varphi_{0P}(t)&\varphi_{0O}(t)&\varphi_{0N}(t)
\end{array}
\right) 
\left(\begin{array}{ccc}
k_{P1}&k_{P2}&k_{P0}\\
k_{O1}&k_{O2}&k_{O0}\\
k_{N1}&k_{N2}&k_{N0}
\end{array}
\right).
\label{Connection-KKK}
\end{eqnarray}
$\Phi_{ij}(t) \xlongrightarrow {t \rightarrow -\infty} \infty $ for $i=2,0$ and $j=1, 2, 0$
mean $k_{P2}>0$ and  $k_{P0}>0$.  In fact, their approximate values are $k_{P2} \approx 0.1678571$ and $k_{P0} \approx 0.6235798$.
More accurate values of them  are given in Section \ref{secFig1},  where they are determined  up to more than $100$ digits.

Now we are able to prove Theorem \ref{theorem-MultiSolution}.
\begin{proof}
Let $c_2 \ge 0$ and $c_1$ be sufficiently small.
Define $\phi_1(t)=c_1 \Phi_{11}+c_2 \Phi_{12}+\Phi_{10}$, $\phi_2(t)=c_1 \Phi_{21}+c_2 \Phi_{22}+\Phi_{20}$,
 $\phi_0(t)=c_1 \Phi_{01}+c_2 \Phi_{02}+\Phi_{00}$.
Then $\phi_1(t)$, $\phi_2(t)$ and $\phi_0(t)$ satisfy (\ref{phis-1})-(\ref{phis-0}).
Next define  $\tilde q_2(t)=\frac{\phi_1(t) }{\phi_0(t)} u(t)$ and  $\alpha(t)=\frac{\phi_2(t) }{\phi_0(t)} u(t)$.
By (\ref{Cole-Holpf}), $\tilde q_2(t)$ and $\alpha(t)$ satisfy (\ref{wq2alpha-deqsA})-(\ref{wq2alpha-deqs}).
By Proposition \ref{Prop-Positive-2}, $\tilde q_2(t)$ and $\alpha(t)$ are smooth on $(-\infty, \infty)$.
By Proposition \ref{Prop-Asymp-P} and  \ref{Prop-Negative}, 
$\tilde q_2(t)$ and $\alpha(t)$  have  desired asymptotics at $t=\infty$ and $t=-\infty$.
\end{proof}

\section{Numerical experiments about Figure 1 \label{secFig1}}
In this section, we give the details to generate Figure 1.
A few important data, such as the numerical values of the connection data, are also given,
as well as some interesting observations from the numerical experiments.

\subsection{Description of the procedure \label{procdure-I}}
By Section \ref{sec-phi012}, we know the singularities of $q_2(t)$ and $\alpha(t)$ 
are completely determined by the zeroes of $\phi_0(t)=\Phi_{00}(t)+c_1 \Phi_{01}(t)+c_2 \Phi_{02}(t)$.
So our first step is to obtain the numerical solutions of $\Phi_{ij}(t)$, $i, j=1,2,0$, for $t \in [t_N, t_P]$.
Since $\phi_0(t) \xlongrightarrow {t \rightarrow \infty} 1>0$,
we must require $\phi_0(-\infty) \ge 0$ in order that $q_2(t)$ and $\alpha(t)$ have no zeroes for $t \in (-\infty, \infty)$.
Therefore, our second step is to compute the matrix elements $k_{P1}$,  $k_{P2}$ and  $k_{P0}$, 
which will reflect the main behaviors of the solution near $t=-\infty$.
For moderate $t$, we use the numerical solutions to resolve if $\phi_0(t)$ has zeroes, which constitutes our last step.
More precisely, we determine the boundary between $R_{smooth}$ and $R_{singular}$  by seeking  the minimal $c_2$
such that $\Phi_{00}(t)+c_1 \Phi_{01}(t)+c_2 \Phi_{02}(t) \ge 0$ for all $t \in (-\infty, \infty)$ 
for given $c_1 \in \mathbb{R}$.
Since $\Phi_{02}(t)>0$, the problem is simplified to find the minimum of 
$\frac{\Phi_{00}(t)}{\Phi_{02}(t)}+c_1 \frac{\Phi_{01}(t)}{\Phi_{02}(t)}$ for given $c_1$,
i.e.,  $c_2=-\min \limits_{\forall t \in (-\infty,\infty)} \left(\frac{\Phi_{00}(t)}{\Phi_{02}(t)}+c_1 \frac{\Phi_{01}(t)}{\Phi_{02}(t)} \right)$.

Obviously, we have to do numerical integration of ODEs.
Currently, the most precise ODE integrator, such as {\sl Taylor}\cite{JZ} or high-order Runge-Kutta, 
can integrate an ODE  numerically  with precision up to $1000$ digits. 
For convenience, we use the build-in `NDSolve'  of {\sl Mathematica} to do the numerical integration for (\ref{wq2alpha-deqsA})-(\ref{wq2alpha-deqs}).
The default option of `NDSolve' is inappropriate to do high-precision numerical integration.
By explicitly giving  the `Method' option of `NDSolve', we can force it to use the   Gauss-Legendre Runge-Kutta method,
which is suitable for the high-precision purpose.
To save running time, we manage to let the typical precision be of order $10^{-120}$
\footnote{It does not mean  the final  error or final relative error is less than  $10^{-120}$.
It just mean,  the relative error is less than $10^{-120}$ for every step.} .
The stages  of the Runge-Kutta method are set  according to the precision goal of the numerical integration.
As a rule, we always let the stages greater than $100$, 
i.e, the order of the numerical scheme is always more than $200$. 
The step-sizes $h$ are chosen as  $0.01 \le h \le 0.05$.
By rough but careful estimations for each case, 
we guarantee that the errors generated by the numeric scheme itself are always  negligible, 
comparing to the errors that exist on the boundaries and  are propagated by the  ODE system.

\subsection{Determine  $T_P$}
$T_P$ is determined by two key factors: 
the truncation orders of $\Phi_{ij}$  at $t=\infty$ 
and the precision goal of the numerical integration.
We use (\ref{Phi11P})-(\ref{Phi00P}) as the truncation of $\Phi_{ij}$ 
since the higher order truncation will involve multiple integrals, which is difficult to get satisfactory high-precision results.

We demand the error of $\Phi_{ij}(t)$ at $t=0$ is of order $10^{-120}$.
For the solution $\Phi_{i1}(t)$, we can show  their errors at $t=0$ are of order $e^{-\frac{8}{9} t_P^{3/2}}$.
Solving $e^{-\frac{8}{9} t_P^{3/2}}=10^{-120}$, we get $t_P \approx 45.888$.
For convenience, we set $t_P=46$, at which the relative errors are of order 
$e^{-\frac{4}{3} 46^{3/2}} \approx 2.19 \times 10^{-181}$. 
So we set the precision goal of the numeric scheme as $10^{-182}$ in computing $\Phi_{i1}(t)$.
By a similar way, we could show it is appropriate to set $t_P=36$ and the precision goal as $10^{-120}$
in computing $\Phi_{i2}(t)$.
In computing $\Phi_{i0}(t)$, we also use $t_P=36$ and the precision $10^{-120}$.

\subsection{Determine  $T_N$}
By (\ref{PropN-1})-(\ref{PropN-0}), each $k_O$ term contributes to  a portion of order  
$e^{-\frac{2 \sqrt{2}}{9}(-t)^{3/2}}$.
By solving $ e^{-\frac{2 \sqrt{2}}{9}(-t)^{3/2}}=10^{-120}$, we get $t \approx -91.7761 $.
This means  the $k_O$ terms  can be neglected when $t<-91.7761$.
For convenience, we set $T_N=-92$.

\subsection{The numerical solution of $u(t)$}
For computation efficiency, the numerical solution of $u(t)$ is first obtained independently on $[t_M, t_H]$.
We demand the max error of $u(t)$ is of order  $10^{-120}$.
Since the best approximation of $u(t)$ by (\ref{u-expansion}) has an error of order $e^{-(-t)^{3/2}}$, 
$t_M$ is obtained by solving $e^{-(-t)^{3/2}}=10^{-120}$, i.e., $t_M \approx -42.42$.
For safety, we set $t_M=-44$. So, for $t_N \le t<t_M$, we use the asymptotic expansion (\ref{u-expansion}) 
up to the $(-t)^{-414}$ term to compute $u(t)$.
For $t_M \le t \le t_H$, $u(t)$   is obtained by the high-precision numerical integration of (\ref{PII}). 

Let $$\tilde u(t)=u(t)+\epsilon \mathscr{U}(t),$$
where $\epsilon$ is an infinitesimal, and $\tilde u$ also satisfies (\ref{PII}).
Then $\mathscr{U}(t)$ satisfies
$$\mathscr{U}''(t)=\left(t +6 u(t)^2 \right) \mathscr{U}(t).$$
As $u(t) \xlongrightarrow{t \rightarrow -\infty} \sqrt{\frac{-t}{2}}+\cdots $,
$\mathscr{U}(t)$ is of order $e^{\frac{1}{3} (-2 t)^{3/2}}$.
In the numerical experiments, $\epsilon \mathscr{U}(t)$ is  understood as the error.
So $\epsilon \sim 3 \times 10^{-240}$ in order that at $t=-44$ the error is of order $10^{-120}$.
Then the error of $u(t)$ at $t=0$ should be of order  $10^{-240}$.
The computational error of $u(t)$ behaves like $\epsilon_u \mathrm{Ai}(t)$ for $t>0$.
So, the relative error at $t=t_H$ should also be of order  $10^{-240}$.
For safety, we manage  the relative error at $t=t_H$ to be of order $10^{-250}$.

The value of $t_H$ is related to how $u(t)$ is approximated near $t=\infty$.
We take 
\begin{eqnarray}
u(t)\approx \mathrm{Ai}(t)-2 \pi  \mathrm{Ai}(t) \int_{\infty}^t \mathrm{Bi}(s) \mathrm{Ai}(s)^3 ds
+2 \pi  \mathrm{Bi}(t) \int_{\infty}^t \mathrm{Ai}(s)^4 ds \label{u-approxP}
\end{eqnarray}
as the approximation of $u(t)$.  
The error order of the approximation (\ref{u-approxP}) is about $e^{-\frac{10}{3}t^{3/2}}$.
So the relative error is of order  $e^{-\frac{8}{3}t^{3/2}}$.
By solving $e^{-\frac{8}{3}t^{3/2}}=10^{-250}$, $t_H \approx 35.985 $ is obtained.
For convenience, we set $t_H=36$.
We use $250$ digits in computing the numerical solution of $u(t)$.


\subsection{Transformations to avoid small step size} 
Fixing the step size, the Runge-Kutta method will be generally  more accurate to integrate a slow-varying system. 
To see the crux, let us consider the approximation of $e^{-t^{3/2}}$ by polynomials.
It is easy to see that the relative error of the approximation on interval $[100,100.01]$ is 
almost the same as the one on interval $[1,1.1]$ when using the same degree of approximation polynomials.
This means smaller step size is needed for large $t$ if the system increases or decreases too fast. 
To avoid the small step size for large $|t|$, 
we use Table 1 to transform the fast  variables to slow  ones.
\begin{center}
\begin{table}[H]
	\centering
	\caption{Transformations used to transform the fast variables to the slow ones}
\begin{tabular}{|c|c|c|c|}
\hline
fast variables& $t<-1$ & $-1 \le t \le 1$ & $t>1$ \\
\hline
$u$ & $u(t)$ & $u(t)$ & $\tilde u(t)= u(t) e^{\frac{2}{3} t^{3/2}} $ \\
\hline
& & &
$\tilde \Phi_{11}(t)= \Phi_{11}(t) e^{-\frac{2}{9}t^{3/2}}$\\
$\Phi_{i1}$, $i=1,2,0.$ &$\tilde \Phi_{i1}(t)=\Phi_{i1}(t) e^{-\frac{2 \sqrt{2}}{9} (-t)^{3/2}} $ &$\Phi_{i1}(t)$ &
$\tilde \Phi_{21}(t)= \Phi_{21}(t) e^{-\frac{2}{9}t^{3/2}}$\\
& & &
$\tilde \Phi_{01}(t)= \Phi_{01}(t) e^{\frac{4}{9}t^{3/2}}$ \\
\hline
& & &
$\tilde \Phi_{12}(t)= \Phi_{12}(t) e^{\frac{2}{9}t^{3/2}}$\\
$\Phi_{i2}$, $i=1,2,0.$ &$\tilde \Phi_{i2}(t)=\Phi_{i2}(t) e^{-\frac{2 \sqrt{2}}{9} (-t)^{3/2}} $ &$\Phi_{i2}(t)$ &
$\tilde \Phi_{22}(t)= \Phi_{22}(t) e^{\frac{2}{9}t^{3/2}}$\\
& & &
$\tilde \Phi_{02}(t)= \Phi_{02}(t) e^{\frac{8}{9}t^{3/2}}$ \\
\hline
& & &
$\tilde \Phi_{10}(t)= \Phi_{10}(t) e^{\frac{2}{3}t^{3/2}}$\\
$\Phi_{i0}$, $i=1,2,0.$ &$\tilde \Phi_{i0}(t)=\Phi_{i0}(t) e^{-\frac{2 \sqrt{2}}{9} (-t)^{3/2}} $ &$\Phi_{i0}(t)$ &
$\tilde \Phi_{20}(t)= \Phi_{20}(t) e^{\frac{2}{3}t^{3/2}}$\\
& & &
$\tilde \Phi_{00}(t)= (\Phi_{00}(t)-1) e^{\frac{4}{3} t^{3/2}} $ \\
\hline
\end{tabular}
\end{table}
\end{center}

\subsection{Numerical results}
The main numerical results are displayed in Figure 1.


\subsubsection{The values of $k_{P1}$, $k_{P2}$ and $k_{P0}$}
In principle, $k_{Pi}$, $i=1, 2, 0$, can be computed by any of $\lim \limits_{t \rightarrow -\infty} \frac{\Phi_{1i}(t)}{\varphi_{1P}(t)} $,
 $\lim \limits_{t \rightarrow -\infty} \frac{\Phi_{2i}(t)}{\varphi_{2P}(t)} $ or
  $\lim \limits_{t \rightarrow -\infty} \frac{\Phi_{0i}(t)}{\varphi_{0P}(t)} $.
In our numerical experiments, we use  
\begin{eqnarray}
k_{Pi}=\frac{\Phi_{0i}(t_N)}{\varphi_{0P}(t_N)} , \label{kPi-CAL}
\end{eqnarray}
which is a little more accurate than the other two choices.
In (\ref{kPi-CAL}), $\Phi_{0i}(t_N)$ are obtained directly from the numerical integration of ODEs of $\tilde \Phi_{ji}$,
while $\varphi_{0P}(t_N)$ is calculated by its asymptotic expansion (\ref{varphi0P}), 
where $\varphi_{0P}(t)$ is computed up to the term $c^0_{354} \times (-t)^{354} (-t)^{\frac{1}{12}} e^{\frac{2 \sqrt{2}}{9}(-t)^{3/2}}$. 
It is not surprising that $c^0_{354}$ is very large since  (\ref{varphi0P}) is an asymptotic expansion.  
In fact, the term $c^0_{354} \times (-t)^{354}$ contributes about $1.74545 \times 10^{-126}$ at $t=-92$.
So  $k_{Pi}$ is determined with an approximate precision of $10^{-120}$.
The final numerical results of $k_{Pi}$ are 
\begin{eqnarray}
&&k_{P1} =-0.0969123435570255523226380385083332 \cdots, \label{kP1-v}\\
&&k_{P2} = 0.167857102921338590132168687360301197 \cdots, \label{kP2-v}\\
&&k_{P0}= 0.62357981669501424223251084362366955 \cdots. \label{kP0-v}
\end{eqnarray}

\subsubsection{$\Phi_{ij}(t)$ near $t=0$}
By (\ref{Phi11P})-(\ref{Phi11P}), (\ref{varphi1P})-(\ref{varphi0P}) and (\ref{kP1-v})-(\ref{kP0-v}),
the main behaviors of $\Phi_{ij}(t)$ at $t=\pm \infty$  have been described.
We demonstrate their behaviors on the ``transition zone"  by Figures 3, 4 and 5.  

\input{Figure3.tex}

\input{Figure4.tex}

\input{Figure5.tex}

\subsubsection{The critical point $P_c$}
In Section \ref{procdure-I}, we have explained
$c_2=-\min \limits_{\forall t \in (-\infty,\infty)} \left(\frac{\Phi_{00}(t)}{\Phi_{02}(t)}+c_1 \frac{\Phi_{01}(t)}{\Phi_{02}(t)} \right)$
on the boundary between $R_{smooth}$ and $R_{singular}$.
Given $c_1$, let the minimum is achieved at $t=t_z$.  The numerical results show  the $t_z$ is unique for any given $c_1$. 
So, on the boundary cure, $t_z=t_z(c_1)$.
It is obvious that both $t_z$ and $c_2$ must approach to $\infty$ when $c_1 \rightarrow -\infty$.
As $c_1$ increases gradually to $P_c$, $t_z$ decreases and finally approaches to $-\infty$ as displayed by Figure 6.

\begin{center}
\begin{tikzpicture}[scale=1.0]
\def\xsl{2.0}
\def\ysl{0.2}
\def\FigureDataA{(3.21723*\xsl,-1.8574758661190321782*\ysl)(3.2172*\xsl,-1.8574931866271077646*\ysl)
	(3.217*\xsl,-1.8576086566809389773*\ysl)(3.216*\xsl,-1.8581860069498466430568336*\ysl)
	(3.215*\xsl,-1.8587633572182017991*\ysl)(3.21*\xsl,-1.8616501085481464402*\ysl)
	(3.2*\xsl,-1.8674236111241669537*\ysl)(3.15*\xsl,-1.8962911210883407680867671234*\ysl)
	(3.1*\xsl,-1.9251586234874645675*\ysl)(3.*\xsl,-1.9828935921925428472*\ysl)
	(2.75*\xsl,-2.1272306337917092481*\ysl)(2.5*\xsl,-2.2715666689318457117*\ysl)(2.*\xsl,-2.560231977605626288*\ysl)
	(1.5*\xsl,-2.8488761612888077207*\ysl)(\xsl,-3.1374679841879923513*\ysl)(0.5*\xsl,-3.4259345704941775874*\ysl)
	(0,-3.7141079331281263929*\ysl)(-0.5*\xsl,-4.0016144778234764478*\ysl)}

\def\FigureDataB{(3.21723*\xsl,-16.556149874022113335*\ysl)(3.2172*\xsl,-15.528256439940496922*\ysl)
	(3.217*\xsl,-14.375164613766567857*\ysl)(3.216*\xsl,-13.29832856211287090247*\ysl)
	(3.215*\xsl,-12.89635371742941964*\ysl)(3.212500000*\xsl, -12.37273724*\ysl)
	(3.21*\xsl,-12.068826758311509639*\ysl)(3.2*\xsl,-11.424928005952301142*\ysl)(3.183750000*\xsl, -10.90872246*\ysl)
	(3.155000000*\xsl, -10.40370937*\ysl)(3.15*\xsl,-10.33887668313003289457955*\ysl)(3.102500000*\xsl, -9.876881768*\ysl)
	(3.1*\xsl,-9.8576940829757269616*\ysl)(3.050000000*\xsl, -9.534224626*\ysl)
	(3.*\xsl,-9.2860795209835792227*\ysl)(2.875000000*\xsl, -8.829573597*\ysl)
	(2.75*\xsl,-8.4930075801334799807*\ysl)(2.625000000*\xsl, -8.219666159*\ysl)
	(2.5*\xsl,-7.9854027613007211661*\ysl)(2.*\xsl,-7.2504803988566987227*\ysl)
	(1.5*\xsl,-6.6659591248477132144*\ysl)(\xsl,-6.1428481030692153702*\ysl)(0.5*\xsl,-5.6464150330525048495*\ysl)
	(0,-5.1614994722671110374*\ysl)(-0.5*\xsl,-4.6846767520716986454*\ysl)}

\draw[->] (-0.8*\xsl,0)--(3.6*\xsl,0) node[right] {$c_1$};
\draw[->] (0,-17.5*\ysl)--(0,1.5*\ysl);
\foreach \x in { 1*\xsl, 2*\xsl, 3*\xsl } \draw (\x,2 pt)--(\x,-2pt);
\foreach \y in {-2*\ysl, -4*\ysl, -6*\ysl, -8*\ysl, -10*\ysl, -12*\ysl, -14*\ysl, -16*\ysl} \draw (2 pt, \y)--(-2pt, \y);
\node at(-0.4,-2*\ysl){$-2$}; \node at(-0.4,-4*\ysl){$-4$}; \node at(-0.4,-6*\ysl){$-6$}; \node at(-0.4,-8*\ysl){$-8$};
\node at(-0.4,-10*\ysl){$-10$}; \node at(-0.4,-12*\ysl){$-12$}; \node at(-0.4,-14*\ysl){$-14$};  \node at(-0.4,-16*\ysl){$-16$};
 \node at(1*\xsl, -0.3){$1$}; \node at(2*\xsl, -0.3){$2$};\node at(3*\xsl, -0.3){$3$};
\node at(0.2,-0.2){O};
\node at(2.5*\xsl,-3.3*\ysl){$c_2$};
\node at(2.5*\xsl,-9*\ysl){$t_z$};
\draw[color=red] plot[smooth] coordinates {\FigureDataA};
\draw[color=green] plot[smooth] coordinates {\FigureDataB};
\fill (3.21723628697558*\xsl,-1.8574722363319859394*\ysl) circle (1 pt) node[right] {$P_c$};

\end{tikzpicture}
\end{center}
\begin{center}
\begin{minipage}{12cm}
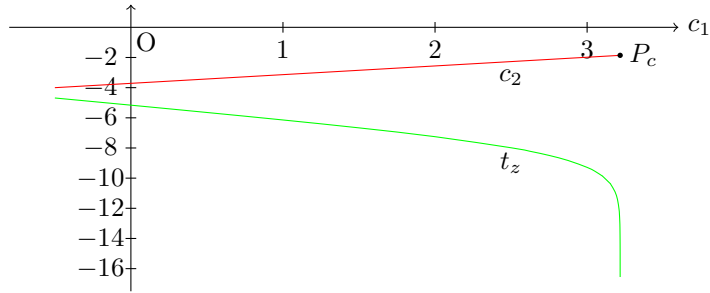
	Figure 6. Plots of $c_2$ (red) and $t_z$ (green).
	The $c_2$ curve, which is the boundary between $R_{smooth}$ and $R_{singular}$, is smooth.
	Though it looks very like a straight line, it is indeed a curve.
	The $t_z$ curve has apparently a singularity near $c_1=-\frac{1}{2}\frac{k_{P0}}{k_{P1}} \approx 3.217236287$. 
\end{minipage}
\end{center}

On the right of $P_c$, the minimum is always achieved at $t=-\infty$, i.e., $t_z=-\infty$.
So we have 
\begin{eqnarray}
c_2=- \left(\frac{k_{P0}}{k_{P2}}+c_1 \frac{k_{P1}}{k_{P2}} \right)= -\frac{k_{P0}}{k_{P2}}-\frac{k_{P1}}{k_{P2}} c_1, \label{c2c1}
\end{eqnarray}
which is the straight line right of $P_c$ in Figure 1.
For the critical point $P_c$, the interesting observation  from the numerical experiment is $c_1=-\frac{1}{2}\frac{k_{P0}}{k_{P1}}$.
Then, from (\ref{c2c1}),  $c_2=-\frac{1}{2}\frac{k_{P0}}{k_{P2}}$ at $P_c$.

\subsubsection{The values of $k_{O1}$, $k_{O2}$, $k_{O0}$, $k_{N1}$, $k_{N2}$ and $k_{N0}$ }
Integrating  $u(t)$ and $\Phi_{ij}(t)$ numerically along the path O-A-B in Figure 7,
we have obtained the values of $k_{P1}$, $k_{P2}$ and $k_{P0}$ with about $120$ digits of precision.
But $k_{O1}$, $k_{O2}$, $k_{O0}$, $k_{N1}$, $k_{N2}$ and $k_{N0}$ can not be obtained in this way.
To calculate them, we have to extend our numerical integration from the real line  to the complex plane of $t$
as  displayed by Figure 7.

\begin{center}
	\begin{minipage}[c]{12 cm}
		\begin{center}
		\begin{minipage}[c]{6 cm}
			\begin{tikzpicture}[scale=1.0]
			\def\radA{3}
			\def\cosA{-0.5}
			\def\sinA{0.86602540378443864676}
			\def\cosB{0.76484218728448842626}
			\def\sinB{0.64421768723769105367}
			\def\cosC{0.637635084398867861}
			\def\sinC{0.77033856137652142}
			\def\xscal{0.03260869565217391}
			\def\yscal{0.03260869565217391}
\def\FigureAdata{
(-62.598229272737975885*\xscal, -67.420039245892837011*\yscal)
(-61.1822*\xscal,-64.1434*\yscal)(-59.5413*\xscal,-60.2646*\yscal)(-58.1007*\xscal,-56.7738*\yscal)
(-56.8267*\xscal,-53.6079*\yscal)(-55.6928*\xscal,-50.7171*\yscal)(-54.678*\xscal,-48.0613*\yscal)(-53.7652*\xscal,-45.6082*\yscal)
(-52.9407*\xscal,-43.3311*\yscal)(-52.1929*\xscal,-41.2077*\yscal)(-51.5124*\xscal,-39.2197*\yscal)(-50.8911*\xscal,-37.3514*\yscal)
(-50.3224*\xscal,-35.5894*\yscal)(-49.8004*\xscal,-33.9223*\yscal)(-49.3203*\xscal,-32.3403*\yscal)(-48.8779*\xscal,-30.8348*\yscal)
(-48.4695*\xscal,-29.3984*\yscal)(-48.0918*\xscal,-28.0245*\yscal)(-47.7422*\xscal,-26.7072*\yscal)(-47.4181*\xscal,-25.4416*\yscal)
(-47.1175*\xscal,-24.2229*\yscal)(-46.8384*\xscal,-23.0472*\yscal)(-46.5792*\xscal,-21.9107*\yscal)(-46.3383*\xscal,-20.8102*\yscal)
(-46.1145*\xscal,-19.7426*\yscal)(-45.9065*\xscal,-18.7052*\yscal)(-45.7133*\xscal,-17.6956*\yscal)(-45.5339*\xscal,-16.7115*\yscal)
(-45.3675*\xscal,-15.7507*\yscal)(-45.2133*\xscal,-14.8115*\yscal)(-45.0707*\xscal,-13.892*\yscal)(-44.9389*\xscal,-12.9905*\yscal)
(-44.8175*\xscal,-12.1055*\yscal)(-44.706*\xscal,-11.2357*\yscal)(-44.6039*\xscal,-10.3796*\yscal)(-44.5108*\xscal,-9.53597*\yscal)
(-44.4263*\xscal,-8.70365*\yscal)(-44.3502*\xscal,-7.88149*\yscal)(-44.2821*\xscal,-7.06841*\yscal)(-44.2218*\xscal,-6.26337*\yscal)
(-44.1691*\xscal,-5.46539*\yscal)(-44.1238*\xscal,-4.6735*\yscal)(-44.0857*\xscal,-3.88679*\yscal)(-44.0547*\xscal,-3.10436*\yscal)
(-44.0307*\xscal,-2.32532*\yscal)(-44.0136*\xscal,-1.54881*\yscal)(-44.0034*\xscal,-0.773987*\yscal)(-44.*\xscal,0)
(-44.0034*\xscal,0.773987*\yscal)(-44.0136*\xscal,1.54881*\yscal)(-44.0307*\xscal,2.32532*\yscal)(-44.0547*\xscal,3.10436*\yscal)
(-44.0857*\xscal,3.88679*\yscal)(-44.1238*\xscal,4.6735*\yscal)(-44.1691*\xscal,5.46539*\yscal)(-44.2218*\xscal,6.26337*\yscal)
(-44.2821*\xscal,7.06841*\yscal)(-44.3502*\xscal,7.88149*\yscal)(-44.4263*\xscal,8.70365*\yscal)(-44.5108*\xscal,9.53597*\yscal)
(-44.6039*\xscal,10.3796*\yscal)(-44.706*\xscal,11.2357*\yscal)(-44.8175*\xscal,12.1055*\yscal)(-44.9389*\xscal,12.9905*\yscal)
(-45.0707*\xscal,13.892*\yscal)(-45.2133*\xscal,14.8115*\yscal)(-45.3675*\xscal,15.7507*\yscal)(-45.5339*\xscal,16.7115*\yscal)
(-45.7133*\xscal,17.6956*\yscal)(-45.9065*\xscal,18.7052*\yscal)(-46.1145*\xscal,19.7426*\yscal)(-46.3383*\xscal,20.8102*\yscal)
(-46.5792*\xscal,21.9107*\yscal)(-46.8384*\xscal,23.0472*\yscal)(-47.1175*\xscal,24.2229*\yscal)(-47.4181*\xscal,25.4416*\yscal)
(-47.7422*\xscal,26.7072*\yscal)(-48.0918*\xscal,28.0245*\yscal)(-48.4695*\xscal,29.3984*\yscal)(-48.8779*\xscal,30.8348*\yscal)
(-49.3203*\xscal,32.3403*\yscal)(-49.8004*\xscal,33.9223*\yscal)(-50.3224*\xscal,35.5894*\yscal)(-50.8911*\xscal,37.3514*\yscal)
(-51.5124*\xscal,39.2197*\yscal)(-52.1929*\xscal,41.2077*\yscal)(-52.9407*\xscal,43.3311*\yscal)(-53.7652*\xscal,45.6082*\yscal)
(-54.678*\xscal,48.0613*\yscal)(-55.6928*\xscal,50.7171*\yscal)(-56.8267*\xscal,53.6079*\yscal)(-58.1007*\xscal,56.7738*\yscal)
(-59.5413*\xscal,60.2646*\yscal)(-61.1822*\xscal,64.1434*\yscal)
(-62.598229272737975885*\xscal, 67.420039245892837011*\yscal)
}

\fill[fill=green!16] (0, 0)--(-\radA*\cosC,-\radA*\sinC)-- plot[smooth] coordinates {\FigureAdata} --(-\radA*\cosC,\radA*\sinC) --cycle;
\fill[fill=yellow!16] plot (\radA *\cosA, \radA *\sinA) arc (120: 240: \radA)-- plot[smooth] coordinates {\FigureAdata} --(-\radA*\cosC,\radA*\sinC) --cycle;
			\draw[densely dashed, thin] (0,0)--(\radA*\cosA, \radA*\sinA);
			\draw[densely dashed, thin] (0,0)--(\radA*\cosA, -\radA*\sinA);
			\draw[red] (0,0)--(-\radA*\cosB, \radA*\sinB);
			\draw[red] (0,0)--(-\radA*\cosB, -\radA*\sinB);
			\draw[green] (0,0)--(-\radA*\cosC, \radA*\sinC);
			\draw[green] (0,0)--(-\radA*\cosC, -\radA*\sinC);
			\draw[->] (-1.2*\radA,0)--(0.5*\radA,0) node[right] {Re(t)};
			\draw[->] (0,-0.9*\radA,0)--(0, 0.9*\radA) node[right] {Im(t)};
		\node at(-97*\xscal,-7*\yscal){$-92$}; \node at(-97*\xscal,6*\yscal){B};
		\node at(-47*\xscal,-7*\yscal){$-44$}; \node at(-40*\xscal,6*\yscal){A};
		\node at(-58*\xscal,40*\yscal){C}; \node at(-58*\xscal,-40*\yscal){$\widebar{\text{C}}$};
		\node at(-76*\xscal,59*\yscal){D}; \node at(-76*\xscal,-59*\yscal){$\widebar{\text{D}}$};
		\node at(-60*\xscal,76*\yscal){E}; \node at(-60*\xscal,-77*\yscal){$\widebar{\text{E}}$};
		\node at(5*\xscal,-5*\yscal){O};
		\draw(0.4, 0) arc (0: 120: 0.4); \node at(10*\xscal,17*\yscal){$\frac{2 \pi}{3}$};	
			\end{tikzpicture}
		\end{minipage}
	\end{center}
	\begin{center}
		\begin{minipage}[c]{12 cm}
	\noindent	Figure 7. Paths used to integrate $\Phi_{ij}(t)$. 
	 $\angle \text{BOD}=\frac{7}{10}$ and $\angle \text{BOE}=\frac{2}{3}\arccos(\frac{135\ln 10}{184 \sqrt{46}}) \approx 0.879372$.
	The boundary between the light yellow region and the light green one is 
	 $r=\frac{44}{\left( \cos(\frac{3}{2}\theta-\frac{3\pi}{2}) \right)^{2/3}} $.
	$\Phi_{ij}$ at point D are used to calculate $k_{O1}$,  $k_{O2}$ and  $k_{O0}$.
	$k_{N1}$,  $k_{N2}$ and  $k_{N0}$ are calculated from $\Phi_{ij}$ at point E.
	For precision reason, we use path O-C-D rather than arc $\wideparen{\text{BD}}$ to  numerically integrate the ODEs 
	for  $\tilde \Phi_{ij}$. Path O-E is used for the same reason. 
	$u(t)$ on the path in the light green region is obtained by the numerical integration of (\ref{PII})  
	while on the path in the light yellow region it is calculated by the expansion (\ref{u-expansion}) up to the
	$(-t)^{-414}$ term.
	\end{minipage}
	\end{center}
	\end{minipage}
\end{center}

We compute $k_{Oi}$  by $k_{Oi} = \frac{\Phi_{1i}(t) -k_{Pi} \times \varphi_{1P}(t) }{\varphi_{1O}(t)}$,
where $t$ is chosen as the point D.
The argument of D is chosen by solving $e^{\frac{2 \sqrt{2}}{9} (92)^{3/2} \cos(\frac{3}{2}\theta +\frac{3\pi}{2})}=10^{60} $, i.e., 
$\pi-\theta \approx 0.699535$. For simplicity, we choose $\pi-\theta=\frac{7}{10}$.
It is easy to show that  $\tilde \Phi_{ij}$ lost their precision when they are integrated numerically 
along arc $\wideparen{\text{BD}}$ starting from B.
So we integrate them numerically along the ray O-D, by which $\tilde \Phi_{ij}$ 
can be guaranteed to have about $120$ digits of precision.
$k_{Oi}$ obtained by this way can be shown to have about $60$ digits of  precision,
which is almost the best that we can expect for the computation of $k_{Oi}$ 
when $u(0)$ is computed with about $240$ digits of precision.

To compute $k_{Ni}$, we use 
$k_{Ni}=\frac{\Phi_{2i}(t) -k_{Pi} \times \varphi_{2P}(t)-k_{Oi}\times \varphi_{2O} }{\varphi_{2N}(t)}$,
where $t$ is chosen as the point E.
The argument of E is chosen by solving $e^{\frac{2 \sqrt{2}}{9} (92)^{3/2} \cos(\frac{3}{2}\theta+\frac{3\pi}{2})}=10^{30} $.
At the first sight, one may want to evaluate $k_{Ni}$ from $\Phi_{ij}$ on the dotted line.
But \cite{FIKN} has proved 
\begin{eqnarray}
u(t)= \sqrt{\frac{-t}{2}} \left(1+O((-t)^{-3/2}) \right)
+\frac{\mathrm{i}}{2^\frac{7}{4} \sqrt{\pi}} (-t)^{-\frac{1}{4}} e^{-\frac{2 \sqrt{2}}{3} (-t)^{3/2}}
\left( 1+O(t^{-\frac{1}{4}})\right), \label{FI-Uexpansion}
\end{eqnarray}
for $\frac{2 \pi}{3} \le \arg(t)<\frac{4 \pi}{3}$.
So we should not use the expansions (\ref{varphi1P})-(\ref{varphi0N}) near $\theta =\frac{2 \pi}{3}$.
Considering the exponential term of (\ref{FI-Uexpansion}), 
we can show $k_{Ni}$ are best calculated near E.
Also, it can be shown that $k_{Ni}$ calculated in this way have about $30$ digits of  precision. 

The final numerical values of $k_{ij}$, $i=O,N$, $j=1,2,0$, are
\begin{eqnarray}
&&k_{O1}^+ =\left(k_{O1}^-\right)^*=  0.474787653555570800096\cdots+ \mathrm{i} \times  0.091372926529406526556\cdots ,\label{kO1}\\
&&k_{O2}^+=\left(k_{O2}^-\right)^* =  0.274118779588219579669\cdots- \mathrm{i} \times 0.158262551185190266698\cdots, \label{kO2}\\
&&k_{O0}^+ =\left(k_{O2}^-\right)^* =-1.018336045084649924885\cdots- \mathrm{i} \times 0.58793658975512151298 \cdots, \label{kO0}\\
&&k_{N1}^+ =\left(k_{N1}^-\right)^*= -0.19583328674156168848\cdots + \mathrm{i} \times 0.048456171778512776161\cdots , \label{kN1}\\
&&k_{N2}^+=\left(k_{N2}^-\right)^* =  0.0484561717785127761613\cdots+ \mathrm{i}\times 0.083928551460669295066\cdots, \label{kN2}\\
&&k_{N0}^+=\left(k_{N0}^-\right)^* =-0.360023975030083963185\cdots . \label{kN0}
\end{eqnarray}

From our numerical results, we observe that $k_{N1}^+= \left( \frac{7 \sqrt{3}}{6}-\frac{1}{2} \mathrm{i} \right)k_{P1} $,
$k_{N2}^+= \left(\frac{\sqrt{3}}{6}+\frac{1}{2} \mathrm{i} \right)k_{P2} $ and $k_{N0}^+= -\frac{\sqrt{3}}{3}k_{P0} $
with the errors less than $10^{-30}$, which are consistent with the estimated precision of 
the numerical $k_{N1}^+$,  $k_{N2}^+$ and $k_{N0}^+$.
Also it is observed that $\left(\frac{k_{O2}^+}{k_{P2}}\right)^*=-\frac{k_{O0}^+}{k_{P0}}$ with more than $60$ digits of precision.

\subsubsection{The solution corresponding to $P_c$}
Let us consider the  solutions of (\ref{phis-1})-(\ref{phis-0}) described by Figure 1.
We note that the solution corresponding to $P_c$ in Figure 1 has a special property.
For simplicity, we scale the solution as 
\begin{eqnarray}
\Phi_c&=&(\Phi_{1c},\Phi_{2c}, \Phi_{0c}) \nonumber\\
&=&\frac{2}{k_{P0}}\left( \Phi_{10},\Phi_{20}, \Phi_{00} \right)
-\frac{1}{k_{P1}}\left( \Phi_{11},\Phi_{21}, \Phi_{01} \right)
-\frac{1}{k_{P2}}\left( \Phi_{12},\Phi_{22}, \Phi_{02} \right) . \nonumber 
\end{eqnarray}
By the numerical connection data (\ref{kP1-v}-\ref{kP0-v}) and (\ref{kO1})-(\ref{kN0}),
it is easy to verify (within the tolerance of precision)
$\Phi_c(t)  \xlongrightarrow {t \rightarrow -\infty}  -2\sqrt{3} \left( \varphi_{1N}(t), \varphi_{2N}(t), \varphi_{0N}(t) \right)$.
So this special solution decreases exponentially to $0$ as $t \rightarrow -\infty$.
We also note the other bounded solutions at $t=-\infty$,
which are spanned by $ \Phi_c$ and $\frac{1}{k_{P1}}\left( \Phi_{11},\Phi_{21}, \Phi_{01} \right)
-\frac{1}{k_{P2}}\left( \Phi_{12},\Phi_{22}, \Phi_{02} \right)$, 
decrease algebraically\footnote{Just as the asymptotic series hints,
the numerical results show the decrease looks very like $(-t)^{-\frac{1}{6}}$.} to $0$.

\section{The wave function of Painlev\'{e} II }
The Lax pair of Painlev\'{e} II is
\begin{eqnarray}
	&&\frac{d \Psi_0}{dx}= \hat L_0 \Psi_0, \label{Psix} \\
	&&\frac{d \Psi_0}{dt}= \hat B_0 \Psi_0, \label{Psit}
\end{eqnarray}
where $\hat L_0$ and  $\hat B_0$ are defined by (\ref{L0}) and (\ref{B0}).
Unlike in Section \ref{sec-q2alpha-ODE} where $\psi_0$ is vector,  here, $\Psi_0$ is a $2 \times 2$ matrix.

Define the six regions in the complex $x$-plane as
$$\Omega_j= \left\{x \left| \frac{\pi}{2}+\frac{j-2}{3} \pi<\arg x<\frac{\pi}{2}+\frac{j}{3}\pi \right. \right\},
\quad j=1, 2, \cdots, 6.$$
Equation (\ref{Psix}) has $6$ canonical solutions $\Psi_0^{(j)}(x)$ defined in the regions $\Omega_j$, $j=1, \cdots, 6$,
\begin{eqnarray}
\Psi_0^{(j)}(x)
\xlongrightarrow {x \rightarrow \infty}
\left(I+\frac{m_1}{x}+\cdots \right) e^{\left(\frac{x^3}{6}-\frac{x t}{2} \right) \sigma_3},
\hspace*{2cm}\frac{\pi}{2}+\frac{j-2}{3} \pi<\arg x<\frac{\pi}{2}+\frac{j}{3}\pi. \nonumber
\end{eqnarray}
For convenience, we denote $\Omega_7=\Omega_1$ and $\Psi_0^{(7)}=\Psi_0^{(1)}$.
If $\Psi_0$ is known, then $u(t)$ can be recovered by
$$u= (m_1)_{21}=-(m_1)_{12}.$$

The sector $\Omega_j$ overlaps with $\Omega_{j+1}$. 
In the crossover region,
\begin{eqnarray}
\Psi_0^{(j+1)}=\Psi_0^{(j)} S_0^{(j)} .\label{overlap}
\end{eqnarray}
For the case that $u(t)$ is the Hastings-McLeod solution,
\begin{eqnarray}
\begin{split}
S_0^{(1)}=\left( \begin{array}{cc} 1&0\\-1 & 1\end{array} \right), \quad
S_0^{(2)}=I_{2\times 2}, \quad
S_0^{(3)}=\left( \begin{array}{cc} 1&0\\1 & 1\end{array} \right), \\
S_0^{(4)}=\left( \begin{array}{cc} 1&-1\\0 & 1\end{array} \right), \quad
S_0^{(5)}=I_{2\times 2}, \quad
S_0^{(6)}=\left( \begin{array}{cc} 1&1\\0 & 1\end{array} \right).
\end{split}
\nonumber
\end{eqnarray}

Define $Y^{(j)}(x,t)$ by
\begin{eqnarray}
Y^{(j)}(x,t)=\Psi_0^{(j)} (x,t) e^{-\left(\frac{x^3}{6}-\frac{xt}{2} \right) \sigma_3}. \label{Y-DEF}
\end{eqnarray}

For convenience, we call both $\Psi_0^{(j)}$ and  $Y^{(j)}$ as the wave functions of Painlev\'{e} II.

By (\ref{overlap}), $Y^{(j)}$ satisfy the Riemann-Hilbert problem illustrated by Figure 8.

\begin{center}
	\begin{minipage}[c]{8 cm}
		\begin{minipage}[c]{8 cm}
			\begin{tikzpicture}
			\draw[->, thin, gray] (-3.4,0) -- (3.4,0) node[below] {Re($x$)};
			\draw[->, thin, gray] (0,-3.4) -- (0,3.4) node[right] {Im($x$)};
			\draw[->, thick] (0,0) -- (0.5*1.7,0.866*1.7);\draw[-,thick] (0.5*1.7,0.866*1.7)--(0.5*3.4,0.866*3.4);
			\draw[->, thick] (0,0) -- (-0.5*1.7,0.866*1.7);\draw[-,thick] (-0.5*1.7,0.866*1.7)--(-0.5*3.4,0.866*3.4);
			\draw[->, thick] (0,0) -- (-0.5*1.7,-0.866*1.7);\draw[-,thick] (-0.5*1.7,-0.866*1.7)--(-0.5*3.4,-0.866*3.4);
			\draw[->, thick] (0,0) -- (0.5*1.7,-0.866*1.7);\draw[-,thick] (0.5*1.7,-0.866*1.7)--(0.5*3.4,-0.866*3.4);
			\node at(1.7,-0.8){$Y^{(5)}$}; \node at(1.7,0.8){$Y^{(6)}$}; \node at(1.7,0.2){$Y^{(5)}=Y^{(6)}$};
			\node at(0,1.7){$Y^{(1)}$};
			\node at(-1.7,-0.8){$Y^{(3)}$}; \node at(-1.7,0.8){$Y^{(2)}$}; \node at(-1.7,0.2){$Y^{(2)}=Y^{(3)}$};
			\node at(0,-1.7){$Y^{(4)}$};
			\node at(0,-3.8){$Y^{(j)} \xlongrightarrow {x \rightarrow \infty} I_{2 \times 2}$};
			\node at(2.5,2){$\left(\begin{array}{cc}1&e^{\frac{x^3}{3}-x t}\\0&1 \end{array}\right) $};
			\node at(-2.7,2){$\left(\begin{array}{cc}1&0\\-e^{x t-\frac{x^3}{3}}&1 \end{array}\right) $};
			\node at(2.7,-2){$\left(\begin{array}{cc}1&-e^{\frac{x^3}{3}-x t}\\0&1 \end{array}\right) $};
			\node at(-2.5,-2){$\left(\begin{array}{cc}1&0\\e^{x t-\frac{x^3}{3}}&1 \end{array}\right) $};
			\end{tikzpicture}
		\end{minipage}
		\begin{minipage}[c]{8 cm}
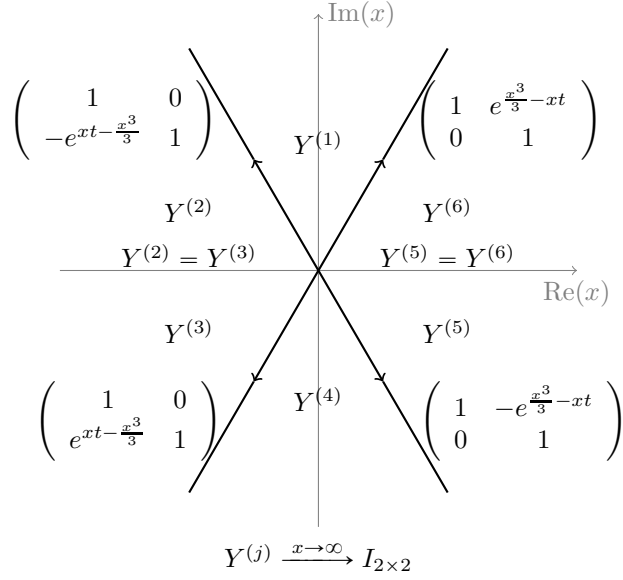
			\begin{center}Figure 8. The original Riemann-Hilbert Problem.\end{center}
		\end{minipage}
	\end{minipage}
\end{center}

To prove Theorem \ref{theorem-UniqueSolution},
a detailed analysis for the case $t \rightarrow \infty$ is needed.
So we deform the original Riemann-Hilbert problem to Figure 9.

\begin{center}
	\begin{minipage}[c]{8 cm}
		\begin{minipage}[c]{8 cm}
			\begin{tikzpicture}
			\draw[->, thin, gray] (-3.4,0) -- (3.4,0) node[below] {Re($x$)};
			\draw[->, thin, gray] (0,-3.4) -- (0,3.4) node[right] {Im($x$)};
			\node at(0,-3.8){$ \tilde Y^{(1)}, \tilde Y^{(3)}, \tilde Y^{(6)} 
				\xlongrightarrow {x \rightarrow \infty} I_{2 \times 2}$};
			\node at(2.6,2){$\left(\begin{array}{cc}1&e^{\frac{x^3}{3}-x t}\\0&1 \end{array}\right) $};
			\node at(-2.8,2){$\left(\begin{array}{cc}1&0\\-e^{x t-\frac{x^3}{3}}&1 \end{array}\right) $};
			\draw [->](0.8,0.4) .. controls (0.5*2,0.866*2-0.2) and (0.5*2.7,0.866*2.7-0.1) .. (0.5*3.4,0.866*3.4-0.01);
			\draw[thick] (0.8,0) circle (0.4cm);
			\draw (0.8,-0.4) .. controls (0.5*2,-0.866*2+0.2) and (0.5*2.7,-0.866*2.7+0.1) .. (0.5*3.4,-0.866*3.4+0.01);
			\draw [->](-0.8,0.4) .. controls (-0.5*2,0.866*2-0.2) and (-0.5*2.7,0.866*2.7-0.1) .. (-0.5*3.4,0.866*3.4-0.01);
			\draw[thick] (-0.8,0) circle (0.4cm);
			\draw (-0.8,-0.4) .. controls (-0.5*2,-0.866*2+0.2) and (-0.5*2.7,-0.866*2.7+0.1) .. (-0.5*3.4,-0.866*3.4+0.01);
			\node at(0,1){$\tilde {Y}^{(1)}=\tilde{Y}^{(4)}$};
			\node at(0,2){$\tilde {Y}^{(1)}=Y^{(1)}$};
			\node at(0,-2){$\tilde {Y}^{(4)}=Y^{(4)}$};
			\node at(2.5,-1){$\tilde {Y}^{(6)}=Y^{(6)}$};
			\node at(-2.5,-1){$\tilde {Y}^{(3)}=Y^{(3)}$};
			\node at(0.8,0){$\tilde {Y}_P$};
			\node at(-0.8,0){$\tilde {Y}_N$};
			\end{tikzpicture}
		\end{minipage}
		\begin{minipage}[c]{8 cm}
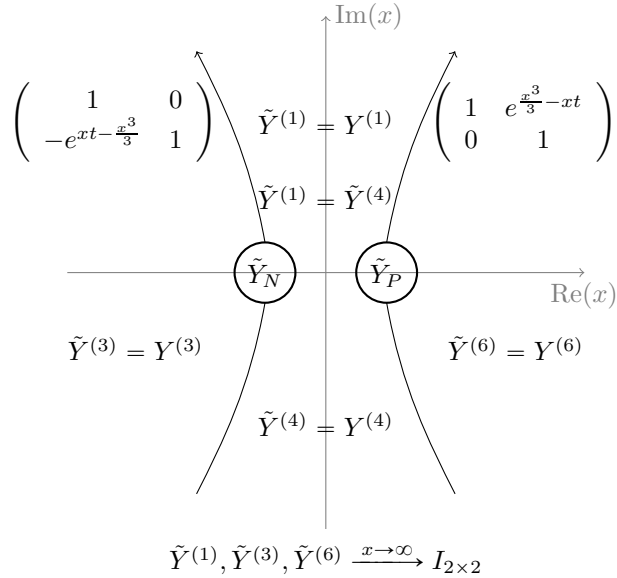
			\begin{center}Figure 9. The final Riemann-Hilbert Problem.\end{center}
		\end{minipage}
	\end{minipage}
\end{center}

By solving the Riemann-Hilbert problem of $Y$, one gets the following result.
\begin{lemma} \label{Y-RHP}  $Y^{(6)}$ and $Y^{(3)}$ have the following asymptotics:
\begin{itemize}
	\item [(A)] For $x \rightarrow \infty$ and  fixed $t$,
	$Y^{(6)}(x,t)$ has expansion $Y^{(6)}(x,t)=I_{2\times 2}+\frac{m_1(t)}{x} +\cdots$.
	
	\item [(B)] For $x \rightarrow -\infty$ and fixed t,
	$Y^{(3)}(x,t)$ has expansion $Y^{(3)}(x,t)=I_{2\times 2}+\frac{m_1(t)}{x}+\cdots$.
	
    \item [(C)] $\lim \limits_{x \rightarrow \infty, t \rightarrow \infty, t^{\frac{1}{4}}|x-\sqrt{t}| \rightarrow \infty} 
	Y^{(6)}(x,t)=I_{2\times 2}$
	and $\lim \limits_{x \rightarrow -\infty, t \rightarrow \infty,  t^{\frac{1}{4}}|x+\sqrt{t}| \rightarrow \infty} 
	Y^{(3)}(x,t)=I_{2\times 2}$.
	
	\item [(D)] For $t \rightarrow \infty$ and $0 \le x < \sqrt{t}- t^{\epsilon -\frac{1}{4}}$,
	$Y^{(6)} \rightarrow \left( \begin{array}{cc}1& -e^{\frac{x^3}{3}-x t} \\0 &1 \end{array}\right)$.
	
	\item [(E)] For $t \rightarrow \infty$ and $-\sqrt{t}+t^{\epsilon -\frac{1}{4}} < x \le 0 $,
	$Y^{(3)} \rightarrow \left( \begin{array}{cc}1& 0\\-e^{x t-\frac{x^3}{3}}  &1 \end{array}\right)$.
\end{itemize}
	In both cases (A) and (B), $m_1(t)=\left( \begin{array}{cc} (u')^2-u^4-t u^2&-u\\u&-(u')^2+u^4+t u^2\end{array} \right)$. 
\end{lemma}

Lemma \ref{Y-RHP} is already known, see for example \cite{FIKN}.

\begin{remark}
Because of (\ref{kappa-Ninf}), at $t=-\infty$,  $\kappa(t)$ is `smaller' than other quantities  in the formulae.
It is unnecessary to estimate $Y^{(3)}$ and $Y^{(6)}$ so accurately at $t=-\infty$.
\end{remark}

Lemma \ref{Y-RHP} fulfils parts of our purpose to prove (\ref{theorem-UniqueSolution}).
In fact we still need more detailed behaviour of $Y^{(3)}$ on $x=k \sqrt{t}$.
For completeness, we also give the results for $Y^{(6)}$.

Before we study the asymptotics of $Y^{(6)}(x,t)$ and $Y^{(3)}(x,t)$  along $x=k \sqrt{t}$, let us
first write down the ODEs for them, which our study will rely on.

By (\ref{Y-DEF}), $Y^{(j)}$ satisfies
\begin{eqnarray}
&&\frac{dY^{(j)}}{dx} =\hat L_0 Y^{(k)}+\left(\frac{t}{2}-\frac{x^2}{2} \right)  Y^{(j)} \sigma_3 ,\nonumber\\
&&\frac{dY^{(j)}}{dt} =\hat B_0 Y^{(k)}+\frac{x}{2} Y^{(j)} \sigma_3  .\nonumber
\end{eqnarray}
The detailed formulae are
\begin{eqnarray}
&&\frac{d}{dx} 
\left( \begin{array}{c}
Y^{(j)}_{11}\\Y^{(j)}_{21}
\end{array}\right)=
\left( \begin{array}{cc}
-u(t)^2&x u(t)-u'(t)\\x u(t)+u'(t)& t-x^2+u(t)^2
\end{array}\right)
\left( \begin{array}{c}
Y^{(j)}_{11}\\Y^{(j)}_{21}
\end{array}\right),
\label{Y1C-EQ-X}\\
&&\frac{d}{dt} 
\left( \begin{array}{c}
Y^{(j)}_{11}\\Y^{(j)}_{21}
\end{array}\right)=
\left( \begin{array}{cc}
0&-u(t)\\-u(t)&x
\end{array}\right)
\left( \begin{array}{c}
Y^{(j)}_{11}\\Y^{(j)}_{21}
\end{array}\right),
\label{Y1C-EQ-T}\\
&&\frac{d}{dx} 
\left( \begin{array}{c}
Y^{(j)}_{12}\\Y^{(j)}_{22}
\end{array}\right)=
\left( \begin{array}{cc}
-t+x^2-u(t)^2&x u(t)-u'(t)\\x u(t)+u'(t)& u(t)^2
\end{array}\right)
\left( \begin{array}{c}
Y^{(k)}_{12}\\Y^{(k)}_{22}
\end{array}\right),
\label{Y2C-EQ-X}\\
&&\frac{d}{dt} 
\left( \begin{array}{c}
Y^{(j)}_{12}\\Y^{(j)}_{22}
\end{array}\right)=
\left( \begin{array}{cc}
-x&-u(t)\\-u(t)& 0
\end{array}\right)
\left( \begin{array}{c}
Y^{(j)}_{12}\\Y^{(j)}_{22}
\end{array}\right).
\label{Y2C-EQ-T}
\end{eqnarray}

Along the line $x=k \sqrt{t}$, by
\begin{eqnarray}
\frac{dY^{(j)}(k \sqrt{t},t)}{dt}=\frac{k}{2 \sqrt{t}}\frac{dY^{(j)}(x,t)}{dx}\bigg|_{x=k\sqrt{t}}
+\frac{dY^{(j)}(x,t)}{dt}\bigg|_{x=k\sqrt{t}} , \quad j=3, 6 \nonumber, 
\end{eqnarray}
we get
\begin{eqnarray}
&&\frac{d}{dt} 
\left( \begin{array}{c}
Y^{(j)}_{11}\\Y^{(j)}_{21}
\end{array}\right)=
\left( \begin{array}{cc}
\frac{-k u(t)^2}{2 \sqrt{t}}&\frac{k^2-2}{2} u(t)-\frac{k}{2 \sqrt{t}} u'(t)\\ \frac{k^2-2}{2} u(t)+\frac{k}{2 \sqrt{t}}u'(t)&-k\frac{(k^2-3)t-u(t)^2}{2 \sqrt{t}}
\end{array}\right)
\left( \begin{array}{c}
Y^{(j)}_{11}\\Y^{(j)}_{21}
\end{array}\right),
\label{Y1C-SQRT-T}\\
&&\frac{d}{dt} 
\left( \begin{array}{c}
Y^{(j)}_{12}\\Y^{(j)}_{22}
\end{array}\right)=
\left( \begin{array}{cc}
k\frac{(k^2-3)t-u(t)^2}{2 \sqrt{t}}&\frac{k^2-2}{2} u(t)-\frac{k}{2 \sqrt{t}}u'(t)\\ \frac{k^2-2}{2} u(t)+\frac{k}{2 \sqrt{t}}u'(t)&\frac{k u(t)^2}{2 \sqrt{t}}
\end{array}\right)
\left( \begin{array}{c}
Y^{(j)}_{12}\\Y^{(j)}_{22}
\end{array}\right).
\label{Y2C-SQRT-T}
\end{eqnarray}

\subsection{The asymptotics of $Y^{(6)}(x,t)$ for $x \rightarrow \infty$ and $t \rightarrow \infty$  along $x=k \sqrt{t}$
\label{Y6-Asymptotic}}
In this case $k>0$. 
First let us assume $k \neq 1$.
By Lemma \ref{Y-RHP}, we know $Y^{(6)}_{11} \rightarrow 1$ and $Y^{(6)}_{22} \rightarrow 1$.
Therefore, the approximate differential equations for $Y^{(6)}_{21}$ and $Y^{(6)}_{12}$ along $x=k \sqrt{t}$ are
\begin{eqnarray}
&&\frac{d}{dt} Y^{(6)}_{21}= \frac{(3-k^2)k}{2} \sqrt{t}  Y^{(6)}_{21}+\frac{k^2-2}{2}\mathrm{Ai}(t)+\frac{k}{2 \sqrt{t}}\mathrm{Ai}'(t),
\label{Y6-21EQ-apprx}\\
&&\frac{d}{dt} Y^{(6)}_{12}=\frac{(k^2-3)k}{2} \sqrt{t}  Y^{(6)}_{12} +\frac{k^2-2}{2}\mathrm{Ai}(t)-\frac{k}{2 \sqrt{t}}\mathrm{Ai}'(t),
\label{Y6-12EQ-apprx}
\end{eqnarray}
while the approximate differential equations for $Y^{(6)}_{21}$ and $Y^{(6)}_{12}$ for fixed but large $t$  are
\begin{eqnarray}
	&&\frac{d}{dx}Y^{(6)}_{21}=(t-x^2)Y^{(6)}_{21}+x \mathrm{Ai}(t)+\mathrm{Ai}'(t), \label{Y6-21EQ-apprx-X}\\
	&&\frac{d}{dx}Y^{(6)}_{12}=(-t+x^2)Y^{(6)}_{12}+x \mathrm{Ai}(t)-\mathrm{Ai}'(t). \label{Y6-12EQ-apprx-X}
\end{eqnarray}

\subsubsection{$Y^{(6)}_{21}$ }
{\bf \noindent Case $0<k<2$:}

In this case\footnote{By the result of Riemann-Hilbert problem, $Y^{(6)}_{21} \rightarrow 0$ and $Y^{(6)}_{22} \rightarrow 1$
are still true for $k=1$.}, the solution of (\ref{Y6-21EQ-apprx}) is
\begin{eqnarray}
&&Y^{(6)}_{21}=\Upsilon_1(k) e^{\frac{k}{3}(3-k^2) t^{3/2}} 
+e^{\frac{k}{3}(3-k^2) t^{3/2}} \int_{\infty}^t e^{\frac{k}{3}(k^2-3) s^{3/2}} \left(\frac{k^2-2}{2}\mathrm{Ai}(s)+\frac{k}{2 \sqrt{s}}\mathrm{Ai}'(s) \right) ds.
\label{Y6-21EQ-apprx-SOL}
\end{eqnarray}
The solution of (\ref{Y6-21EQ-apprx-X}) is
\begin{eqnarray}
&&Y^{(6)}_{21}=C_1(t) e^{t x-\frac{1}{3}x^3} 
+ e^{t x-\frac{1}{3}x^3} \int_{\sqrt{t}}^x e^{-t s+\frac{1}{3}s^3} \left(s \mathrm{Ai}(t)+\mathrm{Ai}'(t) \right) ds.
\label{Y6-21EQ-apprx-SOL-X}
\end{eqnarray}
(\ref{Y6-21EQ-apprx-SOL}) and (\ref{Y6-21EQ-apprx-SOL-X}) must coincide at $x=k \sqrt{t}$.
Therefore, we have
\begin{eqnarray}
C_1(t)-\Upsilon_1(k)&=&\int_{\infty}^t e^{\frac{k}{3}(k^2-3) s^{3/2}} \left(\frac{k^2-2}{2}\mathrm{Ai}(s)+\frac{k}{2 \sqrt{s}}\mathrm{Ai}'(s) \right) ds  \nonumber\\
&&-\int_{\sqrt{t}}^{k \sqrt{t}} e^{-t s+\frac{1}{3}s^3} \left(s \mathrm{Ai}(t)+\mathrm{Ai}'(t) \right) ds .\label{Y6-21-AsymptA}
\end{eqnarray}

Note that the right-side of (\ref{Y6-21-AsymptA}) is a solution of $\frac{\partial^2 }{\partial k \partial t}f(k,t)=0$.
So (\ref{Y6-21-AsymptA}) determines $C_1(t)$ and $\Upsilon_1(k)$ up to a constant.
But we know $\Upsilon_1( k)=0$ for $k \in (0, \sqrt{3}]$.
So we get
\begin{eqnarray}
&&C_1(t)=\int_{\infty}^t e^{-\frac{2}{3} s^{3/2}} \left(-\frac{1}{2}\mathrm{Ai}(s)+\frac{1}{2 \sqrt{s}}\mathrm{Ai}'(s) \right) ds
.\label{Y6-21-C1-SOL}
\end{eqnarray}
Taking the $t \rightarrow \infty$ limit of (\ref{Y6-21-AsymptA}) and considering (\ref{Y6-21-C1-SOL}), we obtain
\begin{eqnarray}
\Upsilon_1(k)=0,  \quad k\in (0, 2). \label{Y6-21-Upsilon1-SOL}
\end{eqnarray}
Therefore, the final result for $0<k<2$ is
\begin{eqnarray}
Y^{(6)}_{21} \approx e^{\frac{k}{3}(3-k^2) t^{3/2}} \int_{\infty}^t e^{\frac{k}{3}(k^2-3) s^{3/2}} \left(\frac{k^2-2}{2}\mathrm{Ai}(s)+\frac{k}{2 \sqrt{s}}\mathrm{Ai}'(s) \right) ds. \nonumber
\end{eqnarray}

{\bf \noindent Case $k=2$:}

Similar to the case $0<k<2$, we can also derive
$$\Upsilon_1(2)=0 .$$
Thus
\begin{eqnarray}
Y^{(6)}_{21} \approx e^{-\frac{2}{3} t^{3/2}} \int_{\infty}^t e^{\frac{2}{3} s^{3/2}} \left(\mathrm{Ai}(s)+\frac{1}{\sqrt{s}}\mathrm{Ai}'(s) \right) ds. \nonumber
\end{eqnarray}

{\bf \noindent Case $k>2$:}

In this case, the solution of (\ref{Y6-21EQ-apprx}) is
\begin{eqnarray}
&&Y^{(6)}_{21}=\Upsilon_1(k) e^{\frac{k}{3}(3-k^2) t^{3/2}} 
+e^{\frac{k}{3}(3-k^2) t^{3/2}} \int_{t_A}^t e^{\frac{k}{3}(k^2-3) s^{3/2}} \left(\frac{k^2-2}{2}\mathrm{Ai}(s)+\frac{k}{2 \sqrt{s}}\mathrm{Ai}'(s) \right) ds.
\label{Y6-21EQ-apprx-SOL-simp}
\end{eqnarray}
Clearly, the first term can be neglected. 
Thus for $k>2$,
\begin{eqnarray}
&&Y^{(6)}_{21} \approx e^{\frac{k}{3}(3-k^2) t^{3/2}} \int_{t_A}^t e^{\frac{k}{3}(k^2-3) s^{3/2}} \left(\frac{k^2-2}{2}\mathrm{Ai}(s)+\frac{k}{2 \sqrt{s}}\mathrm{Ai}'(s) \right) ds.
\label{Y6-21EQ-apprx-SOL-simpF}
\end{eqnarray}
Note $t_A$ is a fixed arbitrary real number.

\begin{remark}
As $k \rightarrow \infty$,  (\ref{Y6-21EQ-apprx-SOL-simpF}) is consistent with
$Y^{(6)}_{21} \xlongrightarrow {x \rightarrow \infty} \frac{u(t)}{x}$.
\end{remark}

\subsubsection{$Y^{(6)}_{12}$}

For all $k>0$, (\ref{Y6-12EQ-apprx}) has the solution
\begin{eqnarray}
&&Y^{(6)}_{12}=\Upsilon_2(k) e^{\frac{k}{3}(k^2-3) t^{3/2}} 
+e^{\frac{k}{3}(k^2-3) t^{3/2}} \int_{\infty}^t e^{\frac{k}{3}(3-k^2) s^{3/2}} \left(\frac{k^2-2}{2}\mathrm{Ai}(s)-\frac{k}{2 \sqrt{s}}\mathrm{Ai}'(s) \right) ds.
\label{Y6-12EQ-apprx-SOL}
\end{eqnarray}
Since (\ref{Y6-12EQ-apprx-SOL}) has to approach $0$ as $t \rightarrow \infty$,
we have $\Upsilon_2(k)=0$ for $k \ge \sqrt{3}$.

The solution of (\ref{Y6-12EQ-apprx-X}) is
\begin{eqnarray}
Y^{(6)}_{12}=C_2(t) e^{-t x+\frac{1}{3}x^3} 
+e^{-t x+\frac{1}{3}x^3} \int_{\sqrt{t}}^x e^{t s-\frac{1}{3}s^3} 
\left(s \mathrm{Ai}(s)-\mathrm{Ai}'(s) \right) ds.
\label{Y6-12EQ-apprx-SOL-X}
\end{eqnarray}

By the consistence of (\ref{Y6-12EQ-apprx-SOL}) and (\ref{Y6-12EQ-apprx-SOL-X}) for $x=k \sqrt{t}$,
we get
\begin{eqnarray}
C_2(t)-\Upsilon_2(k) &=&\int_{\infty}^t e^{\frac{k}{3}(3-k^2) s^{3/2}} \left(\frac{k^2-2}{2}\mathrm{Ai}(s)-\frac{k}{2 \sqrt{s}}\mathrm{Ai}'(s) \right) ds \nonumber\\
&&-\int_{\sqrt{t}}^{k\sqrt{t}} e^{t s-\frac{1}{3}s^3} 
\left(s \mathrm{Ai}(s)-\mathrm{Ai}'(s) \right) ds.
\label{C2-Upsilon2-A}
\end{eqnarray}
Unlike the case of $Y^{(6)}_{21}$, we can not recklessly take the $k \rightarrow 1$ limit of (\ref{C2-Upsilon2-A}).

Let us fix $k$, $k>1$. Consider the asymptotics of (\ref{C2-Upsilon2-A}) as $t \rightarrow \infty$.
The first term of the right-side of (\ref{C2-Upsilon2-A})  can be neglected 
since it is exponentially small for $t \rightarrow \infty$.
So we have
\begin{eqnarray}
C_2(t)-\Upsilon_2(k) 
&\approx& -\int_{\sqrt{t}}^{k \sqrt{t}}
e^{t s-\frac{1}{3} s^3} (\mathrm{Ai(t)} s-\mathrm{Ai}'(t)) ds \nonumber \\
&=& -\int_{0}^{k-1}
e^{\frac{2}{3}t^{3/2}-t^{3/2} r^2-\frac{1}{3}t^{3/2} r^3} \left(\mathrm{Ai(t)}\sqrt{t} (1+r)-\mathrm{Ai}'(t) \right) \sqrt{t} dr
\nonumber\\
&\approx& -\int_{0}^{\infty}  e^{-t^{3/2} r^2} \left(1-\frac{1}{3}t^{3/2} r^3+ \cdots \right)
e^{\frac{2}{3}t^{3/2} } \left(\mathrm{Ai(t)}\sqrt{t} (1+r)-\mathrm{Ai}'(t) \right) \sqrt{t} dr \nonumber\\
&=&-\frac{1}{2}-\frac{1}{12 \sqrt{\pi}} t^{-\frac{3}{4}}+\frac{35}{1728 \sqrt{\pi}}t^{-\frac{9}{4}}+\cdots.
\label{C2Up-Inf-1}
\end{eqnarray}
By the condition that $\Upsilon_2(k)=0$ for $k \ge \sqrt{3}$, we get
\begin{eqnarray}
&&\Upsilon_2(k)=0, \quad k>1, \nonumber\\
&&C_2(t)=-\frac{1}{2}-\frac{1}{12 \sqrt{\pi}} t^{-\frac{3}{4}}+\frac{35}{1728 \sqrt{\pi}}t^{-\frac{9}{4}}+\cdots. \label{C2-inf-1A}
\end{eqnarray}

Similarly,  for $0<k<1$, we have
\begin{eqnarray}
C_2(t)-\Upsilon_2(k)=\frac{1}{2}-\frac{1}{12 \sqrt{\pi}} t^{-\frac{3}{4}}+\frac{35}{1728 \sqrt{\pi}}t^{-\frac{9}{4}}+\cdots.
\label{ C2Up-Inf-2 }
\end{eqnarray}

Therefore,
$$\Upsilon_2(k)=-1, \quad 0<k<1. $$

Thus,
\begin{eqnarray}
\Upsilon_2(k)=
\left\{ \begin{array}{ll}
0,& k>1,\\ -1, &0<k<1.
\end{array}
\right. \label{Upsilon2-SOL}
\end{eqnarray}

One should not try to get the expression of $C_2(t)$ by setting $k=1$ in (\ref{C2-Upsilon2-A}), 
since $\Upsilon_2(1)$ is not defined.

We claim that
\begin{eqnarray}
C_2(t)=-\frac{1}{2}-\frac{1}{2}\int_{\infty}^t e^{\frac{2}{3} s^{3/2}} \left( \mathrm{Ai}(s)+\frac{1}{\sqrt{s}} \mathrm{Ai}'(s)\right) ds. \label{C2-inf-SOL}
\end{eqnarray}

\begin{proof}
Let $k_P=1+\epsilon-\frac{\epsilon^2}{6}+\frac{5}{72} \epsilon^3+\cdots$
and $k_N=1-\epsilon-\frac{\epsilon^2}{6}-\frac{5}{72} \epsilon^3+\cdots$
be the two roots of $\frac{1}{3} (3-k^2) k=\frac{2}{3}-\epsilon^2$, $\epsilon>0$.
Then
\begin{eqnarray}
C_2(t)+\frac{1}{2}&=& \frac{1}{2} \left(C_2(t)-\Upsilon_2(k_P) + C_2(t)-\Upsilon_2(k_N)\right) \nonumber\\
&=& \int_{\infty}^t e^{(\frac{2}{3}-\epsilon^2) s^{3/2}} 
\left(-\frac{1}{2} \left( \mathrm{Ai}(s)+\frac{\mathrm{Ai}'(s)}{\sqrt{s}} \right)+
\frac{1}{12} \epsilon^2 \left( 4 \mathrm{Ai}(s)+\frac{\mathrm{Ai}'(s)}{\sqrt{s}}\right)+\cdots \right) ds \nonumber \\
&&-\int_{\sqrt{t}}^{k_P \sqrt{t}}\cdots ds -\int_{\sqrt{t}}^{k_N \sqrt{t}}\cdots ds. \nonumber
\end{eqnarray}
Since $4 \mathrm{Ai}(s)+\frac{\mathrm{Ai}'(s)}{\sqrt{s}}=e^{-\frac{2}{3}s^{3/2}} 
\left(\frac{3}{2 \sqrt{\pi}} s^{-\frac{1}{4}}-\frac{9}{32 \sqrt{\pi}} s^{-\frac{7}{4}}+\cdots\right) $,
we only need to show 
$\lim \limits_{\epsilon \rightarrow 0} \epsilon^2 \int_{\infty}^t e^{-\epsilon^2 s^{3/2}} s^{-\frac{1}{4}}ds =0 $.
But $\epsilon^2 \int_{\infty}^t e^{-\epsilon^2 s^{3/2}} s^{-\frac{1}{4}}ds
=\frac{2}{3} \epsilon \sqrt{\pi} \left(\mathrm{Erf}(\epsilon t^{\frac{3}{4}})-1 \right) $.
Therefore, we set $\epsilon=\epsilon(t)$ smaller than $t^{-n}$ for any $n>\frac{3}{4}$, for example, $\epsilon(t)=e^{-t}$.
Then  the terms $\int_{\sqrt{t}}^{k_P \sqrt{t}} \cdots ds$ , $\int_{\sqrt{t}}^{k_N \sqrt{t}}\cdots ds$
and  $ \int_{\infty}^t \frac{1}{12} \epsilon^2 \left( 4 \mathrm{Ai}(s)+\frac{\mathrm{Ai}'(s)}{\sqrt{s}}\right) ds$
can  all be neglected. So (\ref{C2-inf-SOL}) is obtained.
\end{proof}

\begin{remark} \label{rem-Y6-12}
	For $k \rightarrow \infty$,  (\ref{Y6-12EQ-apprx-SOL}) with (\ref{Upsilon2-SOL}) is consistent 
	with $Y^{(6)}_{12} \approx -\frac{u}{x}$.
\end{remark}

Altogether, (\ref{Y6-12EQ-apprx-SOL}) with (\ref{Upsilon2-SOL}) is convenient for estimating $Y^{(6)}_{12}$ on $x=k \sqrt{t}$,
$k \neq 1$, and (\ref{Y6-12EQ-apprx-SOL-X}) with (\ref{C2-inf-SOL}) is proper for estimating $Y^{(6)}_{12}$
near $x=\sqrt{t}$.

\subsection{The asymptotics of $Y^{(3)}(x,t)$ for $x \rightarrow -\infty$ and $t \rightarrow \infty$  along $x=k \sqrt{t}$}
The behaviour of $Y^{(3)}$ are similar to $Y^{(6)}$ presented in Section \ref{Y6-Asymptotic}.

In this case, $k<0$. 
$x=-\sqrt{t}$ is the dividing line.
By Lemma \ref{Y-RHP}, we know $Y^{(3)}_{11} \rightarrow 1$ and $Y^{(3)}_{22} \rightarrow 1$.
Therefore, the approximate differential equations for $Y^{(3)}_{21}$ and $Y^{(3)}_{12}$ along $x=k \sqrt{t}$ are
\begin{eqnarray}
	&&\frac{d}{dt} Y^{(3)}_{21}= \frac{(3-k^2)k}{2} \sqrt{t}  Y^{(3)}_{21}+\frac{k^2-2}{2}\mathrm{Ai}(t)+\frac{k}{2 \sqrt{t}}\mathrm{Ai}'(t),
	\label{Y3-21EQ-apprx}\\
	&&\frac{d}{dt} Y^{(3)}_{12}=\frac{(k^2-3)k}{2} \sqrt{t}  Y^{(3)}_{12} +\frac{k^2-2}{2}\mathrm{Ai}(t)-\frac{k}{2 \sqrt{t}}\mathrm{Ai}'(t).
	\label{Y3-12EQ-apprx}
\end{eqnarray}

The approximate differential equations for $Y^{(3)}_{21}$ and $Y^{(3)}_{12}$ for fixed but large $t$  are
\begin{eqnarray}
	&&\frac{d}{dx}Y^{(3)}_{21}=(t-x^2)Y^{(3)}_{21}+x \mathrm{Ai}(t)+\mathrm{Ai}'(t),  \label{Y3-21EQ-apprx-X}\\
	&&\frac{d}{dx}Y^{(3)}_{12}=(-t+x^2)Y^{(3)}_{12}+x \mathrm{Ai}(t)-\mathrm{Ai}'(t). \label{Y3-12EQ-apprx-X}
\end{eqnarray}

\subsubsection{$Y^{(3)}_{12}$ }
The behaviour of $Y^{(3)}_{12}$ is similar to $Y^{(6)}_{21}$.

{\bf \noindent Case $-2<k<0$:}

In this case, the solution of (\ref{Y3-12EQ-apprx}) is
\begin{eqnarray}
&&Y^{(3)}_{21}=\Upsilon_3(k) e^{\frac{k}{3}(k^2-3) t^{3/2}} 
+e^{\frac{k}{3}(k^2-3) t^{3/2}} \int_{\infty}^t e^{\frac{k}{3}(3-k^2) s^{3/2}} \left(\frac{k^2-2}{2}\mathrm{Ai}(s)-\frac{k}{2 \sqrt{s}}\mathrm{Ai}'(s) \right) ds.
\label{Y3-12EQ-apprx-SOL}
\end{eqnarray}
The solution of (\ref{Y3-12EQ-apprx-X}) is
\begin{eqnarray}
&&Y^{(3)}_{12}=C_3(t) e^{-t x+\frac{1}{3}x^3} 
+ e^{-t x+\frac{1}{3}x^3} \int_{-\sqrt{t}}^x e^{t s-\frac{1}{3}s^3} \left(s \mathrm{Ai}(t)-\mathrm{Ai}'(t) \right) ds.
\label{Y3-12EQ-apprx-SOL-X}
\end{eqnarray}
So we have
\begin{eqnarray}
C_3(t)-\Upsilon_3(k)&=&\int_{\infty}^t e^{\frac{k}{3}(3-k^2) s^{3/2}} \left(\frac{k^2-2}{2}\mathrm{Ai}(s)-\frac{k}{2 \sqrt{s}}\mathrm{Ai}'(s) \right) ds \nonumber\\
&&-\int_{-\sqrt{t}}^{k \sqrt{t}} e^{t s-\frac{1}{3}s^3} \left(s \mathrm{Ai}(t)-\mathrm{Ai}'(t) \right) ds .\label{Y3-12-AsymptA}
\end{eqnarray}

By $\Upsilon_3( k)=0$ for $k \in [-\sqrt{3},0)$,
we get
\begin{eqnarray}
C_3(t)=\int_{\infty}^t e^{-\frac{2}{3} s^{3/2}} \left(-\frac{1}{2}\mathrm{Ai}(s)+\frac{1}{2 \sqrt{s}}\mathrm{Ai}'(s) \right) ds
.\label{Y3-12-C3-SOL}
\end{eqnarray}
Taking the $t \rightarrow \infty$ limit of (\ref{Y3-12-AsymptA}) and considering (\ref{Y3-12-C3-SOL}), we obtain
\begin{eqnarray}
\Upsilon_3(k)=0,  \quad k\in (-2, 0). \label{Y3-12-Upsilon3-SOL}
\end{eqnarray}
Therefore, the final result for $-2<k<0$ is
\begin{eqnarray}
Y^{(3)}_{12} \approx e^{\frac{k}{3}(k^2-3) t^{3/2}} \int_{\infty}^t e^{\frac{k}{3}(3-k^2) s^{3/2}} \left(\frac{k^2-2}{2}\mathrm{Ai}(s)-\frac{k}{2 \sqrt{s}}\mathrm{Ai}'(s) \right) ds. \nonumber
\end{eqnarray}

{\bf \noindent Case $k=-2$:}

\begin{eqnarray}
Y^{(3)}_{12} \approx e^{-\frac{2}{3} t^{3/2}} \int_{\infty}^t e^{\frac{2}{3} s^{3/2}} \left(\mathrm{Ai}(s)+\frac{1}{\sqrt{s}}\mathrm{Ai}'(s) \right) ds. \nonumber
\end{eqnarray}

{\bf \noindent Case $k<-2$:}

In this case, the solution of (\ref{Y3-12EQ-apprx}) is
\begin{eqnarray}
&&Y^{(3)}_{12}=\Upsilon_3(k) e^{\frac{k}{3}(k^2-3) t^{3/2}} 
+e^{\frac{k}{3}(k^2-3) t^{3/2}} \int_{t_A}^t e^{\frac{k}{3}(3-k^2) s^{3/2}} \left(\frac{k^2-2}{2}\mathrm{Ai}(s)-\frac{k}{2 \sqrt{s}}\mathrm{Ai}'(s) \right) ds.
\label{Y3-12EQ-apprx-SOL-simp}
\end{eqnarray}
Thus, for $k<-2$, $Y^{(3)}_{12}$ is approximated by
\begin{eqnarray}
&&Y^{(3)}_{12} \approx e^{\frac{k}{3}(k^2-3) t^{3/2}} \int_{t_A}^t e^{\frac{k}{3}(3-k^2) s^{3/2}} \left(\frac{k^2-2}{2}\mathrm{Ai}(s)-\frac{k}{2 \sqrt{s}}\mathrm{Ai}'(s) \right) ds.
\label{Y3-12EQ-apprx-SOL-simpF}
\end{eqnarray}

\subsubsection{$Y^{(3)}_{21}$}

For all $k<0$, (\ref{Y3-21EQ-apprx}) has the solution
\begin{eqnarray}
&&Y^{(3)}_{21}=\Upsilon_4(k) e^{\frac{k}{3}(3-k^2) t^{3/2}} 
+e^{\frac{k}{3}(3-k^2) t^{3/2}} \int_{\infty}^t e^{\frac{k}{3}(k^2-3) s^{3/2}} \left(\frac{k^2-2}{2}\mathrm{Ai}(s)+\frac{k}{2 \sqrt{s}}\mathrm{Ai}'(s) \right) ds.
\label{Y3-21EQ-apprx-SOL}
\end{eqnarray}

The solution of (\ref{Y3-21EQ-apprx-X}) is
\begin{eqnarray}
&&Y^{(3)}_{21}=C_4(t) e^{t x-\frac{1}{3}x^3} 
+e^{t x-\frac{1}{3}x^3} \int_{-\sqrt{t}}^x e^{-t s+\frac{1}{3}s^3} 
\left(s \mathrm{Ai}(t)+\mathrm{Ai}'(t) \right) ds.
\label{Y3-21EQ-apprx-SOL-X}
\end{eqnarray}

By the consistence of (\ref{Y3-21EQ-apprx-SOL}) and (\ref{Y3-21EQ-apprx-SOL-X}) for $x=k \sqrt{t}$,
we get
\begin{eqnarray}
C_4(t)-\Upsilon_4(k) &=&\int_{\infty}^t e^{\frac{k}{3}(k^2-3) s^{3/2}} \left(\frac{k^2-2}{2}\mathrm{Ai}(s)+\frac{k}{2 \sqrt{s}}\mathrm{Ai}'(s) \right) ds \nonumber\\
&&-\int_{-\sqrt{t}}^{k\sqrt{t}} e^{-t s+\frac{1}{3}s^3} 
\left(s \mathrm{Ai}(t)+\mathrm{Ai}'(t) \right) ds.
\label{C4-Upsilon4-A}
\end{eqnarray}

Let us fix $k$, $k<-1$.
Near $t=\infty$, the first term of the right-side of (\ref{C4-Upsilon4-A})  can be neglected since it is exponentially small.
Therefore, we obtain
\begin{eqnarray}
C_4(t)-\Upsilon_4(k) 
&\approx& -\int_{-\sqrt{t}}^{k \sqrt{t}}
e^{-t s+\frac{1}{3} s^3} (\mathrm{Ai(t)} s+\mathrm{Ai}'(t)) ds \nonumber \\
&=& -\int_{0}^{k+1}
e^{\frac{2}{3}t^{3/2}-t^{3/2} r^2+\frac{1}{3}t^{3/2} r^3} \left(\mathrm{Ai(t)}\sqrt{t} (-1+r)+\mathrm{Ai}'(t) \right) \sqrt{t} dr
\nonumber\\
&\approx& -\int_{0}^{-\infty}  e^{-t^{3/2} r^2} \left(1+\frac{1}{3}t^{3/2} r^3+ \cdots \right)
e^{\frac{2}{3}t^{3/2} } \left(\mathrm{Ai(t)}\sqrt{t} (r-1)+\mathrm{Ai}'(t) \right) \sqrt{t} dr \nonumber\\
&=&-\frac{1}{2}-\frac{1}{12 \sqrt{\pi}} t^{-\frac{3}{4}}+\frac{35}{1728 \sqrt{\pi}}t^{-\frac{9}{4}}+\cdots.
\label{C4Up-Inf-1}
\end{eqnarray}
By the condition that $\Upsilon_4(k)=0$ for $k \le -\sqrt{3}$, we get
\begin{eqnarray}
&&\Upsilon_4(k)=0, \quad k<-1, \nonumber\\
&&C_4(t)=-\frac{1}{2}-\frac{1}{12 \sqrt{\pi}} t^{-\frac{3}{4}}+\frac{35}{1728 \sqrt{\pi}}t^{-\frac{9}{4}}+\cdots. \label{C4-inf-1A}
\end{eqnarray}

Similarly,  for $-1<k<0$, we have
\begin{eqnarray}
C_4(t)-\Upsilon_4(k)=\frac{1}{2}-\frac{1}{12 \sqrt{\pi}} t^{-\frac{3}{4}}+\frac{35}{1728 \sqrt{\pi}}t^{-\frac{9}{4}}+\cdots.
\label{ C4Up-Inf-2 }
\end{eqnarray}

Therefore,
$$\Upsilon_4(k)=-1, \quad -1<k<0. $$

Thus,
\begin{eqnarray}
\Upsilon_4(k)=
\left\{ \begin{array}{ll}
0,& k<-1,\\ -1, &-1<k<0.
\end{array}
\right. \label{Upsilon4-SOL}
\end{eqnarray}

We claim that
\begin{eqnarray}
C_4(t)=-\frac{1}{2}-\frac{1}{2}\int_{\infty}^t e^{\frac{2}{3} s^{3/2}} \left( \mathrm{Ai}(s)+\frac{1}{\sqrt{s}} \mathrm{Ai}'(s)\right) ds. \label{C4-inf-SOL}
\end{eqnarray}
The proof is similar to the case in Section \ref{Y6-Asymptotic}, and thus we omit it.

\section{Proof of Theorem \ref{theorem-UniqueSolution}}
By (\ref{Rumanov-Conj}), we only need to prove that $\mathcal{F}(x,t)$ 
satisfies the  Bloemendal-Vir\'{a}g boundary (\ref{BV-boundary}) and  
that $\mathcal{F}(x,t)$ is bounded at  $x^2+t^2 =\infty$.

\begin{center}
\begin{minipage}[c]{12 cm}
\begin{center}
\begin{tikzpicture}
\draw[->, thick] (-3.5,0) -- (3.5,0) node[right] {$x$};
\draw[->, thick] (0,-3.5) -- (0,3.5) node[above] {$t$};
\draw[dashed] (0,0) circle (3cm);
\draw[dashed,light-gray]  (1,2.828)--(1,-2.828);
\draw[dashed,light-gray]  (-1,2.828)--(-1,-2.828);
\draw[dashed,light-gray]  (2.828,1)--(-2.828,1);
\draw[dashed,light-gray]  (2.828,-1)--(-2.828,-1);
\node at  (2.828+0.15,1+0.15) {A};\node at  (1+0.15,2.828+0.15) {B}; 
\node at  (-1-0.15, 2.828+0.15) {C};\node at  (-2.828-0.15,1+0.15) {D};
\node at  (-2.828-0.15,-1-0.15) {E};\node at  (-1-0.15, -2.828-0.15) {F};
\node at  (1+0.15,-2.828-0.15) {G}; \node at  (2.828+0.15,-1-0.15) {H};
\end{tikzpicture}
\end{center}
\end{minipage}

\begin{minipage}[c]{12 cm}
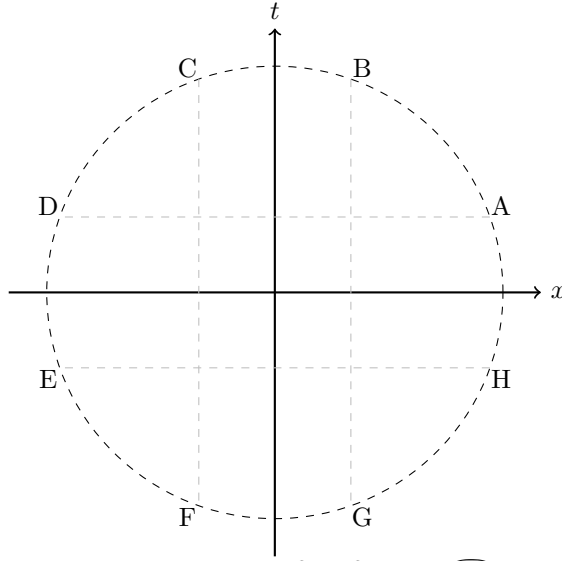
Figure 10. Diagram of the boundary $x^2+t^2=\infty$.
$\wideparen{AB}$: $x \rightarrow \infty$ and  $t \rightarrow \infty$;
$\wideparen{BC}$: fixed $x$ and  $t \rightarrow \infty$;
$\wideparen{CD}$: $x \rightarrow -\infty$ and  $t \rightarrow \infty$;
$\wideparen{DE}$: $x \rightarrow -\infty$  and fixed $t$;
$\wideparen{EF}$: $x \rightarrow -\infty$ and  $t \rightarrow -\infty$;
$\wideparen{FG}$: fixed $x$ and  $t \rightarrow -\infty$;
$\wideparen{GH}$: $x \rightarrow \infty$ and  $t \rightarrow -\infty$;
$\wideparen{HA}$: $x \rightarrow \infty$  and fixed $t$.
\end{minipage}
\end{center}

By considering  $\mathcal{F}(x,t)$ on the boundary $\wideparen{AB}$, \cite{GIKM} proved 
$$
\left( \begin{array}{c}
\mathcal{F}_0(x,t)\\ \mathcal{G}_0(x,t)
\end{array} \right)
= \mathrm{i}  \left( \begin{array}{c}
\Psi^{(6)}_{012}(x,t) \\  \Psi^{(6)}_{022}(x,t)
\end{array} \right),
$$
where $\mathcal{F}_0$ and $\mathcal{G}_0$ are defined by (\ref{psi0-DEF}), 
and $\Psi_0^{(6)}$ is a canonical wave solution of Painlev\'{e} II.

By (\ref{GIKM-F}), we get
\begin{eqnarray}
\mathcal{F}(x,t)=\kappa u^\frac{1}{2} 
\left[u^{-1} \left(\frac{1+q_2}{2}x-\alpha \right)Y_{12}^{(6)}(x,t)+Y_{22}^{(6)}(x,t) \right].
\label{GIKM-F-P}
\end{eqnarray}
Formula (\ref{GIKM-F-P}) is proper for $x \ge 0$.

The expression of $\mathcal{F}(x,t)$ for $x \le 0$ has also been given by \cite{GIKM}
\begin{eqnarray}
\mathcal{F}(x,t)=-\kappa u^\frac{1}{2} e^{\frac{x^3}{3}-x t}
\left[u^{-1} \left(\frac{1+q_2}{2}x-\alpha \right)Y_{11}^{(3)}(x,t)+Y_{21}^{(3)}(x,t) \right]. \label{GIKM-F-N}
\end{eqnarray}
Note $a=1$ has been applied to the expression in  \cite{GIKM}.
By (\ref{overlap}) and (\ref{Y-DEF}), (\ref{GIKM-F-P}) and (\ref{GIKM-F-N}) coincide on $x=0$.

\subsection{Boundedness of $\mathcal{F}(x,t)$ to $c_1=0$ and $c_2=0$ \label{c0c0}}
Let us investigate  $\mathcal{F}(x,t)$ at the boundary of 
$\wideparen{CD}$, i.e., $x \rightarrow -\infty$ and $t \rightarrow \infty$ simultaneously.

It is convenient to  study the case along $x=k \sqrt{t}$, $k \in (-\sqrt{3}, 0)$.
In this case,  $e^{\frac{x^3}{3}-x t}$ is very large:
\begin{eqnarray}
e^{\frac{x^3}{3}-x t}=e^{\frac{1}{3}k (k^2-3) t^{3/2}}. \nonumber
\end{eqnarray}
The largest case is $k=-1$, i.e.,
$$ e^{\frac{x^3}{3}-x t}=e^{\frac{2}{3} t^{3/2}}.$$

By (\ref{Y3-21EQ-apprx-SOL}), 
we know $-\kappa u^\frac{1}{2} e^{\frac{x^3}{3}-x t} Y_{21}^{(3)}(x,t)$ only contributes a finite term $-\Upsilon_4(k)$.
Thus we can temporally neglect it.

Also we know  for $t>0$, $Y_{11}^{(3)}(x,t) \xlongrightarrow {x \rightarrow -\infty} 1$.
Then by (\ref{kappa-pinf}), we have
\begin{eqnarray}
\lim_{x=k \sqrt{t}, k \in (-\sqrt{3},0), t \rightarrow \infty} \mathcal{F}(x,t)
=-\Upsilon_4(k)+\lim_{k \in (-\sqrt{3},0), t \rightarrow \infty}  
-e^{\frac{1}{3}k (k^2-3) t^{3/2}} u(t)^{-1} \left(k \sqrt{t} \frac{1+q_2(t)}{2}-\alpha(t) \right) .\nonumber
\end{eqnarray}

Since $\frac{1}{3}k (k^2-3) $ varies in $(0,\frac{2}{3}]$ for $k \in (-\sqrt{3},0)$
and  $u(t)^{-1}$ has order of $e^{\frac{2}{3} t^{3/2}}$,
we have to require  $k \sqrt{t} \frac{1+q_2(t)}{2}-\alpha(t)$ has order of $ e^{-\frac{4}{3} t^{3/2}}$.
By (\ref{PinfSolution-A}) and (\ref{PinfSolution}), we find it is only possible for $c_1=c_2=0$.

\subsection{$c_1=c_2=0$ to boundedness of $F(\beta=6; x,t)$ at $x^2+t^2=\infty$}

\subsubsection{On $\wideparen{HA}$}
In this case,  $t$ is fixed, $x \rightarrow \infty$,
$$\mathcal{F}(x,t)=\kappa u^\frac{1}{2} 
\left[u^{-1} \left(\frac{1+q_2}{2}x-\alpha \right)Y_{12}^{(6)}(x,t)+Y_{22}^{(6)}(x,t) \right].$$
Recall $Y_{12}^{(6)} \rightarrow -\frac{u(t)}{x}$ and $Y_{22}^{(6)} \rightarrow 1$ in this case,
we get 
\begin{eqnarray}
\lim_{x \rightarrow \infty}\mathcal{F}(x,t)= -\kappa(t) u(t)^{\frac{1}{2}} \frac{q_2(t)+1}{2}
+\kappa(t) u(t)^{\frac{1}{2}}=\frac{1}{2}\kappa(t) u(t)^{\frac{1}{2}} (1-q_2(t) ). \label{F-HA}
\end{eqnarray} 
It is straightforward to verify that (\ref{F-HA}) is the same as (\ref{GIKM-F6}).

\subsubsection{On $\wideparen{AB}$}
$$\mathcal{F}(x,t)=\kappa u^\frac{1}{2} 
\left[u^{-1} \left(\frac{1+q_2}{2}x-\alpha \right)Y_{12}^{(6)}(x,t)+Y_{22}^{(6)}(x,t) \right].$$

By $Y^{(6)}_{22}(x,t) \rightarrow 1$ on $\wideparen{AB}$,
we know
\begin{eqnarray}
\lim_{x \rightarrow \infty, t\rightarrow \infty } \kappa(t) u(t)^\frac{1}{2} Y^{(6)}_{22}(x,t) =1.
\label{F-AB-1}
\end{eqnarray}

By $Y^{(6)}_{12}(x,t) \rightarrow 0$ on $\wideparen{AB}$,
we have
\begin{eqnarray}
\lim_{x \rightarrow \infty, t\rightarrow \infty } -\kappa(t) u(t)^{-\frac{1}{2}} \alpha(t) Y^{(6)}_{12}(x,t) =0.
\label{F-AB-2}
\end{eqnarray}

Then we will show
\begin{eqnarray}
\lim_{x \rightarrow \infty, t\rightarrow \infty } \kappa(t) u(t)^{-\frac{1}{2}} \frac{q_2(t)+1}{2} x Y^{(6)}_{12}(x,t) =0.
\label{F-AB-3}
\end{eqnarray}

To prove (\ref{F-AB-3}), we divide the problem into $2$ cases\footnote{The division is at liberty.
For example, for given  $\epsilon>0$,  any
division  of $x \ge  t^{\frac{1}{2}+\epsilon}$ and $x \le  t^{\frac{1}{2}+\epsilon}$ works.}:
(1) $x \ge  t$;
(2) $x \le t $.
In the case (1),  by the Riemann-Hilbert problem of $Y$, 
it is easy to show   $Y^{(6)}_{12} = o(1) \times x^{-1}$ , and thus (\ref{F-AB-3}) is true.
In the case (2),
\begin{eqnarray}
\lim_{x \rightarrow \infty, t\rightarrow \infty }\left| \kappa(t) u(t)^{-\frac{1}{2}} \frac{q_2(t)+1}{2} x Y^{(6)}_{12}(x,t) \right|
\le \lim_{x \rightarrow \infty, t\rightarrow \infty }
\left|\kappa(t) u(t)^{-\frac{1}{2}} \frac{q_2(t)+1}{2}\times t \times Y^{(6)}_{12}(x,t)\right|.
\nonumber
\end{eqnarray}
Considering $Y^{(6)}_{12}\rightarrow 0$ for $x>0$, we know (\ref{F-AB-3}) is also true in this case.
Thus (\ref{F-AB-3}) is proved.

Gathering (\ref{F-AB-1}), (\ref{F-AB-2}) and (\ref{F-AB-3}), we get
\begin{eqnarray}
\lim_{x \rightarrow \infty, t\rightarrow \infty }\mathcal{F}(x,t)=1. \nonumber
\end{eqnarray} 

\begin{remark}
	By Lemma \ref{Y-RHP}, (\ref{Y6-12EQ-apprx-SOL}) and Remark \ref{rem-Y6-12},
	we know $c_1=0$ is enough to guarantee 
	$\lim \limits_{x \rightarrow \infty, t\rightarrow \infty }\mathcal{F}(x,t)=1$ on $\wideparen{AB}$.
\end{remark}

\subsubsection{On $\wideparen{BC}$}
{\noindent \bf Case $x \ge 0$.}
$$\mathcal{F}(x,t)=\kappa u^\frac{1}{2} 
\left[u^{-1} \left(\frac{1+q_2}{2}x-\alpha \right)Y_{12}^{(6)}(x,t)+Y_{22}^{(6)}(x,t) \right].$$

In this case, $x$ is finite and $t$ is positive infinite.
So, $\left| Y^{(6)}_{12}(x,t) \right| \le 1$ and $Y^{(6)}_{22}(x,t) \rightarrow 1$.
Therefore,
\begin{eqnarray}
\lim_{t \rightarrow \infty} \mathcal{F}(x,t) &=&
\lim_{t \rightarrow \infty}   \kappa u^\frac{1}{2} 
\left[u^{-1} \left(\frac{1+q_2}{2}x-\alpha \right)Y_{12}^{(6)}(x,t)+Y_{22}^{(6)}(x,t) \right]  \nonumber \\
&=& 1, \quad  x \ge 0. \nonumber
\end{eqnarray}

{\noindent \bf Case $x \le 0$.}

\begin{eqnarray}
\mathcal{F}(x,t)=-\kappa(t) u(t)^\frac{1}{2} e^{\frac{x^3}{3}-x t}
\left[u(t)^{-1} \left(\frac{1+q_2(t)}{2}x-\alpha(t) \right)Y_{11}^{(3)}(x,t)+Y_{21}^{(3)}(x,t) \right]. \nonumber
\end{eqnarray}

By Lemma  \ref{Y-RHP}, we know  $Y^{(3)}_{11}(x,t) \approx 1$ and $Y^{(3)}_{21}(x,t) \approx -e^{x t-\frac{1}{3}x^3}$.
So we get 
\begin{eqnarray}
\lim_{t \rightarrow \infty} \mathcal{F}(x,t) &=&
\lim_{t \rightarrow \infty} -\kappa(t) u(t)^\frac{1}{2} e^{\frac{x^3}{3}-x t}
\left[u(t)^{-1} \left(\frac{1+q_2(t)}{2}x-\alpha(t) \right)Y_{11}^{(3)}(x,t)+Y_{21}^{(3)}(x,t) \right]. \nonumber\\
&=&1, \quad x \le 0 .\nonumber
\end{eqnarray}

\subsubsection{On $\wideparen{CD}$}
\begin{eqnarray}
\mathcal{F}(x,t)=-\kappa(t) u(t)^\frac{1}{2} e^{\frac{x^3}{3}-x t}
\left[u(t)^{-1} \left(\frac{1+q_2(t)}{2}x-\alpha(t) \right)Y_{11}^{(3)}(x,t)+Y_{21}^{(3)}(x,t) \right]. \nonumber
\end{eqnarray}

This is the most complicated case. 
The result is
\begin{eqnarray}
\lim_{x \rightarrow -\infty, t \rightarrow \infty} \mathcal{F}(x,t)=
\left\{ 
\begin{array}{ll}
0, & \frac{x+\sqrt{t}}{t^{-\frac{3}{4}}} \rightarrow -\infty,\\
1, & \frac{x+\sqrt{t}}{t^{-\frac{3}{4}}} \rightarrow \infty,\\
\in (0,1), & \mathrm{near } x=-\sqrt{t}.
\end{array}
\right. \nonumber
\end{eqnarray}
The corresponding proof is given in Appendix \ref{Appd-CD}.

\subsubsection{On $\wideparen{DE}$}
\begin{eqnarray}
\mathcal{F}(x,t)=-\kappa(t) u(t)^\frac{1}{2} e^{\frac{x^3}{3}-x t}
\left[u(t)^{-1} \left(\frac{1+q_2(t)}{2}x-\alpha(t) \right)Y_{11}^{(3)}(x,t)+Y_{21}^{(3)}(x,t) \right]. \nonumber
\end{eqnarray}

Since $t$ is finite, $Y_{11}^{(3)}(x,t)  \xlongrightarrow {x \rightarrow -\infty}1$ 
and $Y_{21}^{(3)}(x,t)  \xlongrightarrow {x \rightarrow -\infty}0$,
we obtain
$$ \lim_{x \rightarrow -\infty}\mathcal{F}(x,t)=0 .$$

\subsubsection{On $\wideparen{EF}$ \label{F-EF}}
\begin{eqnarray}
\mathcal{F}(x,t)&=&-\kappa(t) u(t)^\frac{1}{2} e^{\frac{x^3}{3}-x t}
\left[u(t)^{-1} \left(\frac{1+q_2(t)}{2}x-\alpha(t) \right)Y_{11}^{(3)}(x,t)+Y_{21}^{(3)}(x,t) \right]. \nonumber
\end{eqnarray}
Let us first evaluate $Y_{11}^{(3)}(x,t)$ and $Y_{21}^{(3)}(x,t)$ along the curve $x=-\sqrt{2} \sqrt{\Lambda_1+t}$
for $t \in [-\Lambda_1,0]$
with $\Lambda_1 \gg 1$.

Along the curve,
\begin{eqnarray}
\frac{d}{dt} \left( \begin{array}{c}Y^{(3)}_{11}(t)\\Y^{(3)}_{21}(t) \end{array} \right)
=\frac{1}{\sqrt{2} \sqrt{\Lambda_1+t}}
\left( \begin{array}{cc} 
u(t)^2 &u'(t) \\
-u'(t)&-t-u(t)^2 
\end{array} \right) 
\left( \begin{array}{c}Y^{(3)}_{11}(t)\\Y^{(3)}_{21}(t) \end{array} \right).
\label{ODE-Y-EF}
\end{eqnarray}

By (\ref{LODE-f-NEQ-S2}), we get
$$\ln \sqrt{Y^{(3)}_{11}(0)^2+Y^{(3)}_{21}(0)^2} \ge \ln \sqrt{Y^{(3)}_{11}(t)^2+Y^{(3)}_{21}(t)^2}
+\frac{1}{2} \int_t^0 \frac{ -s -| s+2 u(s)^2 |} {\sqrt{2} \sqrt{\Lambda_1+s}} ds .$$

But $Y^{(3)}_{11}(0) =Y^{(3)}_{11}(x=-\sqrt{2} \sqrt{\Lambda_1},t=0) \approx 1$ 
and  $Y^{(3)}_{21}(0)=Y^{(3)}_{21}(x=-\sqrt{2} \sqrt{\Lambda_1},t=0) \approx 0$.
So  
$$\sqrt{Y^{(3)}_{11}(t)^2+Y^{(3)}_{21}(t)^2}  < e^{\frac{1}{2} \int_0^t \frac{ -s -| s+2 u(s)^2 |} {\sqrt{2} \sqrt{\Lambda_1+s}} ds }.$$
Since $ s+2 u(s)^2  \approx 0$ for  large negative $s$,
we  assume  $e^{\frac{1}{2} \int_0^t \frac{ -s -| s+2 u(s)^2 |} {\sqrt{2} \sqrt{\Lambda_1+s}} ds }<1$ .
Thus $\sqrt{Y^{(3)}_{11}(t)^2+Y^{(3)}_{21}(t)^2} <1 $ for large negative $t$.
Also considering (\ref{NinfSolution-A})-(\ref{NinfSolution}), (\ref{kappa-Ninf}) and (\ref{u-expansion}), 
we immediately obtain
$$ \lim_{x \rightarrow -\infty, t \rightarrow -\infty}\mathcal{F}(x,t)=0.$$

\subsubsection{On $\wideparen{FG}$ \label{Case-FG}}
{\noindent \bf Case $x \le 0$.}
\begin{eqnarray}
\mathcal{F}(x,t)=-\kappa(t) u(t)^\frac{1}{2} e^{\frac{x^3}{3}-x t}
\left[u(t)^{-1} \left(\frac{1+q_2(t)}{2}x-\alpha(t) \right)Y_{11}^{(3)}(x,t)+Y_{21}^{(3)}(x,t) \right]. \nonumber
\end{eqnarray}
\begin{eqnarray}
&&\frac{d Y^{(3)}_{11}(x,t)}{dt}= -u(t) Y^{(3)}_{21}(x,t), \nonumber\\
&&\frac{d Y^{(3)}_{21}(x,t)}{dt}=x Y^{(3)}_{21}(x,t) -u(t) Y^{(3)}_{11}(x,t). \nonumber
\end{eqnarray}

At  $t=-\infty$, by (\ref{u-expansion}),  we know
\begin{eqnarray}
&&Y^{(3)}_{11}(x,t) =C_1(x) \times\left(1+\frac{x}{2 \sqrt{2} (-t)^{1/2}}+\frac{x^2}{16 t} 
-\frac{8 \sqrt{2}+9 \sqrt{2} x^3}{192 (-t)^{3/2}}+\cdots \right)\times e^{\frac{tx}{2} +\frac{1}{6} (x^2-2 t)^{3/2}}
+\cdots, \nonumber\\
&&Y^{(3)}_{21}(x,t) =C_1(x) \times\left(1-\frac{x}{2 \sqrt{2} (-t)^{1/2}}+\frac{x^2}{16 t} 
+\frac{-8 \sqrt{2}+9 \sqrt{2} x^3}{192 (-t)^{3/2}}+\cdots \right)\times e^{\frac{tx}{2} +\frac{1}{6} (x^2-2 t)^{3/2}}
+\cdots. \nonumber
\end{eqnarray}

Also considering (\ref{NinfSolution-A})-(\ref{NinfSolution}), (\ref{u-expansion}) and (\ref{kappa-Ninf}),
we finally get
$$\lim_{t \rightarrow -\infty} \mathcal{F}(x,t) =0, \quad x \le 0,$$
which, actually, has been proved in Section \ref{F-EF}.

{\noindent \bf Case $x \ge 0$.}
$$\mathcal{F}(x,t)=\kappa u^\frac{1}{2} 
\left[u^{-1} \left(\frac{1+q_2}{2}x-\alpha \right)Y_{12}^{(6)}(x,t)+Y_{22}^{(6)}(x,t) \right].$$
\begin{eqnarray}
&&\frac{d Y^{(6)}_{12}(x,t)}{dt}=-x Y^{(6)}_{12}(x,t) -u(t) Y^{(6)}_{22}(x,t), \nonumber\\
&&\frac{d Y^{(6)}_{22}(x,t)}{dt}= -u(t) Y^{(6)}_{12}(x,t). \nonumber
\end{eqnarray}

Similar to the case $x<0$, we get
\begin{eqnarray}
&&Y^{(6)}_{12}(x,t) =C_2(x) \times\left(1+\frac{x}{2 \sqrt{2} (-t)^{1/2}}+\frac{x^2}{16 t} 
-\frac{8 \sqrt{2}+9 \sqrt{2} x^3}{192 (-t)^{3/2}}+\cdots \right)\times e^{-\frac{tx}{2} +\frac{1}{6} (x^2-2 t)^{3/2}}
+\cdots, \nonumber \\
&&Y^{(6)}_{22}(x,t) =C_2(x) \times\left(1-\frac{x}{2 \sqrt{2} (-t)^{1/2}}+\frac{x^2}{16 t} 
+\frac{-8 \sqrt{2}+9 \sqrt{2} x^3}{192 (-t)^{3/2}}+\cdots \right)\times e^{-\frac{tx}{2} +\frac{1}{6} (x^2-2 t)^{3/2}}
+\cdots. \nonumber
\end{eqnarray}
Therefore,  we also have
$$\lim_{t \rightarrow -\infty} \mathcal{F}(x,t) =0, \quad x \ge 0 ,$$
which can also be inferred from Section \ref{F-GH}.

\subsubsection{On $\wideparen{GH}$ \label{F-GH}}
$$ \mathcal{F}(x,t)=  \kappa(t) \left( u(t)^{-\frac{1}{2}} \left( \frac{q_2(t)+1}{2} x-\alpha(t)\right) Y^{(6)}_{12}(x,t)
+u(t)^\frac{1}{2}Y^{(6)}_{22}(x,t) \right)  . $$

Along the curve $x=\sqrt{2} \sqrt{\Lambda_2+t}$,  $Y^{(6)}_{12}(x,t)$ and  $Y^{(6)}_{22}(x,t)$
satisfy
\begin{eqnarray}
&&\frac{d}{dt} \left( \begin{array}{c}Y^{(6)}_{12}(t)\\Y^{(6)}_{22}(t) \end{array} \right)
=\frac{1}{\sqrt{2} \sqrt{\Lambda_2+t}}
\left( \begin{array}{cc} 
-\frac{1}{2}(t+2 u(t)^2) &-u'(t) \\
u'(t)&\frac{1}{2}(t+2 u(t)^2) 
\end{array} \right) 
\left( \begin{array}{c}Y^{(6)}_{12}(t)\\Y^{(6)}_{22}(t) \end{array} \right).
\label{ODE-Y-GH}
\end{eqnarray}

By  (\ref{LODE-f-NEQ-S1}) and (\ref{LODE-f-NEQ-S2}),
\begin{eqnarray}
&&\ln \sqrt{Y^{(6)}_{12}(0)^2+ Y^{(6)}_{22}(0)^2} \le \ln \sqrt{Y^{(6)}_{12}(t)^2+ Y^{(6)}_{22}(t)^2} 
+\frac{1}{2} \int_t^0  \frac{|s+2 u(s)^2 |}{\sqrt{2} \sqrt{\Lambda_2+s}} ds, \nonumber\\
&&\ln \sqrt{Y^{(6)}_{12}(0)^2+ Y^{(6)}_{22}(0)^2} \ge \ln \sqrt{Y^{(6)}_{12}(t)^2+ Y^{(6)}_{22}(t)^2} 
-\frac{1}{2} \int_t^0  \frac{|s+2 u(s)^2 |}{\sqrt{2} \sqrt{\Lambda_2+s}} ds . \nonumber
\end{eqnarray}

Considering (\ref{u-expansion}), 
for a large $\Lambda_2$, we have
$$\int_t^0  \frac{|s+2 u(s)^2 |}{\sqrt{2} \sqrt{\Lambda_2+s}} ds \approx 0 .$$
Also we know $Y^{(6)}_{12}(0) =Y^{(6)}_{12}(x=\sqrt{2} \sqrt{\Lambda_2},t=0) \approx 0$ 
and  $Y^{(6)}_{22}(0)=Y^{(6)}_{22}(x=\sqrt{2} \sqrt{\Lambda_2},t=0) \approx 1$.
Therefore, 
\begin{eqnarray}
Y^{(6)}_{12}(t)^2+ Y^{(6)}_{22}(t)^2 \approx 1 \label{Y6-GH}
\end{eqnarray}
on the curve  $x=\sqrt{2} \sqrt{\Lambda_2+t}$, $t \in [-\Lambda_2, 0]$.
By (\ref{NinfSolution-A})-(\ref{NinfSolution}), (\ref{kappa-Ninf}), (\ref{u-expansion}) and (\ref{Y6-GH}),  we obtain
\begin{eqnarray}
\lim_{x \rightarrow \infty, t \rightarrow -\infty} \mathcal{F}(x,t)
&=& \lim_{x \rightarrow \infty, t \rightarrow -\infty} \kappa(t) \left( u(t)^{-\frac{1}{2}} \left( \frac{q_2(t)+1}{2} x-\alpha(t)\right) Y^{(6)}_{12}(x,t)
+u(t)^\frac{1}{2}Y^{(6)}_{22}(x,t) \right)   \nonumber\\
&=&\lim_{x \rightarrow \infty, t \rightarrow -\infty} \frac{1}{2} \kappa(t)   u(t)^{-\frac{1}{2}}  x  Y^{(6)}_{12}(x,t)  
\label{F-GH-Simp}.
\end{eqnarray}
By (\ref{kappa-Ninf}) and  (\ref{u-expansion}), 
if $-t \gg  \left( 18 \ln \Lambda_2 \right)^{\frac{1}{3}}$,
we have $\mathcal{F}(t) \approx 0$.

Now let us prove when $-t <2 \times \left( 18 \ln \Lambda_2 \right)^{\frac{1}{3}}$, 
$\mathcal{F}(t) \rightarrow 0$ as $t \rightarrow -\infty$.

By the mean value theorem,
\begin{eqnarray}
Y^{(6)}_{12}(t) &=& Y^{(6)}_{12}(0)+ \left( Y^{(6)}_{12}\right)'(\xi) t \nonumber\\
&=&Y^{(6)}_{12}(2 \sqrt{\Lambda_2}, 0) 
+\frac{1}{\sqrt{2} \sqrt{\Lambda_2+\xi}} \left(-\frac{1}{2} (\xi+2 u(\xi)^2) Y^{(6)}_{12}(\xi)
-u'(\xi) Y^{(6)}_{22}(\xi) \right) t , \nonumber
\end{eqnarray}
where $t<\xi<0$.

From
$$Y^{(6)}_{12}(2 \sqrt{\Lambda_2}, 0)  \approx \frac{-u(0)}{\sqrt{2} \sqrt{\Lambda_2}},
\quad \frac{\sqrt{2} \sqrt{\Lambda_2+t}}{\sqrt{2} \sqrt{\Lambda_2+\xi}}  < 1, $$
we know
\begin{eqnarray}
\left|\kappa(t) u(t)^{-\frac{1}{2}} Y^{(6)}_{12}(t) x(t) \right| <\left|\kappa(t) u(t)^{-\frac{1}{2}} \right| 
\left( |u(0)|+ \left|\frac{1}{2}(\xi+2 u(\xi)^2) Y^{(6)}_{12}(\xi)+u'(\xi) Y^{(6)}_{22}(\xi)\right| |t| \right). \nonumber
\end{eqnarray}

By $t<\xi<0$, (\ref{kappa-Ninf}), (\ref{u-expansion}) and (\ref{Y6-GH}),
we have
$$ \lim \limits_{t \rightarrow -\infty}\left|\kappa(t) u(t)^{-\frac{1}{2}} \right| 
\left( |u(0)|+ \left|\frac{1}{2}(\xi+2 u(\xi)^2) Y^{(6)}_{12}(\xi) +u'(\xi) Y^{(6)}_{22}(\xi)\right| 
|t| \right)=0 .$$
Thus
$$\kappa(t) u(t)^{-\frac{1}{2}} Y^{(6)}_{12}(t) x(t) \approx 0, $$
when $ t >-2 \times \left( 18 \ln \Lambda_2 \right)^{\frac{1}{3}}$ but large negative enough.

Altogether, we have
$$\lim_{x \rightarrow \infty, t \rightarrow -\infty} \mathcal{F}(x,t) =0.$$

\begin{appendices}
	
\section{$k_0<\frac{10}{3} $ \label{Append-k0}}
From $k_0=-\min\left( \frac{t}{u^2} \right)$,
we know 
\begin{eqnarray}
k_0=2 \times \left[ \min \left( \frac{u}{\sqrt{\frac{-t}{2}}}\right) \right]^{-2}, \quad t \in (-\infty, 0). \label{k0-HM}
\end{eqnarray}

Following the original arguments of \cite{HM},  we give a lower bound for the local minimum of $\frac{u}{\sqrt{\frac{-t}{2}}}$ 
for  large negative $t$.
\begin{proposition} \label{prop-HM}
If  there is a local minimum of $\frac{u}{\sqrt{\frac{-t}{2}}}$ for $t<-\frac{11}{8}$, 
it  must be  greater than $\sqrt{\frac{1203}{1331}} $. 
\end{proposition}
\begin{proof}
	Let $u(t)= \sqrt{\frac{-t}{2}} z(t)$. Obviously, $z(t)>0$. Then $z$ satisfies $z''(t)+\frac{z'(t)}{t}=\left( \frac{1}{4 t^2}
	-t \left(z(t)^2-1 \right)\right) z(t) $.
	At a local minimum, we have  $u'(t)=0 $ and $u''(t)>0$. 
	Then, we have $\frac{1}{4 t^2}-t \left(z(t)^2-1 \right) >0$, i.e., $z(t)> \sqrt{1+\frac{1}{4 t^3}} $.
	Since $t<-\frac{11}{8}$, the local minimum is greater than $\sqrt{\frac{1203}{1331}}$. 
\end{proof}

Note that Proposition \ref{prop-HM} does not mean $\frac{u}{\sqrt{\frac{-t}{2}}}>\sqrt{\frac{1203}{1331}}$ 
for $t \in (-\infty, -\frac{11}{8})$ since $\frac{u}{\sqrt{\frac{-t}{2}}}$ may be smaller near the boundary $t=-\frac{11}{8}$.
But if we could also prove  $\frac{u}{\sqrt{\frac{-t}{2}}}>\sqrt{\frac{1203}{1331}}$  for $t \in [-\frac{11}{8},0)$,
then we can still conclude $\frac{u}{\sqrt{\frac{-t}{2}}}>\sqrt{\frac{1203}{1331}}$ for $t \in (-\infty,0)$.
The next proposition fulfills this aim.
\begin{proposition} \label{prop-HXZ}
For $t \in [-\frac{11}{8},0)$,	$\frac{u}{\sqrt{\frac{-t}{2}}}>\sqrt{\frac{1203}{1331}}$. 
\end{proposition}
\begin{proof}
Huang et. al  \cite{HXZ} proved
\begin{eqnarray}
\left| u(0)-\frac{98}{267} \right|<11 \times 10^{-4}, \quad
\left| u'(0)+\frac{153}{518} \right|<12 \times 10^{-4}. \label{HXZh-ini}
\end{eqnarray}
They also defined  the approximate solution as
\begin{eqnarray}
y_b(t)&=&\frac{t^{15}}{13206825}+\frac{t^{14}}{717099}+\frac{t^{13}}{81755}+\frac{t^{12}}{15201}
+\frac{t^{11}}{47200}+\frac{13 t^{10}}{24088}+\frac{39 t^9}{53333} 
+\frac{18 t^8}{61523}\nonumber\\
&&-\frac{17 t^7}{20578}-\frac{93 t^6}{35396}-\frac{224 t^5}{30615}
-\frac{360 t^4}{36911}+\frac{203 t^3}{10806} +\frac{33530 t^2}{688889}-\frac{153 t}{518}+\frac{98}{267} \label{HXZh-approx}
\end{eqnarray}
and the remainder term as
\begin{eqnarray}
R_4(t)=y_b''(t)-t y_b(t)-2 y_b(t)^3 . \label{R4-DEF}
\end{eqnarray}
We can verify  \footnote{Since $R_4(t)$ is a polynomial, Sturm's theorem applies. 
	The following several cases of verification can also be done in this way.}
\begin{eqnarray}
\left| R_4(t) \right|<2 \times 10^{-3}, \quad t \in [-\frac{11}{8},0]. \label{R4-Range}
\end{eqnarray}
Let $\delta_4(t)=u(t)-y_b(t)$.
It is easy to show 
$$\delta_4''(t)=(6 y_b(t)^2+t) \delta_4(t)+6 y_b(t) \delta_4(t)^2+2 \delta_4(t)^3-R_4(t).$$
Next, we will show $\delta_4(t)$ is sufficiently small for $t \in [-\frac{11}{8},0]$.

We can verify
$ \frac{4}{5}<  6 y_b(t)^2+t < \frac{13}{5} $ and $\frac{11}{5} <  6 y_b(t) < \frac{49}{10}$ 
for $t \in [-\frac{11}{8},0]$.
Therefore, we have $\delta_4(t) \ge \delta_b(t)$ in the interval,
where $\delta_b(t)$ is defined by
\begin{eqnarray}
&&\delta_b''(t)=\frac{13}{5} \delta_b(t)-\frac{49}{10} \delta_b(t)^2+2 \delta_b(t)^3-\frac{1}{500},
\quad \delta_b(0)=-\frac{11}{10000}, \quad \delta'_b(0)=\frac{3}{2500}. \label{deta-b-DEQ}
\end{eqnarray}
So we have
\begin{eqnarray}
&&\delta_b'(t)=\sqrt{ \frac{13}{5} \delta_b(t)^2-\frac{49}{15} \delta_b(t)^3+ \delta_b(t)^4-\frac{1}{250}\delta_b(t)- \frac{183310481923}{3}\times 10^{-16} },
\quad \delta_b(0)=-\frac{11}{10000}.\label{deta-b1-DEQ}
\end{eqnarray}
By
\begin{eqnarray}
\int_{-11\times 10^{-4}}^{-120 \times 10^{-4}} 
\left( \frac{13}{5} \delta_b^2-\frac{49}{15} \delta_b^3+ \delta_b^4-\frac{1}{250}\delta_b- \frac{183310481923}{3}\times 10^{-16} \right)^{-\frac{1}{2}} d\delta_b <-\frac{11}{8},
\nonumber
\end{eqnarray}
we obtain $\delta_b(t)>-120 \times 10^{-4}=-\frac{3}{250} $ for $t \in [-\frac{11}{8},0]$.

When $t \in [-\frac{11}{8},0]$, we can show $\frac{y_b(t)-\frac{3}{250}}{ \sqrt{\frac{-t}{2}}}>\sqrt{\frac{1203}{1331}}$ .
Therefore, for $t \in [-\frac{11}{8},0)$, we have
\begin{eqnarray}
\min \left(\frac{u(t)}{\sqrt{\frac{-t}{2}}}\right) =\min \left(\frac{y_b(t)+\delta_4(t)}{\sqrt{\frac{-t}{2}}} \right)
>\min \left(\frac{y_b(t)+\delta_b(t)}{\sqrt{\frac{-t}{2}}}\right) >\min \left( \frac{y_b(t)-\frac{3}{250}}{ \sqrt{\frac{-t}{2}}}\right)>\sqrt{\frac{1203}{1331}}.\nonumber
\end{eqnarray}

\end{proof}

Combining Proposition \ref{prop-HM} and \ref{prop-HXZ}, 
we obtain $ \min \left( \frac{u(t)}{\sqrt{\frac{-t}{2}}} \right) >\sqrt{\frac{1203}{1331}}$.
By (\ref{k0-HM}), $k_0<\frac{2662}{1203}$ is obtained.

\section{The growth rate estimate for the solution of a second order linear ODE \label{Append-estimation}}
For $a>0$, $c>0$, $b^2-4 a c<0$ and $V=(x,y)^T$, let us define
$$|V|=\sqrt{ a x^2+b x y + c y^2}, $$
where $|V|$ can be understood as the length of vector $V$.
Denote 
$$M=\left( \begin{array}{cc}m_{11}&m_{12}\\m_{21}&m_{22}\end{array} \right),$$
then
$$r^2=\frac{|M V|^2}{|V|^2} $$
has both a minimal value $\mathcal{M}_1=r_1^2$ and a maximal value $\mathcal{M}_2=r_2^2$.
These two extreme values satisfy
\begin{eqnarray}
\Delta \mathcal{M}^2+ 2 \delta \mathcal{M}+\Delta D^2=0, \label{Eq-ratio} 
\end{eqnarray}
where
\begin{eqnarray*}
&&\Delta=b^2-4 a c, \\
&&D=\det{M},\\
&&\delta=-b^2(m_{11} m_{22}+m_{12} m_{21})-2 b (a m_{12}-c m_{21}) (m_{11}-m_{22})
+ 2 a^2 m_{12}^2+2 c^2 m_{21}^2 +2 a c (m_{11}^2+m_{22}^2). 
\end{eqnarray*}
Assume the linear ODE is of form
\begin{eqnarray}
\frac{d}{dt}  \left(\begin{array}{c}x(t)\\y(t)\end{array}\right)
= \left( \begin{array}{cc} f_{11}(t) &f_{12}(t)\\ f_{21}(t) &f_{22}(t)\end{array}\right) 
\left(\begin{array}{c}x(t)\\y(t)\end{array}\right). \label{LODE-f}
\end{eqnarray}
We know
\begin{eqnarray}
&&\left( \begin{array}{c} x(t_1)\\y(t_1) \end{array} \right)
=\lim_{N \rightarrow \infty}
\left( \begin{array}{cc}1+h f_{11}(t_0+(N-1)h) & h f_{12}(t_0+(N-1)h)\\ h f_{21}(t_0+(N-1)h) &1+ h f_{22}(t_0+(N-1)h)\end{array}\right)
\cdots \nonumber\\
&&\cdots \left( \begin{array}{cc}1+h f_{11}(t_0+h) & h f_{12}(t_0+h)\\ h f_{21}(t_0+h) &1+ h f_{22}(t_0+h)\end{array}\right)
\left( \begin{array}{cc}1+h f_{11}(t_0) & h f_{12}(t_0)\\ h f_{21}(t_0) &1+ h f_{22}(t_0)\end{array}\right)
\left( \begin{array}{c} x(t_0)\\y(t_0) \end{array} \right),
\nonumber
\end{eqnarray}
where $h=\frac{t_1-t_0}{N}$.
Then by (\ref{Eq-ratio}), for $t_1>t_0$, we can prove 
\begin{eqnarray}
&&\ln |(x(t_1),y(t_1))^T| \le \ln  |(x(t_0),y(t_0))^T|
+ \int_{t_0}^{t_1} \left( \frac{1}{2}(f_{11}(s)+f_{22}(s))+\frac{\sqrt{H(s)}}{\sqrt{-\Delta}} \right)  ds, \label{LODE-f-NEQ1}\\
&&\ln |(x(t_1),y(t_1))^T| \ge \ln  |(x(t_0),y(t_0))^T|
+ \int_{t_0}^{t_1} \left( \frac{1}{2}(f_{11}(s)+f_{22}(s))-\frac{\sqrt{H(s)}}{\sqrt{-\Delta}} \right)  ds, \label{LODE-f-NEQ2}
\end{eqnarray}
where
$$H(s)=a c (f_{11}(s)-f_{22}(s))^2-b (a f_{12}(s)-c f_{21}(s)) (f_{11}(s)-f_{22}(s)) 
+(a f_{12}(s)+c f_{21}(s))^2-b^2 f_{12}(s) f_{21}(s) .$$

If $a=c=1$, $b=0$, for $t_1>t_0$, (\ref{LODE-f-NEQ1}) and (\ref{LODE-f-NEQ2}) are reduced to
\begin{eqnarray}
\ln |(x(t_1),y(t_1))^T| &\le& \frac{1}{2} \int_{t_0}^{t_1} \left( f_{11}(s)+f_{22}(s)+\sqrt{(f_{11}(s)-f_{22}(s))^2+(f_{12}(s)+f_{21}(s))^2} \right)  ds \nonumber\\
&&+\ln  |(x(t_0),y(t_0))^T|, \label{LODE-f-NEQ-S1}\\
\ln |(x(t_1),y(t_1))^T| &\ge& \frac{1}{2} \int_{t_0}^{t_1} \left(  f_{11}(s)+f_{22}(s)-\sqrt{(f_{11}(s)-f_{22}(s))^2+(f_{12}(s)+f_{21}(s))^2} \right)  ds \nonumber\\
&&+\ln  |(x(t_0),y(t_0))^T| . \label{LODE-f-NEQ-S2}
\end{eqnarray}

\section{$\mathcal{F}(x,t)$ on  $\wideparen{CD}$ \label{Appd-CD}}
\begin{eqnarray}
\mathcal{F}(x,t)=-\kappa(t) u(t)^\frac{1}{2} e^{\frac{x^3}{3}-x t}
\left[u(t)^{-1} \left(\frac{1+q_2(t)}{2}x-\alpha(t) \right)Y_{11}^{(3)}(x,t)+Y_{21}^{(3)}(x,t) \right]. \nonumber
\end{eqnarray}

We divide the region $\wideparen{CD}$ into two parts: $x \le -\sqrt{3} \sqrt{t} $ and $x \ge -\sqrt{3} \sqrt{t}$.

\subsection{Case $x \le -\sqrt{3} \sqrt{t} $.}
In this case,  $Y_{21}^{(3)}(x,t) \rightarrow 0$ and $e^{\frac{x^3}{3}-x t} \le 1$,
so 
$$\lim_{x \le -\sqrt{3} \sqrt{t}, t \rightarrow \infty} -\kappa(t) u(t)^\frac{1}{2} e^{\frac{x^3}{3}-x t}
 Y_{21}^{(3)}(x,t)  =0.$$

By  $Y_{11}^{(3)}(x,t) \rightarrow 1$,
we know 
$$\lim_{x \le -\sqrt{3} \sqrt{t}, t \rightarrow \infty} \kappa(t) u(t)^\frac{1}{2} e^{\frac{x^3}{3}-x t}
u(t)^{-1} \alpha(t) Y_{11}^{(3)}(x,t)  =0.$$

Thus
\begin{eqnarray}
 \lim_{x \le -\sqrt{3} \sqrt{t}, t \rightarrow \infty} \mathcal{F}(x,t)
&=&\lim_{x \le -\sqrt{3} \sqrt{t}, t \rightarrow \infty} 
-\kappa(t) u(t)^\frac{1}{2} e^{\frac{x^3}{3}-x t} u(t)^{-1} \frac{q_2(t)+1}{2} x Y_{11}^{(3)}(x,t) \nonumber\\
&=&\lim_{x \le -\sqrt{3} \sqrt{t}, t \rightarrow \infty}x e^{\frac{x^3}{3}-x t} u(t)^{-1} \frac{q_2(t)+1}{2}.
\label{CD-sqrt3-less}
\end{eqnarray}
(\ref{CD-sqrt3-less}) can be proved to be $0$ 
by diving the region $x \le -\sqrt{3} \sqrt{t}$ into two parts, 
for example $ x \le -2 \sqrt{t}$ and $-2 \sqrt{t} \le x \le -\sqrt{3} \sqrt{t}$.
In both parts, 
$x e^{\frac{x^3}{3}-x t} u(t)^{-1} \frac{q_2(t)+1}{2} \rightarrow 0$ 
is obvious.

Therefore,  we have
$$\lim_{x \le -\sqrt{3} \sqrt{t}, t \rightarrow \infty} \mathcal{F}(x,t) =0.$$

\subsection{Case $x \ge -\sqrt{3} \sqrt{t}$ .}
First we   show
\begin{eqnarray}
\lim_{x \rightarrow -\infty, t\rightarrow \infty, x \ge -\sqrt{3} \sqrt{t}}
-\kappa(t) u(t)^\frac{1}{2} e^{\frac{x^3}{3}-x t} u(t)^{-1} \left(\frac{1+q_2(t)}{2}x-\alpha(t) \right)
 Y_{11}^{(3)}(x,t)=0. \label{F-CD-Y11}
\end{eqnarray}
In fact,
by (\ref{q2Pinf00}) and (\ref{alphaPinf00}), we obtain their expansions at $t=\infty$ as
\begin{eqnarray}
&&q_2(t)+1=e^{-\frac{4}{3}t^{3/2}} \left(\frac{1}{8 \pi} t^{-\frac{3}{2}}-\frac{59}{192 \pi}t^{-3} +\cdots \right),\label{q2-pinf-expansion} \\
&&\alpha(t)=e^{-\frac{4}{3}t^{3/2}} \left(\frac{3}{16 \pi} t^{-1}-\frac{29}{128 \pi}t^{-\frac{5}{2}} +\cdots \right) . \label{alpha-pinf-expansion} 
\end{eqnarray}
Also we have
$$u(t)^{-1} \approx \frac{1}{\mathrm{Ai}(t)}=e^{\frac{2}{3} t^{3/2}} \left( 2 \sqrt{\pi} t^{\frac{1}{4}}+
\frac{5}{24} \sqrt{\pi} t^{-\frac{5}{4}}+ \cdots \right). $$

Therefore,
\begin{eqnarray}
&&\lim_{x \rightarrow -\infty, t\rightarrow \infty, x \ge -\sqrt{3} \sqrt{t}}
-\kappa(t) u(t)^\frac{1}{2} e^{\frac{x^3}{3}-x t} u(t)^{-1} \left(\frac{1+q_2(t)}{2}x-\alpha(t) \right)
Y_{11}^{(3)}(x,t) \nonumber\\
=&&\lim_{x \rightarrow -\infty, t\rightarrow \infty, x \ge -\sqrt{3} \sqrt{t}}
-e^{\frac{x^3}{3}-x t} u(t)^{-1} \left(\frac{1+q_2(t)}{2}x-\alpha(t) \right) \nonumber \\
=&&\lim_{x \rightarrow -\infty, t\rightarrow \infty, x \ge -\sqrt{3} \sqrt{t}}
-e^{\frac{x^3}{3}-x t-\frac{2}{3}t^{3/2}} \left( \frac{1}{4 \sqrt{\pi}} t^{-\frac{5}{4}} x
-\frac{3}{8 \sqrt{t}} t^{-\frac{3}{4}}\right). \nonumber
\end{eqnarray}
But $\frac{x^3}{3}-x t-\frac{2}{3}t^{\frac{3}{2}} \le 0$ in this case
and also $|x|\le \sqrt{3} t^{\frac{1}{2}}$. So we get (\ref{F-CD-Y11}).

By (\ref{Y3-21EQ-apprx-SOL}), we have
\begin{eqnarray}
\lim_{x=k \sqrt{t}, k<0, k \neq -1, t\rightarrow \infty}
-\kappa(t) u(t)^\frac{1}{2} e^{\frac{x^3}{3}-x t} Y_{21}^{(3)}(x,t)=-\Upsilon_4(k). \label{F-CD-Y21}
\end{eqnarray}

Near $x=-\sqrt{t}$, by (\ref{Y3-21EQ-apprx-SOL-X}), we get
\begin{eqnarray}
&&-\kappa(t) u(t)^\frac{1}{2} e^{\frac{x^3}{3}-x t} Y_{21}^{(3)}(x,t) \nonumber\\
&&\approx \frac{1}{2} 
+\frac{1}{2} \int_{\infty}^t e^{\frac{2}{3} s^{3/2}} \left(\mathrm{Ai}(s)+\frac{1}{\sqrt{s}} \mathrm{Ai}'(s) \right)ds
-\int_{-\sqrt{t}}^x e^{-t s+\frac{1}{3} s^3} \left( s \mathrm{Ai}(t)+\mathrm{Ai}'(t)\right) ds \label{F-CD-Y21-A}\\
&&\approx \frac{1}{2} -\int_0^{\frac{x}{\sqrt{t}}+1} e^{\frac{2}{3} t^{3/2}-t^{3/2} r^2+\frac{1}{3}t^{3/2} r^3}
\left(\mathrm{Ai}(t) \sqrt{t} (r-1)+\mathrm{Ai}'(t)\right) \sqrt{t} dr. \label{F-CD-Y21-B}
\end{eqnarray}
By (\ref{F-CD-Y21-A}), we see $-\kappa(t) u(t)^\frac{1}{2} e^{\frac{x^3}{3}-x t} Y_{21}^{(3)}(x,t)$
is monotonic increasing with $x$. So it must lie in $(0,1)$.
(\ref{F-CD-Y21-B}) is convenient for estimating its value.

\subsubsection{Case $k\sqrt{t}<x \le 0$, $k>-1$.}
(\ref{Y3-12EQ-apprx-X}) is valid for a large  positive $t$.
So does (\ref{Y3-12EQ-apprx-SOL-X}).
Along the line that $t$ is fixed, by (\ref{Y3-12EQ-apprx-SOL-X}) we get
\begin{eqnarray}
\frac{d}{dx} \left(- e^{\frac{x^3}{3}-x t} Y_{21}^{(3)}(x,t) \right) =
-e^{-t x+\frac{1}{3}x^3}  \left(x \mathrm{Ai}(t)+\mathrm{Ai}'(t) \right).  \label{F-Diff-X}
\end{eqnarray}
As
$$\frac{\mathrm{Ai}'(t)}{\mathrm{Ai}(t)}=-\sqrt{t}-\frac{1}{4t}+\frac{t}{32} t^{-\frac{5}{2}}+\cdots,$$
(\ref{F-Diff-X}) never vanishes in the region.
Thus $- e^{\frac{x^3}{3}-x t} Y_{21}^{(3)}(x,t)$ is monotonic increasing in the region.
But it is known 
$$ - e^{\frac{x^3}{3}-x t} Y_{21}^{(3)}(x,t)|_{x=k \sqrt{t}} \approx 1, \quad k>-1, $$
and 
$$ - e^{\frac{x^3}{3}-x t} Y_{21}^{(3)}(x,t)|_{x=0} \approx \mathcal{F}(x,t)|_{x=0} \approx 1 .$$
We must conclude
$$\lim_{k \sqrt{t}< x \le 0, k>-1, t \rightarrow \infty} \mathcal{F}(x,t) =1.$$

\end{appendices}

{\bf \noindent Acknowledgement.}
Part of this work was done while Y. Li was visiting  the Department of Mathematical Sciences of IUPUI.
Y. Li  would like to thank  A. Its for his hospitality, encouragement and suggestions.
The work is partly supported by NSFC(11375090, 11675054, 11435005) and Shanghai Collaborative Innovation Center of Trustworthy
Software for Internet of Things (ZF1213).

\noindent
YUQI Li\\
Institute of Computer Theory, School of Computer Science and Software Engineering\\
East China Normal University\\
Shanghai, 200062\\
China\\
E-mail: yqli@sei.ecnu.edu.cn

\end{document}